\pdfoutput=1
\newif\ifshort\shortfalse
\ifshort
\documentclass[acmsmall,screen]{acmart}[2021/12/05]
\else
\documentclass[acmsmall,screen,nonacm,pdfusetitle]{acmart}[2021/12/05]
\fi

\usepackage{semantic}
\usepackage{mathtools}
\usepackage{stmaryrd}
\usepackage{quantikz}
\usepackage{color}
\usepackage{tensor}
\usepackage{pifont}
\usepackage{cleveref}
\usepackage{graphicx}

\ifshort
\setcopyright{rightsretained}
\acmDOI{10.1145/3571225}
\acmYear{2023}
\acmSubmissionID{popl23main-p197-p}
\acmJournal{PACMPL}
\acmVolume{7}
\acmNumber{POPL}
\acmArticle{32}
\acmMonth{1}
\received{2022-07-07}
\received[accepted]{2022-11-07}
\else
\setcopyright{cc}
\acmConference[POPL 2023]{50th ACM SIGPLAN Symposium on Principles of Programming Languages}{Jan 18--20, 2023}{Boston, Massachusetts}
\fi
\acmPrice{}
\copyrightyear{2023}
\setcctype{by}
\DeclareMathOperator{\iso}{iso}
\DeclareMathOperator{\classical}{classical}

\DeclareMathOperator{\inj}{inj}
\DeclareMathOperator{\preps}{preps}
\DeclareMathOperator{\flags}{flags}

\DeclareMathOperator{\FV}{FV}

\DeclareMathOperator{\size}{size}
\DeclareMathOperator{\encode}{enc}
\DeclareMathOperator{\enc}{enc}
\DeclareMathOperator{\dom}{dom}

\DeclareMathOperator{\tr}{tr}
\DeclareMathOperator{\depth}{depth}
\DeclarePairedDelimiter{\norm}{\|}{\|}

\DeclarePairedDelimiter{\psem}{\llparenthesis}{\rrparenthesis}
\DeclarePairedDelimiter{\msem}{\llbracket}{\rrbracket}

\newenvironment{DIFnomarkup}{}{}

\newcommand*{\defeq}{\mathrel{\vcenter{\baselineskip0.5ex \lineskiplimit0pt
                     \hbox{\scriptsize.}\hbox{\scriptsize.}}}%
                     =}
\newcommand*{\defeqq}{\mathrel{\vcenter{\baselineskip0.5ex \lineskiplimit0pt
                     \hbox{\scriptsize.}\hbox{\scriptsize.}}}%
                     \mathrel{\vcenter{\baselineskip0.5ex \lineskiplimit0pt
                     \hbox{\scriptsize.}\hbox{\scriptsize.}}}%
                     =}
\newcommand{\spanning}[2]{\textup{spanning}_{\color{gray} #1}\left( \begin{aligned} #2 \end{aligned} \right)}
\newcommand{\ortho}[2]{\textup{ortho}_{\color{gray} #1}\left( \begin{aligned} #2 \end{aligned} \right)}
\newcommand{\erases}[1]{\textup{erases}_{\color{gray} #1}}
\newcommand{\MLPi}{\textup{ML}_{\Pi}}

\newcommand{\grover}[1]{\texttt{grover}_{\color{gray}#1}}

\newcommand{\Void}{\texttt{Void}}
\newcommand{\Unit}{\texttt{()}}
\newcommand{\Child}{\texttt{Child}}
\newcommand{\Vertex}{\texttt{Vertex}}

\newcommand{\Maybe}[1]{\texttt{Maybe}\; #1}

\newcommand{\nothing}[1]{\texttt{nothing}_{\color{gray}#1}}
\newcommand{\equals}[1]{\texttt{equals}_{\color{gray}#1}}
\newcommand{\validate}[1]{\texttt{validate}_{\color{gray}#1}}
\newcommand{\just}[1]{\texttt{just}_{\color{gray}#1}}

\newcommand{\Bit}{\texttt{Bit}}
\newcommand{\lef}[2]{\texttt{left}_{\color{gray}{#1 \oplus #2}}}
\newcommand{\rit}[2]{\texttt{right}_{\color{gray}{#1 \oplus #2}}}
\newcommand{\tagsum}[2]{\texttt{tag}_{\color{gray}{#1 \oplus #2}}}
\newcommand{\pair}[2]{\texttt{(} #1 \texttt{,} #2 \texttt{)}}
\newcommand{\nprep}[1][]{{n_{\textsc{prep#1}}}}
\newcommand{\nflag}[1][]{{n_{\textsc{flag#1}}}}
\newcommand{\ngarb}[1][]{{n_{\textsc{garb#1}}}}

\newcommand{\consarrow}[1]{\hookrightarrow_{\color{gray}{#1}}}

\newcommand{\zero}{\texttt{0}}
\newcommand{\one}{\texttt{1}}

\newcommand{\cdown}{\texttt{cdown}}
\newcommand{\cleft}{\texttt{cleft}}
\newcommand{\cright}{\texttt{cright}}
\newcommand{\vleft}{\texttt{vleft}}
\newcommand{\vright}{\texttt{vright}}
\newcommand{\Coin}{\texttt{Coin}}
\newcommand{\Leaf}{\texttt{Leaf}}

\newcommand{\rooot}[1]{\texttt{root}_{\color{gray}{#1}}}
\newcommand{\rootone}[1]{\texttt{root\textquotesingle}_{\color{gray}{#1}}}

\newcommand{\downcast}[1]{\texttt{downcast}_{\color{gray}{#1}}}
\newcommand{\leftchild}[1]{\texttt{leftchild}_{\color{gray}{#1}}}
\newcommand{\rightchild}[1]{\texttt{rightchild}_{\color{gray}{#1}}}
\newcommand{\nextcoin}[1]{\texttt{nextcoin}_{\color{gray}{#1}}}

\newcommand{\asleaf}[1]{\texttt{asleaf}_{\color{gray}{#1}}}

\newcommand{\diffusion}[1]{\texttt{diffusion}_{\color{gray}{#1}}}
\newcommand{\walk}[1]{\texttt{walk}_{\color{gray}{#1}}}

\newcommand{\gphase}[2]{\texttt {gphase}_{\color{gray}{#1}}\texttt{(} #2  \texttt{)}}
\newcommand{\rphase}[4]{\texttt {rphase}_{\color{gray}{#1}} \left\{ \begin{aligned} #2 &\mapsto #3 \\ \texttt{else} &\mapsto #4 \end{aligned} \right\}}

\newcommand{\had}{\texttt{had}}

\newcommand{\fst}[2]{\texttt{fst}_{\color{gray}{#1 \otimes #2}}}

\newcommand{\snd}[2]{\texttt{snd}_{\color{gray}{#1 \otimes #2}}}

\newcommand{\reflect}[1]{\texttt{reflect}_{\color{gray}#1}}

\newcommand{\trycatch}[2]{\texttt{try}\; #1 \; \texttt{catch}\; #2}

\newcommand{\cntrl}[4]{\texttt{ctrl } #1 \tensor*[_{\color{gray}{#2}}]{\left\{ \begin{aligned} #3 \end{aligned} \right\}}{_{\color{gray}{#4}}}}

\newcommand{\match}[3]{\texttt{match} \tensor*[_{\color{gray}{#1}}]{\left\{ \begin{aligned} #2 \end{aligned} \right\}}{_{\color{gray}{#3}}}}
\newcommand{\meas}[1]{\texttt{meas}_{\color{gray}{#1}}}

\newcommand{\divby}{\texttt{/}}
\newcommand{\complex}{\mathbb C}

\newcommand{\X}{\mathbb X}

\newcommand{\V}{\mathbb V}

\newcommand{\E}{\mathcal E}
\newcommand{\linear}{\mathcal L}

\newcommand{\Hilb}{\mathcal{H}}
\newcommand{\unit}{\texttt{()}}

\newcommand{\uthree}[3]{\texttt u_{\texttt 3} \texttt{(} #1 \texttt{,} #2 \texttt{,} #3 \texttt{)} }

\newcommand{\qiso}[2]{\texttt{q}_{\color{gray} #1 \rightsquigarrow #2 }}
\newcommand{\qarr}[2]{\texttt{q}_{\color{gray}#1 \Rrightarrow #2}}

\newcommand\doubleplus{\mathbin{+\!+}}
\newcommand\partition{\mathbin{\|}}
\newcommand{\cmark}{\ding{51}}
\newcommand{\xmark}{\ding{55}}
\newcommand{\subcap}[1]{_{\textsc{#1}}}
\newcommand{\cptp}{\textsc{cptp}}
\newcommand{\Ccptp}{\mathcal{C}\subcap{cptp}}

\newcommand{\projv}[1]{P_{\V(#1)}^{\text{\cmark}}}
\newcommand{\projvx}[1]{P_{\V(#1)}^{\text{\xmark}}}

\newcommand{\op}[2]{| #1 \rangle\!\langle #2 |}
\newcommand{\ip}[2]{\langle #1 | #2 \rangle}

\graphicspath{ {./figures/} }

\begin{document}

\title{Qunity}
\ifshort
\subtitle{A Unified Language for Quantum and Classical Computing}
\else
\subtitle{A Unified Language for Quantum and Classical Computing (Extended Version)}
\fi
\begin{abstract}
We introduce Qunity, a new quantum programming language designed to treat quantum computing as a natural generalization of classical computing.
Qunity presents a unified syntax where familiar programming constructs can have both quantum and classical effects.
For example, one can use sum types to implement the direct sum of linear operators, exception-handling syntax to implement projective measurements, and aliasing to induce entanglement.
Further, Qunity takes advantage of the overlooked \textsc{bqp} subroutine theorem, allowing one to construct reversible subroutines from irreversible quantum algorithms through the uncomputation of ``garbage'' outputs.
Unlike existing languages that enable quantum aspects with separate add-ons (like a classical language with quantum gates bolted on), Qunity provides a unified syntax and a novel denotational semantics that guarantees that programs are quantum mechanically valid.
We present Qunity's syntax, type system, and denotational semantics, showing how it can cleanly express several quantum algorithms.
We also detail how Qunity can be compiled into a low-level qubit circuit language like Open\textsc{Qasm}, proving the realizability of our design.
\end{abstract}

\author{Finn Voichick}
\orcid{0000-0002-1913-4178}
\affiliation{
	\institution{University of Maryland}
	\city{College Park}
	\country{USA}
}
\email{finn@umd.edu}

\author{Liyi Li}
\orcid{0000-0001-8184-0244}
\affiliation{
	\institution{University of Maryland}
	\city{College Park}
	\country{USA}
}
\email{liyili2@umd.edu}

\author{Robert Rand}
\orcid{0000-0001-6842-5505}
\affiliation{
	\institution{University of Chicago}
	\city{Chicago}
	\country{USA}
}
\email{rand@uchicago.edu}

\author{Michael Hicks}
\orcid{0000-0002-2759-9223}
\authornote{Work completed before starting at Amazon.}
\affiliation{
	\institution{University of Maryland and Amazon}
	\city{College Park}
	\country{USA}
}
\email{mwh@cs.umd.edu}

\ccsdesc[500]{Software and its engineering~Multiparadigm languages}
\ccsdesc[300]{Theory of computation~Denotational semantics}
\ccsdesc[300]{Theory of computation~Quantum information theory}
\ccsdesc[300]{Software and its engineering~Syntax}
\ccsdesc[300]{Software and its engineering~Control structures}
\ccsdesc[100]{Software and its engineering~Data types and structures}
\ccsdesc[100]{Software and its engineering~Compilers}
\ccsdesc[100]{Software and its engineering~Functional languages}
\ccsdesc[100]{Software and its engineering~Patterns}
\ccsdesc[100]{Software and its engineering~Semantics}
\ccsdesc[100]{Theory of computation~Control primitives}
\ccsdesc[100]{Theory of computation~Categorical semantics}
\ccsdesc[100]{Theory of computation~Quantum query complexity}
\ccsdesc[100]{Theory of computation~Assertions}

\keywords{algebraic data types, Kraus operators, quantum subroutines, reversible computing}

\maketitle

\section{Introduction}
\label{sec:intro}

		Quantum computing \emph{generalizes} classical computing.
		That is, any efficiently-implementable classical algorithm can also be efficiently implemented on a quantum computer.
                However, quantum programming languages today do not fully leverage this connection.
                Rather, to varying degrees, they impose a separation between their classical and quantum programming constructs.
                Such a separation owes in part to the \textsc{qram} computing model \cite{qram-pseudocode}, and is reflected in slogans such as ``quantum data, classical control'' \cite{quantum-data-classical-control}.
                While keeping the quantum and classical separated in the language makes some sense in the near term (the ``\textsc{nisq}'' era \cite{nisq}), it artificially limits the long-term potential of quantum algorithm designs.
                It also fails to take advantage of a classical programmer's intuition and limits the reuse of classical code and ideas in quantum algorithms.

                In this paper, we present \emph{Qunity} (``\textsc{kyoo}-nih-tee''), a new programming language whose programming constructs will be familiar to classical programmers but are generalized to include quantum behavior.
                Thus, Qunity aims to \emph{unify} quantum and classical concepts in a single language.
                Qunity draws inspiration from prior languages which contain some unified constructs~\cite{qml,silq}, but Qunity broadens and deepens that unification.

\subsection{Motivating Example: Deutsch's Algorithm}
\label{sec:deutsch}

To give a sense of Qunity's design, we present Deutsch's algorithm \cite{deutsch}.
Given black-box access to a function $f : \{\zero, \one\} \to \{\zero, \one\}$, this algorithm computes whether or not $f(\zero) \stackrel{?}{=} f(\one)$ using only a single query to the function. 

\begin{align*}
	&\texttt{deutsch}(f) \defeq \\
	&\texttt{let } x \texttt{ =}_{\color{gray}{\Bit}}\; {(\had\; \zero)} \texttt{ in} \\
	&\left(\cntrl{(f\; x)}{\color{gray}\Bit}{\zero &\mapsto x \\ \one &\mapsto x \triangleright \gphase{\Bit}{\pi}}{\Bit}\right) \triangleright \had
\end{align*}

The algorithm has three steps:
\begin{enumerate}
\item Apply a Hadamard operator \texttt{had} to a qubit in the zero state, yielding a qubit in state $\ket{+}$; the \texttt{let} expression binds this qubit to $x$.
\item Query an oracle to conditionally flip the phase of the qubit, coherently performing the linear map:
		\[
			\ket x \mapsto (-1)^{f(x)} \ket{x}
			\]
This step is implemented by the \texttt{ctrl} expression: if $f x$ is $\zero$ then $\ket{x}$ is unchanged, but if $f x$ is $\one$ then $x \triangleright \gphase{\Bit}{\pi}$ applies a phase of $e^{i\pi}$ (which is $-1$) to $\ket{x}$.
\item Finally, apply a Hadamard operator to the qubit output from step 2.
\end{enumerate}
If $f(\zero) \neq f(\one)$, the output will be $\ket \one$ up to global phase; otherwise, it will be $\ket \zero$.

Qunity's version of Deutsch's algorithm reads like a typical functional program and is more general than typically presented versions.
Some presentations \cite{mike-and-ike} require constructing a two-qubit unitary oracle $U_f$ from the classical function $f$ such that $U_f \ket{x,y} = \ket{x, y \oplus f(x)}$ for all $x, y \in \{\zero,\one\}$.
In Qunity, no separate oracle is needed; the ``oracle'' is $f$ itself. 
Existing programming languages \cite{reverc,reqwire,vqo} support automatic construction of oracles $U_f$ from explicitly-written programs $f$, but these require that the implementation of $f$ is strictly classical.
For example, Silq's \cite{silq} uncomputation construct requires that the program of interest be ``qfree,'' with strictly classical behavior, and Quipper \cite{quipper} supports automatic oracle construction only from classical Haskell programs.
In Qunity, $f$ can be an \emph{arbitrary} quantum algorithm that takes a qubit as input and produces a qubit as output, even if it uses measurement or interspersed classical operations.

Qunity takes advantage of the \textsc{bqp} subroutine theorem \cite{bqp,bqp-watrous}, which allows for the construction of reversible subroutines from arbitrary (not necessarily reversible) quantum algorithms.
If the quantum algorithm has probabilistic behavior, the reversible subroutine may have a degree of error, but the \textsc{bqp} subroutine theorem places reasonable bounds on this error.
These bounded-error subroutines are commonly used in the design of quantum algorithms \cite{kothari}, but existing programming languages provide no convenient way to construct and compose them, a gap that Qunity fills.

 According to Qunity's type system (Section~\ref{sec:types}), the following typing judgment is valid:
 \[
	\inference{\vdash f : \Bit \Rrightarrow \Bit}{{\color{gray} \varnothing \partition \varnothing} \vdash \texttt{deutsch}(f) : \Bit}
\]

\sloppy The rule says: ``given an arbitrary quantum algorithm for computing a function $f : \{\zero,\one\} \to \{\zero,\one\}$, our $\texttt{deutsch}(f)$ program outputs a bit.''
Per Qunity's formal semantics
(Section~\ref{sec:semantics}),
$\texttt{deutsch}(f)$ corresponds to a single-qubit pure state whenever $f$ corresponds to a single-qubit quantum channel.
This is possible because of the unique way that Qunity's semantics interweaves the usage of pure and mixed quantum states.

\subsection{Design Principles}

Qunity's design is motivated by four key principles: generalization of classical constructs, expressiveness, compositionality, and realizability.

\paragraph{Generalization of classical constructs}
To make quantum computing easier for programmers, Qunity allows them to draw on intuition from classical computing.
Many elements of Qunity's syntax are simply quantum generalizations of classical constructs: for example, tensor products generalize pairs, projective measurements generalize \texttt{try}-\texttt{catch}, and quantum control generalizes pattern matching.
Rather than use linear types (as in \textsc{Qwire} \cite{qwire} and the Proto-Quipper languages~\cite{protoquipper,protoquipperD}), Qunity allows variables to be freely duplicated and discarded as in classical languages, but its semantics treats variable duplication as an entangling operation, and variable discarding as a partial trace.

\paragraph{Expressiveness}
Qunity allows for writing algorithms at a higher level of abstraction than existing languages.
One way that it does this is through algebraic data types: rather than manipulate fixed-length arrays of qubits directly, programmers can work with more complicated types.
For example, in our quantum walk algorithm
given in Section~\ref{sec:quantum-walk},
we deal with superpositions of variable-length lists.
Qunity's semantics also allows one to ``implement'' mathematical objects that are frequently used in algorithm analysis but seldom used in algorithm implementation.
For example, the semantics of a Qunity program can be a superoperator \cite[p.~57]{klm}, an isometry \cite[p.~210]{advanced-linear-algebra}, or a projector \cite[p.~70]{mike-and-ike}, and these can be composed in useful ways.

		In detail, the operators that make up Qunity's semantics are drawn from a broad class called \emph{Kraus operators} \cite[p.~60]{klm}, which include norm-decreasing operators such as projectors.
		Projectors are used in quantum algorithms, but more often for \textit{analyzing} quantum algorithms, and few quantum languages allow projectors to be directly implemented.
		Motivated by the fact that operators produced by the \textsc{bqp} subroutine theorem can be viewed as norm-\textit{decreasing} rather than norm-preserving, we give Qunity programs norm non-increasing semantics.
		Expanding the language to include projectors turns out to be quite useful:
		We can treat the null space of a projector as a sort of ``exception space,'' allowing us to reason about quantum projectors by a syntactic analogy with classical programming strategies, namely exception handling.
		This allows us to implement a few more useful program transformations, like a ``reflection'' in the style of Grover's diffusion operator. Given the ability to implement a projector $P$, it is straightforward to implement the unitary reflection $(2P - I)$, a common feature in quantum algorithms.

\paragraph{Compositionality}
Qunity programs can be composed in useful ways.
For example, our \texttt{ctrl} construct uses the \textsc{bqp} subroutine theorem \cite{bqp,bqp-watrous} to reversibly use an irreversible quantum algorithm as a condition for executing another---this happens in step 2 of our version of Deutsch's algorithm.
In general, Qunity's denotational semantics allows composing programs that are mathematically described in different ways, like defining an expression's ``pure'' operator semantics in terms of a subexpression's ``mixed'' superoperator semantics.
By contrast, prior languages have not managed such a compositional semantics~\cite{mingsheng-ying,qml-update,qml-thesis}, as discussed below. %

\paragraph{Realizability}
		We have designed Qunity so that it can be compiled into qubit-based unitary circuits written in a lower-level language such as Open\textsc{Qasm}.
		We have designed such a compilation procedure and proven that the circuits it produces correctly implement Qunity's denotational semantics.
		In Section~\ref{sec:compilation}, we give an overview of our compilation strategy, and we include the full details and proofs in
\ifshort
the supplemental report \cite{supplement}.
\else
Appendix~\ref{app:compilation}.
\fi
		To limit our analysis to finite-dimensional Hilbert spaces, we limit Qunity programs to working with finite types only.
		Some quantum languages \cite[Chapter 7]{mingsheng-ying} work with infinite-dimensional Hilbert spaces such as Fock spaces, making it possible to coherently manipulate superpositions of \emph{unbounded} lists, for example.
		However, since qubit-based circuits work only with finite-dimensional Hilbert spaces, allowing Qunity to work with infinite data types would mean that compilation would have to approximate infinite-dimensional Hilbert spaces with finite ones, which is challenging to do satisfyingly.
		For this reason, Qunity has no notion of recursive types, though we use classical metaprogramming to define parameterized types. %

\subsection{Related Work}               
                
\begin{table}[th]
	\begin{minipage}{\columnwidth}
		\begin{center}
\begin{tabular}{@{}lllllll@{}} \toprule
& \multicolumn{6}{c}{Language} \\ \cmidrule(r){2-7}
	Feature & Qunity & Silq & \textsc{Qml} & \textsc{Spm} & Quipper & Qu\textsc{Gcl} \\ \midrule
Decoherence & \cmark & \cmark & \cmark & \xmark & \cmark & \cmark \\
Denotational Semantics & \cmark & \xmark & \cmark* & \xmark & \xmark & \cmark* \\
Quantum Sum Types & \cmark & \xmark & \cmark & \cmark & \xmark & \xmark \\
Isometries & \cmark & \cmark & \cmark* & \xmark & \cmark* & \xmark \\
Projectors & \cmark & \xmark & \xmark & \xmark & \xmark & \cmark \\
Classical Uncomputation & \cmark & \cmark & \xmark & \xmark & \cmark & \xmark \\
Quantum Uncomputation & \cmark & \xmark & \xmark & \xmark & \xmark & \xmark \\ \bottomrule
\end{tabular}
\end{center}
\bigskip
\end{minipage}
\caption{Comparison with existing languages}
\label{tab:comparison}
\end{table}

Qunity's design was inspired by several existing languages, summarized in Table~\ref{tab:comparison}: Silq \cite{silq}, \textsc{Qml} \cite{qml}, the symmetric pattern matching language \cite{symmetric-pattern-matching} (hereafter \textsc{Spm}), Quipper \cite{quipper}, and Qu\textsc{Gcl} \cite[Chapter 6]{mingsheng-ying}.
This table shows some of Qunity's most interesting features and whether these features are present in existing languages.
The features listed are:
\begin{itemize}
	\item
		Decoherence.
		All of these languages support some form of measurement or discarding, except for \textsc{Spm}, which is restricted to unitary operators.
	\item
		Compositional denotational semantics.
		Most of these languages define an operational semantics, not a denotational one.
		Qu\textsc{Gcl}'s denotational semantics is explicitly \emph{non}-compositional, meaning that its equivalence relation does not allow ``equivalent'' programs to be substituted as subexpressions.
		Later work on \textsc{Qml} \cite{qml-update,qml-thesis} found that the denotational semantics derived from \textsc{Qml}'s operational semantics was non-compositional because of entangled ``garbage'' outputs related to sum types, and \textsc{Qml}'s denotational semantics is only partially defined even with sum types removed.
\item
		Quantum sum types. 
		Like \textsc{Qml} and \textsc{Spm}, Qunity can coherently manipulate tagged unions, and the ``tag'' qubit can be in superposition.
	\item
		Isometries.
		Silq is the only other language here with a semantics defined in terms of non-unitary isometries.
		Quipper and \textsc{Qml} allow isometries to be implemented through qubit initialization, but their semantics are defined only in terms of unitaries, meaning there is no notion of equivalence between different unitaries implementing the same isometry.
	\item
		Projectors.
		Like with isometries, Quipper allows projectors to be implemented through assertative termination, but there is no notion of equivalence between different unitaries implementing the same projector.
	\item
		Classical uncomputation.
		Silq and Quipper both have convenient facilities for converting (irreversible) classical programs into reversible quantum subroutines.
	\item
		Quantum uncomputation.
		Silq can only uncompute ``qfree'' programs that are strictly classical.
		Quipper does have some facilities for uncomputing values produced by quantum programs, such as the \texttt{with\_computed} function.
		However, this function's correctness depends on conditions that are nontrivial to verify.
		In the words of Quipper's documentation: ``This is a very general but relatively unsafe operation. It is the user's responsibility to ensure that the computation can indeed be undone.'' In Qunity, safety is assured.
\end{itemize}

\subsection{Contributions and Roadmap}
		Our core contribution is a new quantum programming language, Qunity, designed to unify classical and quantum computing through an expressive generalization of classical programming constructs.
		Qunity's powerful semantics brings constructions commonly used in algorithm \emph{analysis}---such as bounded-error quantum subroutines, projectors, and direct sums---into the realm of algorithm \emph{implementation}.
We
describe Qunity's formal syntax (Section~\ref{sec:syntax}), an efficiently checkable typing relation (Section~\ref{sec:types}), a compositional denotational semantics (Section~\ref{sec:semantics}), and a strategy for compiling to lower-level unitary quantum circuits (Section~\ref{sec:compilation}).
We prove that well-typed Qunity programs have a well-defined semantic  denotation, 
and we show that denotation is realizable by proving that our compilation strategy does indeed produce correct circuits.
We also show how Qunity can be used to program several interesting examples, including Grover's algorithm, the quantum Fourier transform, and a quantum walk (Section~\ref{sec:examples}).

\section{Syntax}
\label{sec:syntax}

\begin{figure}[t]
  \centering
  \begin{minipage}[t]{.45\textwidth}
\begin{alignat*}{3}
	T \defeqq &\quad&&&\textit{(data type)} \\
						&&& \Void &\quad \textit{(bottom)} \\
						|&&& \Unit &\quad \textit{(unit)} \\
						|&&& T \oplus T &\quad \textit{(sum)} \\
						|&&& T \otimes T &\quad \textit{(product)} \\
						F \defeqq &\quad &&& \textit{(program type)} \\
						&&&T \rightsquigarrow T &\quad \textit{(coherent map)} \\
						|&&& T \Rrightarrow T &\quad \textit{(quantum channel)} \\
\end{alignat*}
\caption{Qunity types}
\label{fig:types}
\smallskip
\begin{alignat*}{3}
	\Gamma &\defeqq&& &\textit{(context)} \\
								 &&& \varnothing & \textit{(empty)} \\
	&&|\;& \Gamma, x : T & \textit{(binding)} \\
	\dom(\varnothing) &\defeq&& \varnothing &        \textit{(dom-none)}\\
	\dom(\Gamma, x : T)&\defeq&& \dom(\Gamma) \cup \{ x \} \; & \textit{(dom-bind)}\\
\end{alignat*}
\caption{Typing contexts}
\label{fig:context}
  \end{minipage}
  \hspace{.1in}
  \begin{minipage}[t]{.5\textwidth}
\begin{alignat*}{3}
	e \defeqq &&&&\quad\textit{(expression)} \\
						&& \;& \unit &\quad \textit{(unit)} \\
						|&&& x &\quad \textit{(variable)} \\
						|&&& \pair e e &\quad \textit{(pair)} \\
						|&&& \cntrl{e}{T}{e &\mapsto e \\ &\cdots \\ e &\mapsto e}{T} &\quad \textit{(coherent control)} \\
						|&&& \trycatch{e}{e} &\quad \textit{(error recovery)} \\
						|&&& f\: e &\quad \textit{(application)} \\
	f \defeqq &&&&\quad\textit{(program)} \\
						&&& \uthree r r r &\quad \textit{(qubit gate)} \\
						|&&& \lef T T &\quad \textit{(left tag)} \\
						|&&& \rit T T &\quad \textit{(right tag)} \\
						|&&& \lambda e \xmapsto{{\color{gray}{T}}} e &\quad \textit{(abstraction)} \\
						|&&& \rphase T e r r &\quad \textit{(relative phase)}
\end{alignat*}
\caption{Base Qunity syntax}
\label{fig:syntax}
\end{minipage}
\end{figure}

Qunity's formal syntax is defined in
Figures~\ref{fig:types}--\ref{fig:syntax}. 
Qunity's types are shown in Figure~\ref{fig:types}. 
The algebraic \emph{data} types $T$ have essentially the same interpretation as in a typical classical programming language, with the caveat that values can be in superposition.
The symbols $\oplus$ and $\otimes$ are used because of how these data types will correspond to direct sums and tensor products.
The two \emph{program} types $F$ differ in whether they decohere quantum states: programs of type $T \rightsquigarrow T'$ will have a semantics defined by an (often unitary) linear operator, while programs of type $T \Rrightarrow T'$ will have a semantics defined by a superoperator, a completely positive trace-non-increasing map that may involve measurement or discarding.

Qunity's term language is defined in Figure~\ref{fig:syntax}:
\emph{expressions} $e$ are assigned data types $T$ and \emph{programs} $f$
are assigned program types $F$. Figure~\ref{fig:notation} shows additional derived forms. 
In the figure, $x$ ranges over some infinite set $\X$ of variables (for example \textsc{Ascii} strings), and $r$ ranges over some representation of real numbers
\ifshort
(one example of which is defined in the supplemental report \cite{supplement}).
\else
(such as that defined in Appendix~\ref{app:real}).
\fi
Some syntax elements use type annotations $T$, which are grayed out
to reduce visual clutter in the more complex
examples. A type inference algorithm could allow
these annotations could to be removed, and throughout this paper, we occasionally omit type annotations for brevity.

\begin{figure}[ht]
  \centering
  \begin{minipage}{.4\textwidth}
	\begin{align*}
		\Bit &\defeq \Unit \oplus \Unit \\
		\zero &\defeq \texttt{left}_{\color{gray}\Bit} \unit \\
		\one &\defeq \texttt{right}_{\color{gray}\Bit} \unit \\
	f^{\dagger {\color{gray}{T}}} &\defeq (\lambda (f\: x) \xmapsto{{\color{gray}{T}}} x) \\
		\Maybe(T) &\defeq \Unit \oplus T \\
		\nothing{T} &\defeq \texttt{left}_{\color{gray}\Maybe(T)} \unit \\
		\just{T} &\defeq \texttt{right}_{\color{gray}\Maybe(T)} \\
		T^{\otimes 0} &\defeq \Unit \\
		T^{\otimes(n+1)} &\defeq T \otimes T^{\otimes n} \\
	e^{\otimes 0} &\defeq \unit \\
	e^{\otimes (n+1)} &\defeq \pair{e}{e^{\otimes n}} \\
		\gphase{T}{r} &\defeq \rphase{T}{x}{r}{r} \\
  \end{align*}
  \end{minipage}
  \begin{minipage}{.59\textwidth}
	\begin{align*}
	(e \triangleright f) &\defeq (f\; e) \\
	(\texttt{let } e_1 \texttt{ =}_{\color{gray}{T}}\; e_2 \texttt{ in } e_3) &\defeq (e_2 \triangleright \lambda e_1 \xmapsto{{\color{gray}{T}}} e_3) \\
	(f \circ_{\color{gray}{T}} f') &\defeq (\lambda x \xmapsto{{\color{gray}{T}}} f (f'\: x)) \\
	\fst{T_0}{T_1} &\defeq \lambda \pair{x_0}{x_1} \xmapsto{{\color{gray}{T_0 \otimes T_1}}} x_0 \\
	\snd{T_0}{T_1} &\defeq \lambda \pair{x_0}{x_1} \xmapsto{{\color{gray}{T_0 \otimes T_1}}} x_1 \\
	\texttt{had} &\defeq \uthree{\pi \divby 2}{0}{\pi} \\
	\texttt{plus} &\defeq \had\:\zero \\
	\texttt{minus} &\defeq \had\:\one \\
	\texttt{equals}_{\color{gray}T}(e) &\defeq \left( \begin{aligned} &\lambda x \xmapsto{\textcolor{gray}T} \\ &\texttt{try} (x \triangleright \lambda e \xmapsto{\textcolor{gray}T} \one) \texttt{ catch 0} \end{aligned} \right) \\
	\reflect{T}(e) &\defeq \rphase{T}{e}{0}{\pi} \\
  \end{align*}
  \end{minipage}
	\caption{Syntactic sugar}
	\label{fig:notation}
\end{figure}

Qunity makes use of standard classical programming language features  generalized to the quantum setting.
As examples, pairs generalize to creating tensor products of quantum states, sums generalize to allowing their data to be in superposition, and \texttt{ctrl} generalizes classical pattern-matching (where in branch $e \mapsto e'$ the $e$ is a pattern which may bind variables in $e'$) 
to \emph{quantum control flow} using superposition.
The type system requires that the left-hand-side patterns are non-overlapping and thus correspond to orthogonal subspaces, and that the right-hand-side expressions appropriately use all of the variables from the condition expression, ensuring that this data is reversibly uncomputed rather than discarded.
In general, we allow the left-hand-side patterns to be non-exhaustive, in which case the semantics is a norm-\emph{decreasing} operator rather than norm-preserving.
Operationally, a decrease in norm corresponds to the probability of being in a special ``exceptional state.''

Existing languages like Proq \cite{proq} have similarly used norm-decreasing operators such as projectors to describe assertions.
Their system is effective for testing and debugging, but Qunity takes these projective assertions a step further by allowing them to be used in the control flow itself.
Failed assertions are treated as \emph{exceptions}, which can be dynamically caught and handled using the \texttt{try}-\texttt{catch} construct, generalizing another familiar classical programming construct to the quantum setting.
The quantum behavior here is well-defined using the language of \emph{projective measurements}.
Assuming projector $P$ is implemented by Qunity program $f_P$, and state $\ket \psi$ is produced by Qunity program $e_\psi$, the Qunity expression $\trycatch{\just{T} (f_P\:e_\psi)}{\nothing{T}}$ (using \texttt{Maybe} syntax from Figure~\ref{fig:notation}) produces the mixed state defined by the density operator:
\[
	P \op{\psi}{\psi} P \oplus \bra{\psi} (I - P) \ket{\psi}.
\]
This state has $P\ket{\psi}$ in the ``\texttt{just}'' subspace and the norm of $(I - P) \ket{\psi}$ in the ``\texttt{nothing}'' subspace.
Though we use the language of exception handling, this construct is also useful for non-exceptional conditions, like in the definition of the \texttt{equals} program in Figure~\ref{fig:notation}.
This function can be used to measure whether two states are equal; for example, a simple ``coin flip'' can be implemented by $(\had\: \zero \triangleright \texttt{equals}_{\color{gray}\Bit}(\one))$, which applies a Hadamard gate to a qubit in the $\ket\zero$ state and then measures whether the result is $\ket\one$.

Notice that the innermost lambda in the definition of \texttt{equals} uses a non-exhaustive pattern on the left side; the program ``$\lambda \zero \xmapsto{{\color{gray}\Bit}} \one$'' implements the projector $\op{\one}{\zero}$.
While a classical operational semantics typically interprets lambdas in terms of \emph{substitution} of arguments for parameters, Qunity's semantics is best interpreted as a more general linear mapping.
This interpretation means that Qunity's type system can allow for a much wider range of expressions on the left ``parameter'' side of the lambda.
As an extreme case, consider the definition of ``$f^{\dagger {\color{gray}{T}}}$'' syntax in Figure~\ref{fig:notation}, which applies a function on the left side.
Semantically, this lambda corresponds to the adjoint of $f$, and can be interpreted like this: ``given $f(x)$ as input, output $x$.''
(However, this interpretation is imprecise when $f$ is norm-decreasing and its adjoint is not its inverse.)

Two of Qunity's constructs have no classical analog.
The $\texttt u_{\texttt 3}$ construct is a parameterized gate that allows us to implement any single-qubit gate \cite{openqasm2}; e.g., it is used to implement \texttt{had}, per Figure~\ref{fig:notation}. The $\texttt{rphase}$ construct induces a relative phase of $e^{ir}$, where the value $r$ comes from the first branch for states in the subspace spanned by the first branch's parameter $e$ and from the second branch for states in the orthogonal subspace.
It is used to implement $\texttt{gphase}$ in Figure~\ref{fig:notation}, used to implement conditional phase flip from Deutsch's algorithm (Section~\ref{sec:deutsch}). More generally, it can be used to implement reflections, as used in Grover's search algorithm, the iterator for which is shown below. We see that it includes the same conditional phase flip as Deutsch's, composed with $\reflect{\Bit^{\otimes n}}(\texttt{plus}^{\otimes n})$ to perform inversion about the mean.
If the state $\ket{\psi}$ is implemented by the expression $e_\psi$, then $\reflect{T}(e)$ implements the reflection $(2\op{\psi}{\psi} - I)$ by coherently applying a phase of $e^{i \pi} = -1$ to any state orthogonal to $\ket{\psi}$.

\begin{alignat*}{3}
	\grover{n}(f) &\defeq&& \lambda x \xmapsto{{\color{gray}\Bit^{\otimes n}}} \cntrl{f\: x}{\Bit}{\zero &\mapsto x \\ \one &\mapsto x \triangleright \gphase{\Bit^{\otimes n}}{\pi}}{\Bit^{\otimes n}} \triangleright \reflect{\Bit^{\otimes n}}(\texttt{plus}^{\otimes n})
\end{alignat*}

Several lambda calculi \cite{qlambda,qlambda-vantonder,lineal} have explored the use of higher-order functions in a quantum setting, but Qunity does not.
We aim for Qunity's denotational semantics to closely correspond to existing notations and conventions used in quantum algorithms, and higher-order functions and ``superpositions of programs'' are uncommon and inconsistently defined.
In our experience, quantum mechanical notation (and to some degree quantum computing in general) is ill-suited for higher-order programs.
In particular, we interpret programs $f$ as quantum operations acting on their argument, and expressions $e$ as quantum operations acting on their free variables, but allowing for higher-order functions means that programs can also have free variables, and one must take some sort of tensor product of a program's two inputs.
To avoid this extra complication, we have first-order functions only, and our typing relation prevents them from containing free variables.

\section{Typing}
\label{sec:types}

This section describes Qunity's type system. We prove that well-typed programs have a well-defined semantics (given in Section~\ref{sec:semantics}), and we have implemented a type checker for Qunity programs \cite{typechecker}.

Qunity's type system takes the form of three different,
interdependent typing judgments: \emph{pure expression typing},
\emph{mixed expression typing}, and \emph{program typing}.
The distinction between pure and mixed typing comes from the fact that there are two ways to mathematically represent quantum states depending on whether any classical probability is involved: pure states are described by state vectors and do not involve classical probability, while mixed states are described by density matrices, usually interpreted as a classical probability distribution over pure states.
Program types are similarly divided into pure and mixed versions, where pure programs correspond to linear operators from pure states to pure states, and mixed programs correspond to \emph{superoperators} from mixed states to mixed states.
\ifshort
In the supplemental report \cite{supplement},
\else
In Appendix~\ref{app:superop},
\fi
we include further discussion on the need for both of these in Qunity.

Judgments are parameterized by \emph{typing contexts}, which are ordered lists of variable-type pairs, as defined in Fig.~\ref{fig:context}. 
We say that a typing context $(x_1 : T_1, \ldots, x_n : T_n)$ is \textit{well-formed} if all variables are distinct; that is, if $x_j \neq x_k$ whenever $j \neq k$. Concatenating two typing contexts $\Gamma_1$ and $\Gamma_2$ is written $\Gamma_1, \Gamma_2$.
Within the inference rules, there is an implicit assumption that contexts are well-formed (so concatenation requires $\dom(\Gamma_1) \cap \dom(\Gamma_2) = \varnothing$). 

\subsection{Pure expression typing}

The pure expression
typing judgment is written $\Gamma\partition \Delta \vdash e : T$,\footnote{Whether a context is classical or quantum is based on its \emph{position} in the typing judgment---left of $\partition$ is classical and right of it is quantum. We write $\Gamma$ ($\Delta$, resp.) for classical (quantum, resp.) variables as a common but not universal convention.}
indicating that expression $e$ has pure type $T$ with respect to classical context $\Gamma$ and quantum context $\Delta$.
Variables in the classical context are in a classical basis state, and are automatically uncomputed, while quantum context variables are consumed and may be in superposition. 
Operationally, classical data is still compiled into qubits, but these qubits are only used as control wires for controlled operations, and they are uncomputed when they go out of scope.
Whenever $\varnothing \partition \varnothing \vdash e : T$ holds, the semantics corresponds to a pure state in $\Hilb(T)$, the Hilbert space assigned to $T$.

\begin{figure}[ht]
\[
	\inference{}{\Gamma\partition \varnothing \vdash \unit : \Unit}[\textsc{T-Unit}]
	\quad
	\inference{}{\Gamma, x : T, \Gamma'\partition \varnothing \vdash x : T}[\textsc{T-Cvar}]
	\quad
	\inference{x \notin \dom(\Gamma)}{\Gamma\partition x : T \vdash x : T}[\textsc{T-Qvar}]
\]
\vspace{2mm}
\[
	\inference{\Gamma\partition \Delta, \Delta_0 \vdash e_0 : T_0 \qquad \Gamma\partition \Delta, \Delta_1 \vdash e_1 : T_1}{\Gamma\partition \Delta, \Delta_0, \Delta_1 \vdash \pair {e_0} {e_1} : T_0 \otimes T_1}[\textsc{T-PurePair}]
\]
\vspace{2mm}
\[
	\inference{\Gamma, \Delta \Vdash e : T \qquad \ortho{T}{e_1, \ldots, e_n} \qquad \varnothing\partition \Gamma_j \vdash e_j : T \text{ for all } j \\ \erases{T'}(x; e_1', \ldots, e_n') \text{ for all } x \in \dom(\Delta) \qquad \Gamma, \Gamma', \Gamma_j\partition \Delta, \Delta' \vdash e_j' : T' \text{ for all } j}{\Gamma, \Gamma' \partition \Delta, \Delta' \vdash \cntrl{e}{T}{e_1 &\mapsto e_1' \\ &\cdots \\ e_n &\mapsto e_n'}{T'} : T'}[\textsc{T-Ctrl}]
\]
\vspace{2mm}
\[
	\inference{\vdash f : T \rightsquigarrow T' \qquad \Gamma\partition \Delta \vdash e : T}{\Gamma\partition \Delta \vdash f\: e : T'}[\textsc{T-PureApp}]
	\quad
	\inference{\Gamma \partition \Delta \vdash e : T}{\pi\subcap{g}(\Gamma) \partition \pi\subcap{d}(\Delta) \vdash e : T}[\textsc{T-PurePerm}]
\]
\caption{Pure expression typing rules}
\label{fig:t-pure-exp}
\end{figure}

The pure expression typing rules are given in Figure~\ref{fig:t-pure-exp}. As is typical in quantum computing languages, these rules are \emph{substructural} \cite{substructural}. In particular, as in \textsc{Qml} \cite{qml}, quantum variables are \emph{relevant}, meaning they must be used \emph{at least once}.
This invariant is evident from the \textsc{T-QVar} rule, which requires the quantum context to contain only the variable $x$ of interest, and the \textsc{T-CVar} and \textsc{T-Unit} rules, which require the quantum context to be empty. Indeed, we can prove that $\Gamma\partition \Delta \vdash e : T$ implies $\dom(\Delta) \subseteq \FV(e)$. On the other hand, the rules do not enforce $\dom(\Gamma) \subseteq \FV(e)$. 

By enforcing variable relevance, Qunity's type system can control what parts of a program are allowed to discard information.
Semantically, the non-use of a variable corresponds to a \emph{partial trace}, a quantum operation described by decoherence.
Quantum control flow is ill-defined in the presence of decoherence \cite{alternation}, and decoherence is inherently irreversible, so we use relevant types wherever computation must be reversible or subject to quantum control.

A quantum lambda calculus \cite{qlambda} uses \emph{affine} types for quantum data, ensuring quantum variables are used at most once, while Proto-Quipper~\cite{protoquipper} and \textsc{Qwire} \cite{qwire} use \emph{linear} types, requiring them to be used exactly once. 
This is done as static enforcement of the no-cloning theorem \cite{nocloning}, which makes sense in the \textsc{qram} computational model \cite{qram-pseudocode}. However, Qunity's aim (like \textsc{Qml}'s \cite{qml}) is to treat quantum computing as a generalization of classical computing. Qunity fixes a particular standard computational basis and treats expression $(x,x)$ as \emph{entangling} variable $x$ when it is in superposition, essentially as a linear isometry $\alpha\ket{\zero} + \beta\ket{\one} \mapsto \alpha\ket{\zero\zero} + \beta\ket{\one\one}$. The expression is deemed well typed by \textsc{T-PurePair}---notice that $\Delta$ appears when typing both $e_0$ and $e_1$, allowing duplication (``sharing'') of quantum resources across the pair. In contrast, \textsc{qram}-based languages reject $(x,x)$ to avoid confusion with the non-physical cloning function 
$\alpha\ket{\zero} + \beta\ket{\one} \mapsto (\alpha\ket{\zero} + \beta\ket{\one}) \otimes (\alpha\ket{\zero} + \beta\ket{\one})$.

The difference between ``sharing'' and ``cloning'' affects only non-basis states.
\emph{Sharing} a qubit $\alpha\ket{\zero} + \beta\ket{\one}$ to a second qubit register produces the entangled state $\alpha\ket{\texttt{00}} + \beta\ket{\texttt{11}}$, while \emph{cloning} would produce the unentangled state $(\alpha\ket{\zero} + \beta\ket{\one})(\alpha\ket{\zero} + \beta\ket{\one}) = \alpha^2\ket{\texttt{00}} + \alpha \beta \ket{\texttt{01}} + \alpha\beta\ket{\texttt{10}} + \beta^2 \ket{\texttt{11}}$.
The former is basis-dependent and implementable in Qunity as ``$\pair{x}{x}$,'' while the latter is basis-independent and physically prohibited by the no-cloning theorem.

The \textsc{T-PureApp} rule applies a linear operator $f$ to its argument $e$. 
The \textsc{T-PurePerm} rule exists to allow for the usual structural rule of exchange, which is typically implicit.
Here and throughout this work, the functions $\pi$ are list permutation functions, arbitrarily permuting the bindings within a context.
Making the exchange rule explicit allows compilation to be a direct function of the typing judgment, with \textsc{swap} gates introduced at uses of the explicit exchange rule (Section~\ref{sec:qunity2hl}).

Rule \textsc{T-Ctrl} types pattern matching and quantum control. We defer discussing it to Section~\ref{sec:ctrl}, after we have considered the other judgments. 

\subsection{Mixed expression typing}

We write $\Delta \Vdash e : T$ to indicate that expression $e$ has mixed type $T$ under quantum context $\Delta$;\footnote{We use the double-lined ``$\Vdash$'' symbol because this typing judgment can be used for measurements, and quantum circuit diagrams conventionally use double-lined wires to carry measurement results.} the rules are in Figure~\ref{fig:t-mixed-exp}.
$\varnothing \Vdash e : T$ implies $e$'s semantics corresponds to a \emph{mixed state} in $\Hilb(T)$.

\begin{figure}[ht]
\[
	\inference{\varnothing \partition \Delta \vdash e : T}{\Delta \Vdash e : T}[\textsc{T-Mix}]
	\quad
	\inference{\Delta \Vdash e : T}{\pi(\Delta) \Vdash e : T}[\textsc{T-MixedPerm}]
	\]
\vspace{2mm}
	\[
	\inference{\Delta, \Delta_0 \Vdash e_0 : T_0 \qquad \Delta, \Delta_1 \Vdash e_1 : T_1}{\Delta, \Delta_0, \Delta_1 \Vdash \pair{e_0}{e_1} : T_0 \otimes T_1}[\textsc{T-MixedPair}]
\]
\vspace{2mm}
\[
	\inference{\Delta_0 \Vdash e_0 : T \qquad \Delta_1 \Vdash e_1 : T}{\Delta_0, \Delta_1 \Vdash \trycatch{e_0}{e_1} : T}[\textsc{T-Try}]
	\quad
	\inference{\vdash f : T \Rrightarrow T' \qquad \Delta \Vdash e : T}{\Delta \Vdash f\: e : T'}[\textsc{T-MixedApp}]
\]
\caption{Mixed expression typing rules}
\label{fig:t-mixed-exp}
\end{figure}

Rule \textsc{T-Mix} allows pure expressions to be typed as mixed.
(One could equivalently treat pure types as a \emph{subtype} of mixed types.)
Rules \textsc{T-MixedPerm} and \textsc{T-MixedPair} are analogous to the pure-expression versions. Rule \textsc{T-MixedApp} allows applying a quantum channel $f$ on an expression $e$---such an $f$ may perform measurements, as discussed below.
Rule \textsc{T-Try} allows exceptions occurring in expression $e_0$ to be caught and replaced by expression $e_1$.
Operationally, this is effectively ``measuring whether an error occurred'' and thus \texttt{try}-\texttt{catch} expressions have no \emph{pure} type; it also cannot be done without perturbing the input data $\Delta_0$ and thus expression $e_1$ is typed in a separate context.

\subsection{Program typing}
  
We write $\vdash f : F$ to indicate that program $f$ has type $F$; the rules are in Figure~\ref{fig:t-prog}.
Whenever $\vdash f : T \rightsquigarrow T'$ holds, the semantics corresponds to a linear operator mapping $\Hilb(T)$ to $\Hilb(T')$.
Whenever $\vdash f : T \Rrightarrow T'$ holds, the semantics corresponds to a superoperator mapping mixed states in $\Hilb(T)$ to mixed states in $\Hilb(T')$.

\begin{figure}[ht]
\[
	\inference{}{\vdash \uthree{r_\theta}{r_\phi}{r_\lambda} : \Bit \rightsquigarrow \Bit}[\textsc{T-Gate}]
\]
\vspace{2mm}
\[
	\inference{}{\vdash \lef{T_0}{T_1} : T_0 \rightsquigarrow T_0 \oplus T_1}[\textsc{T-Left}]
	\quad
	\inference{}{\vdash \rit{T_0}{T_1} : T_1 \rightsquigarrow T_0 \oplus T_1}[\textsc{T-Right}]
\]
\vspace{2mm}
\[
	\inference{\varnothing\partition \Delta \vdash e : T \qquad \varnothing\partition \Delta \vdash e' : T'}{\vdash \lambda e \xmapsto{{\color{gray}T}} e' : T \rightsquigarrow T'}[\textsc{T-PureAbs}]
	\quad
	\inference{\varnothing \partition \Delta \vdash e : T}{\vdash \rphase{T}{e}{r}{r'} : T \rightsquigarrow T}[\textsc{T-Rphase}]
\]
\vspace{2mm}
\[
	\inference{\vdash f : T \rightsquigarrow T'}{\vdash f : T \Rrightarrow T'}[\textsc{T-Channel}]
	\quad
	\inference{\varnothing\partition \Delta, \Delta_0 \vdash e : T \qquad \Delta \Vdash e' : T'}{\vdash \lambda e \xmapsto{{\color{gray}T}} e' : T \Rrightarrow T'}[\textsc{T-MixedAbs}]
\]
\caption{Program typing rules}
\label{fig:t-prog}
\label{fig:t-pure-prog}
\label{fig:t-mixed-prog}
\end{figure}

\textsc{T-Gate} types a single-qubit unitary gate, and \textsc{T-Left} and \textsc{T-Right} type sum introduction. These are linear operations. \textsc{T-PureAbs} types a linear abstraction. In abstraction $\lambda e \xmapsto{{\color{gray}T}} e'$, variables introduced in $e$ are \emph{relevant} in $e'$---they must be present in quantum context $\Delta$ used to type both $e$ and $e'$. (It is not hard to prove that $\Delta$ can always be uniquely determined.) \textsc{T-MixedAbs} is more relaxed: variables in $e$ can be contained in a context $\Delta_0$ which is \emph{not} used to type $e'$; such variables are \emph{discarded} in $e'$, which implies measuring them. So the typing judgment $\vdash \lambda x \xmapsto{{\color{gray}\Bit}} \zero : \Bit \rightsquigarrow \Bit$ is invalid but the typing judgment $\vdash \lambda x \xmapsto{{\color{gray}\Bit}} \zero : \Bit \Rrightarrow \Bit$ is valid.
\textsc{Rphase} uses the expression $e$ as a pattern for coherently inducing a phase, either $e^{i r}$ or $e^{i r'}$ depending on whether pattern $e$ is matched.
\textsc{T-Channel} permits a linear operator to be treated as a superoperator.

Qunity programs are typed without context to help avoid scenarios where ``entangling through variable re-use'' might be confusing.
For example, the program $\lambda x \xmapsto{{\color{gray}\Bit}} (x \triangleright (\lambda y \xmapsto{{\color{gray}\Bit}} x))$ is ill-typed in Qunity because the subprogram $(\lambda y \xmapsto{{\color{gray}\Bit}} x)$ has a free variable $x$.
This is not a major loss in expressiveness because the rewritten program $\lambda x \xmapsto{{\color{gray}\Bit}} \texttt{let } \pair{x_0}{x_1} =_{\color{gray}\Bit \otimes \Bit} \pair{x}{x} \texttt{ in}\; x_0$ is valid, with type $\Bit \Rrightarrow \Bit$.
This program measures a qubit by re-using the variable to share the qubit to a second register and then discarding the entangled qubit, and the valid program makes this sharing explicit.
Qunity is designed so that pairing is the only way to perform this kind of entanglement, by using a variable on both sides of the pair.

\subsection{Typing quantum control}
\label{sec:ctrl}

Qunity allows for quantum control by generalizing pattern matching via \texttt{ctrl}, which is typed via the \textsc{T-Ctrl} rule (Figure~\ref{fig:t-pure-exp}). The $e_j$s in the premises of this rule refer to the indexed expressions in the conclusion.
		A ``prime'' symbol should be viewed as part of the variable name.
The following is the \textsc{T-Ctrl} rule for $n=2$:
		\[
	\inference{\Gamma, \Delta \Vdash e : T \qquad \ortho{T}{e_1, e_2} \qquad \varnothing\partition \Gamma_1 \vdash e_1 : T \qquad \varnothing\partition \Gamma_2 \vdash e_2 : T \\ \erases{T'}(x; e_1', e_2') \text{ for all } x \in \dom(\Delta) \qquad \Gamma, \Gamma', \Gamma_1\partition \Delta \vdash e_1' : T' \qquad \Gamma, \Gamma', \Gamma_2\partition \Delta \vdash e_2' : T'}{\Gamma,\Gamma'\partition \Delta \vdash \cntrl{e}{T}{e_1 &\mapsto e_1' \\ e_2 &\mapsto e_2'}{T'} : T'}
			\]

To ensure realizable circuits, the type rule imposes several restrictions on patterns. First, both pattern and body expressions $e_j$ and $e_j'$ must be pure expressions; this means they cannot include invocations of superoperators, which might involve measurements.  
Some recent work has tried to generalize the definition of quantum control to a notion of ``quantum alternation'' that allows measurements to be controlled.
Qu\textsc{Gcl} \cite[Chapter 6]{mingsheng-ying} attempts this, but the resulting denotational semantics is non-compositional.
B\u adescu and Panangaden \cite{alternation} shed more light on the problems with quantum alternation, proving that quantum alternation is not monotone with respect to the L\"owner order \cite[p.~17]{mingsheng-ying} and concluding that ``quantum alternation is a fantasy arising from programming language semantics rather than from physics.''

For \texttt{ctrl} semantics to be physically meaningful, \textsc{T-Ctrl} depends on two additional judgments applied to the body expressions---``ortho'' for the left-hand-side (defined further below), and ``erases'' for the right-hand-side (defined in Figure~\ref{fig:erasure}).
The ortho judgment ensures that the left-hand-side patterns $e_j$ are purely classical and non-overlapping.
The erases judgment (defined using \texttt{gphase} syntax from Figure~\ref{fig:notation}) ensures that all of the right-hand-side patterns $e_j'$ use the variables from $e$ in a consistent way, and its name comes from the way this judgment is used by the compiler to coherently ``erase'' these variables.
The contexts $\Gamma_j$ are the primary motivation that the typing relation includes classical contexts $\Gamma$ at all.
Semantically, the variables they represent are purely classical; operationally, the information on these registers is used without being ``consumed,'' under the assumption that it will be uncomputed.
When typing the expressions $e_j$, these contexts $\Gamma_j$ appear in the ``quantum'' part of the context so that the type system can enforce variable relevance, but the restrictions placed by the ortho judgment ensure that these variables are still practically classical.

\begin{figure}[t]
\[
	\inference{\erases{T}(x; e_1, \ldots, e_n)}{\erases{T}(x; e_1, \ldots, e_{j-1}, e_j \triangleright \gphase{T}{r}, e_{j+1}, \ldots, e_n)}[\textsc{E-Gphase}]
\]
\vspace{2mm}
\[
	\inference{\erases{T}(x; e_1, \ldots, e_{j-1}, e_{j,1}, \ldots, e_{j,m}, e_{j+1}, \ldots, e_n)}{\erases{T}\left(x; e_1, \ldots, e_{j-1}, \cntrl{e}{T'}{e_1' &\mapsto e_{j,1} \\ &\cdots \\ e_m' &\mapsto e_{j,m}}{T}, e_{j+1}, \ldots, e_n\right)}[\textsc{E-Ctrl}]
\]
\vspace{2mm}
\[
	\inference{}{\erases{T}(x; x, x, \ldots, x)}[\textsc{E-Var}]
	\quad
	\inference{\erases{T_0}(x; e_{0,1}, \ldots, e_{0,n})}{\erases{T_0 \otimes T_1}(x; \pair{e_{0,1}}{e_{1,1}}, \ldots, \pair{e_{0,n}}{e_{1,n}})}[\textsc{E-Pair0}]
\]
\vspace{2mm}
\[
	\inference{\erases{T_1}(x; e_{1,1}, \ldots, e_{1,n})}{\erases{T_0 \otimes T_1}(x; \pair{e_{0,1}}{e_{1,1}}, \ldots, \pair{e_{0,n}}{e_{1,n}})}[\textsc{E-Pair1}]
\]
\caption{Erasure inference rules}
\label{fig:erasure}
\end{figure}

The orthogonality judgment is defined in terms of a ``spanning'' judgment defined
\ifshort
in the supplemental report \cite{supplement}.
\else
in Appendix~\ref{app:spanning}.
\fi
This judgment, written $\spanning{T}{e_1, \ldots, e_n}$ and largely inspired by \textsc{Spm}'s \cite{symmetric-pattern-matching} ``orthogonal decomposition'' judgment, denotes that the set of expressions $\{e_1, \ldots, e_n\}$ is a spanning set for type $T$.
This judgment enforces that the expressions form an exhaustive set of patterns for the type $T$, which in our quantum setting semantically corresponds to a set of states that span the Hilbert space corresponding to $T$.

We write $\ortho{T}{e_1, \ldots, e_n}$ to denote that the set of expressions $\{e_1, \ldots, e_n\}$ is orthogonal.
An orthogonal set of expressions is simply a subset of some spanning set.
That is, orthogonality judgments can be defined by the following inference rule alone:
\[
	\inference{\spanning{T}{e_1', \ldots, e_m'} \qquad [e_1, \ldots, e_n] \textup{ is a subsequence of } [e_1', \ldots, e_m']}{\ortho{T}{e_1, \ldots, e_n}}
\]
It is not hard to prove that orthogonality holds for a set of expressions regardless of the order they appear in the judgment (e.g., $\ortho{T}{e_1,  e_2}$ \emph{iff} $\ortho{T}{e_2,  e_1}$).

\section{Semantics}
\label{sec:semantics}

In this section, we define Qunity's denotational semantics in terms of linear operators (for pure expressions) and superoperators (for mixed expressions).
These definitions are designed to naturally generalize from a classical semantics defined on a classical sublanguage of Qunity.
After presenting the semantics, we present some metatheoretical results, most notably that well-typed Qunity programs are well-defined according to the semantics.

\subsection{Classical sublanguage semantics}
\label{sec:classical-semantics}

Qunity's semantics may be more intuitive if we first restrict ourselves to a classical sublanguage.

\begin{definition}[classical sublanguage]
	\label{def:sublanguage}
	Qunity's classical sublanguage is defined by removing the \texttt{u3} and \texttt{rphase} constructs from the language, producing the following:
\begin{align*}
	e &\defeqq \unit \mid x \mid \pair{e}{e} \mid \cntrl{e}{T}{e &\mapsto e \\ &\cdots \\ e &\mapsto e}{T} \mid \trycatch{e}{e} \mid f\: e \\
	f &\defeqq \lef T T \mid \rit T T \mid \lambda e \xmapsto{{\color{gray}{T}}} e
\end{align*}
\end{definition}

We can define a classical denotational semantics for this sublanguage using partial functions over \emph{values} and \emph{valuations}.

\begin{definition}[value and valuation]
	For any type $T$, we write $\V(T)$ to denote the set of expressions that are values of that type.
	\begin{align*}
		\V(\Void) &\defeq \varnothing \\
		\V(\Unit) &\defeq \{ \unit \} \\
		\V(T_0 \oplus T_1) &\defeq \{ \lef{T_0}{T_1}\: v_0 \mid v_0 \in \V(T_0)\} \cup \{ \rit{T_0}{T_1}\: v_1 \mid v_1 \in \V(T_1)\} \\
		\V(T_0 \otimes T_1) &\defeq \{ \pair{v_0}{v_1} \mid v_0 \in \V(T_0), v_1 \in \V(T_1) \}
	\end{align*}
	A valuation, written $\sigma$ or $\tau$, is a list of variable-value pairs: %
	\[
		\sigma \defeqq \; \varnothing \; \mid \; \sigma, x \mapsto v
	\]
	Like with $\Gamma$ and $\Delta$, we will generally use the letter $\sigma$ for classical data and $\tau$ for quantum data, but the two are interchangeable.
	Each typing context has a corresponding set of valuations, defined as follows:
	\[
		\V(x_1 : T_1, \ldots, x_n : T_n) \defeq \{ x_1 \mapsto v_1, \ldots, x_n \mapsto v_n \mid v_1 \in \V(T_1), \ldots, v_n \in \V(T_n) \}
	\]
	Like with typing contexts, we use a comma to denote the concatenation of valuations, sometimes mixing valuations with explicit variable-value pairs.
	For example, if $\tau_0 \in \V(\Delta_0)$ and $v \in \V(T)$ and $\tau_1 \in \V(\Delta_1)$, then $\tau_0, x \mapsto v, \tau_1 \in \V(\Delta_0, x : T, \Delta_1)$.
\end{definition}

\ifshort
In the supplemental report \cite{supplement},
\else
In Appendix~\ref{app:sublanguage},
\fi
we define a classical denotational semantics for this sublanguage.
The semantics of an expression is a partial function from valuations to values, while the semantics of a program is a partial function from values to values.
Given a purely-typed expression or program, this partial function will be injective.
Stripped of quantum constructs, Qunity thus becomes a (classical) \emph{reversible} programming language comparable to other programming languages for reversible computing.

\begin{theorem}[expressiveness of classical sublanguage]
	\label{thm:mlpi}
	Any combinator written in the reversible language $\Pi$ \cite{mlpi} can be translated into a pure program in Qunity's classical sublanguage, and any arrow computation written in the arrow metalanguage $\MLPi$ \cite{mlpi} can be translated into a mixed program in Qunity's classical sublanguage.
	These translations preserve types and semantics. 
\end{theorem}

We prove Theorem~\ref{thm:mlpi}
\ifshort
in the supplemental report.
\else
in Appendix~\ref{app:mlpi}.
\fi
We choose the language $\MLPi$ because its typing and semantics are directly comparable with Qunity's, and because a translation from a more typical \texttt{let}-based language to $\MLPi$ already exists.

Qunity's quantum semantics, defined in the next section, can be viewed as a generalization of the classical sublanguage semantics.
The advantage of this approach is that one does not have to explicitly convert between separate quantum and classical languages, as any program written in the classical sublanguage can be applied to quantum data.
\ifshort
\else
A reader less familiar with quantum computing may wish to read the classical semantics defined in Appendix~\ref{app:classical-semantics} to understand the role of reversibility before proceeding to the quantum generalization.
\fi

Where the classical semantics uses finite sets for input and output, the quantum semantics uses finite-dimensional vector spaces.
As we state more precisely and prove
\ifshort
in the supplemental report,
\else
in Appendix~\ref{app:mlpi},
\fi
the classical semantics is simply the quantum semantics restricted to the standard computational basis.

\begin{theorem}[generalization of classical semantics]
	\label{thm:classical-generalization}
	The classical semantics of any classical Qunity program coincides with its quantum semantics applied to values treated like basis states.
\end{theorem}

\subsection{Full semantics}
\label{sec:denotation}

\begin{definition}
	\label{def:typed-hilbert}
	We associate Hilbert spaces with types and contexts as follows:
	\begin{align*}
		\Hilb(\Void) &\defeq \{ 0 \} \\
		\Hilb(\Unit) &\defeq \complex \\
		\Hilb(T_0 \oplus T_1) &\defeq \Hilb(T_0) \oplus \Hilb(T_1) \\
		\Hilb(T_0 \otimes T_1) &\defeq \Hilb(T_0) \otimes \Hilb(T_1) \\
		\Hilb(x_1 : T_1, \ldots, x_n : T_n) &\defeq \Hilb(T_1) \otimes \cdots \otimes \Hilb(T_n)
	\end{align*}
	On the right-hand side above, we use the symbols $\oplus$ and $\otimes$ to refer to the usual direct sum and tensor product of Hilbert spaces \cite{advanced-linear-algebra}.
	For those unfamiliar with the direct sum, we give an overview of its important properties
	\ifshort
	in the supplemental report.
	\else
	in Appendix~\ref{app:hilbert}.
	\fi
	Each Hilbert space $\Hilb(T)$ has a canonical orthonormal basis $\{ \ket v : v \in \V(T) \}$, where the meaning of $\ket v \in \Hilb(T)$ is defined as follows:
	\begin{align*}
		\ket{\unit} &\defeq 1 \\
		\ket{\lef{T_0}{T_1}\: v_0} &\defeq \ket{v_0} \oplus 0 \\
		\ket{\rit{T_0}{T_1}\: v_1} &\defeq 0 \oplus \ket{v_1} \\
		\ket{\pair{v_0}{v_1}} &\defeq \ket{v_0} \otimes \ket{v_1}
	\end{align*}
	We can do the same for typing contexts $\Delta$, constructing an orthonormal basis $\{ \ket \tau : \tau \in \V(\Delta) \} \subset \Hilb(\Delta)$ as follows:
	\begin{alignat*}{2}
		\ket{x_1 \mapsto v_1, \ldots, x_n \mapsto v_n} &\defeq \ket{v_1} \otimes \cdots \otimes \ket{v_n}
	\end{alignat*}
\end{definition}

Using notation from linear algebra \cite{linear-algebra}, we write $\linear(\Hilb_0, \Hilb_1)$ to denote the space of linear operators from $\Hilb_0$ to $\Hilb_1$, with $\linear(\Hilb') \defeq \linear(\Hilb', \Hilb')$.

Qunity's semantics is defined by four mutually recursive functions of a valid typing judgment:
\begin{itemize}
	\item
		If $\Gamma\partition \Delta \vdash e : T$ and $\sigma \in \V(\Gamma)$, then $\msem{\sigma : \Gamma \partition \Delta \vdash e : T} \in \linear(\Hilb(\Delta), \Hilb(T))$ defines the pure semantics of expression $e$.
		We give the denotation in Figure~\ref{fig:sem-pure-expr}.
		The $\sigma$ is a sort of ``classical data,'' so the pure expression semantics may be viewed as a two-parameter function $\V(\Gamma) \times \Hilb(\Delta) \to \Hilb(T)$, linear in its second argument.
	\item
		If $\Delta \Vdash e : T$, then $\msem{\Delta \Vdash e : T} \in \linear(\linear(\Hilb(\Delta)), \linear(\Hilb(T)))$ defines the mixed semantics of expression $e$.
		We give the denotation in Figure~\ref{fig:sem-mixed-expr}.
	\item
		If $\vdash f : T \rightsquigarrow T'$, then $\msem{\vdash f : T \rightsquigarrow T'} \in \linear(\Hilb(T), \Hilb(T'))$ defines the pure semantics of program $f$.
		We give the denotation in Figure~\ref{fig:sem-pure-prog}.
	\item
		If $\vdash f : T \Rrightarrow T'$, then $\msem{\vdash f : T \Rrightarrow T'} \in \linear(\linear(\Hilb(T)), \linear(\Hilb(T')))$ defines the mixed semantics of program $f$.
		We give the denotation in Figure~\ref{fig:sem-mixed-prog}.
\end{itemize}

\begin{figure}[t]
\begin{alignat*}{3}
	\msem{\sigma : \Gamma \partition \varnothing \vdash \unit : \Unit} \ket{\varnothing} &\defeq&\;& \ket{\unit} \\
	\msem{\sigma : \Gamma \partition \varnothing \vdash x : T} \ket{\varnothing} &\defeq&& \ket{\sigma(x)} \\
	\msem{\sigma : \Gamma \partition x : T \vdash x : T} \ket{x \mapsto v} &\defeq&& \ket{v} \\
	\msem{\sigma : \Gamma \partition \Delta, \Delta_0, \Delta_1 \vdash \pair{e_0}{e_1} : T_0 \otimes T_1} \ket{\tau, \tau_0, \tau_1} &\defeq&& \msem{\sigma : \Gamma \partition \Delta, \Delta_0 \vdash e_0 : T_0} \ket{\tau, \tau_0} \\
																																																																	 &&&\otimes \msem{\sigma : \Gamma \partition \Delta, \Delta_1 \vdash e_1 : T_1} \ket{\tau, \tau_1} \\
	\msem{\sigma, \sigma' : \Gamma, \Gamma' \partition \Delta, \Delta' \vdash \cntrl{e\hspace{-2mm}}{T}{e_1 &\mapsto e_1' \\ &\cdots \\ e_n &\mapsto e_n'}{T'} \hspace{-2mm}: T'} \ket{\tau,\tau'} &\defeq&& \sum_{v \in \V(T)} \bra{v} \msem{\Gamma, \Delta \Vdash e : T}\left( \op{\sigma, \tau}{\sigma, \tau} \right) \ket v \\
																																														 &&& \cdot \sum_{j=1}^n \sum_{\sigma_j \in \V(\Gamma_j)} \bra{\sigma_j} \msem{\varnothing : \varnothing \partition \Gamma_j \vdash e_j : T}^\dagger \ket v \\
																																														 &&&\cdot \msem{\sigma, \sigma_j : \Gamma, \Gamma_j \partition \Delta,\Delta' \vdash e_j' : T'} \ket{\tau, \tau'} \\
	\msem{\sigma : \Gamma \partition \Delta \vdash f\; e : T'} \ket{\tau} &\defeq&& \msem{\vdash f : T \rightsquigarrow T'} \msem{\sigma : \Gamma \partition \Delta \vdash e : T} \ket{\tau} \\
	\msem{\pi\subcap{g}(\sigma) : \pi\subcap{g}(\Gamma) \partition \pi\subcap{d}(\Delta) \vdash e : T} \ket{\pi\subcap{d}(\tau)} &\defeq&& \msem{\sigma : \Gamma \partition \Delta \vdash e : T} \ket{\tau}
\end{alignat*}
\caption{Pure expression semantics}
\label{fig:sem-pure-expr}
\end{figure}

\begin{figure}[t]
\begin{alignat*}{2}
	\msem{\Delta \Vdash e : T}\left(\op{\tau}{\tau'}\right) &\defeq&\;& \msem{\varnothing : \varnothing \partition \Delta \vdash e : T} \op{\tau}{\tau'} \msem{\varnothing : \varnothing \partition \Delta \vdash e : T}^\dagger \\
	\msem{\Delta, \Delta_0, \Delta_1 \Vdash \pair{e_0}{e_1} : T_0 \otimes T_1}\left(\op{\tau, \tau_0, \tau_1}{\tau', \tau_0', \tau_1'}\right) &\defeq&& \msem{\Delta, \Delta_0 \Vdash e_0 : T_0}\left(\op{\tau, \tau_0}{\tau', \tau_0'}\right) \\
																																					 &&&\otimes \msem{\Delta, \Delta_1 \Vdash e_1 : T_1}\left(\op{\tau, \tau_1}{\tau', \tau_1'}\right) \\
\msem{\Delta_0,\Delta_1 \Vdash \trycatch{e_0}{e_1} : T}\left( \rho_0 \otimes \rho_1 \right) &\defeq&& \tr(\rho_1) \msem{\Delta_0 \Vdash e_0 : T}\left( \rho_0 \right) \\&&&+ (\tr(\rho_0) - \tr(\msem{\Delta_0 \Vdash e_0 : T}\left( \rho_0 \right))) \\&&&\cdot \msem{\Delta_1 \Vdash e_1 : T}\left( \rho_1 \right) \\
\msem{\Delta \Vdash f\;e : T'}(\op{\tau}{\tau'}) &\defeq&& \msem{\vdash f : T \Rrightarrow T'}\left(\msem{\Delta \Vdash e : T}(\op{\tau}{\tau'})\right) \\
\msem{\pi(\Delta) \Vdash e : T} (\op{\pi(\tau)}{\pi(\tau')}) &\defeq&& \msem{\Delta \Vdash e : T} (\op{\tau}{\tau'})
\end{alignat*}
\caption{Mixed expression semantics}
\label{fig:sem-mixed-expr}
\end{figure}

We define Qunity semantics on the standard computational basis,\footnote{The exception is \textsc{T-Try}, which is defined on all product states.} but these should be understood to be linear operators after extending by linearity.
This semantics is compositional by construction.
For example, the pure semantics of $\pair{e_0}{e_1}$ can be computed in terms of the semantics of $e_0$ and $e_1$.
The operator for each subexpression is applied to the relevant part of the input basis state, and a tensor product of the resulting states defines $\pair{e_0}{e_1}$.

The most complicated definition is the one for \texttt{ctrl}, which uses a superoperator to construct an operator.
The definition uses $\bra{v} \msem{\Gamma, \Delta \Vdash e : T}\left( \op{\sigma, \tau}{\sigma, \tau} \right) \ket v$, the probability that expression $e$ outputs $v$ given input $(\sigma, \tau)$, as seen by a classical observer.
This \emph{probability} from the mixed semantics is then used directly as an \emph{amplitude} in the resulting pure semantics. %
Note that this introduces an unavoidable source of error whenever probabilities between 0 and 1 are involved because there is no square root involved---the resulting semantics will be norm-decreasing even if the original superoperator was trace-preserving.

The definition for \texttt{try}-\texttt{catch} uses $(1 - \tr(\msem{\Delta_0 \Vdash e_0 : T}\left( \op{\tau_0}{\tau_0'} \right)))$, the probability that $e_0$ throws an exception.
If $e_0$ succeeds, then its results are used, but otherwise, the results from $e_1$ are used.

\begin{figure}[t]
\begin{align*}
	\msem{\vdash \uthree{r_\theta}{r_\phi}{r_\lambda} : \Bit \rightsquigarrow \Bit} \ket{\zero} &\defeq \cos(r_\theta / 2) \ket{\zero} + e^{i r_\phi} \sin(r_\theta / 2) \ket{\one} \\
	\msem{\vdash \uthree{r_\theta}{r_\phi}{r_\lambda} : \Bit \rightsquigarrow \Bit} \ket{\one} &\defeq - e^{i r_\lambda} \sin(r_\theta / 2) \ket{\zero}  + e^{i(r_\phi + r_\lambda)} \cos(r_\theta / 2) \ket{\one} \\
	\msem{\vdash \lef{T_0}{T_1} : T_0 \rightsquigarrow T_0 \oplus T_1}\ket{v} &\defeq \ket {\lef{T_0}{T_1}\: v} \\
	\msem{\vdash \rit{T_0}{T_1} : T_1 \rightsquigarrow T_0 \oplus T_1}\ket{v} &\defeq \ket {\rit{T_0}{T_1}\: v} \\
	\msem{\vdash \lambda e \xmapsto{T} e' : T \rightsquigarrow T'}\ket{v} &\defeq \msem{\varnothing : \varnothing \partition \Delta \vdash e' : T'} \msem{\varnothing : \varnothing \partition \Delta \vdash e : T}^\dagger \ket{v} \\
	\msem{\vdash \rphase{T}{e}{r}{r'} : T \rightsquigarrow T} \ket{v} &\defeq e^{i r} \msem{\varnothing \partition \Delta \vdash e : T} \msem{\varnothing \partition \Delta \vdash e : T}^\dagger \ket{v}
\\[-0.5em]&\qquad\qquad + e^{i r'} \left(\mathbb{I} - \msem{\varnothing \partition \Delta \vdash e : T} \msem{\varnothing \partition \Delta \vdash e : T}^\dagger\right) \ket{v}
\end{align*}
\caption{Pure program semantics}
\label{fig:sem-pure-prog}
\end{figure}

\begin{figure}[t]
\begin{align*}
	& \msem{\vdash f : T \Rrightarrow T'}(\op{v}{v'}) \defeq \msem{\vdash f : T \rightsquigarrow T'}\op{v}{v'} \msem{\vdash f : T \rightsquigarrow T'}^\dagger \\
	& \msem{\vdash \lambda e \xmapsto{{\color{gray}T}} e' : T \Rrightarrow T'}(\op{v}{v'}) \defeq \\
	&\qquad \msem{\Delta \Vdash e' : T'}\left(\tr_{\Delta_0}\left(\msem{\varnothing : \varnothing \partition \Delta, \Delta_0 \vdash e : T}^\dagger \op{v}{v'} \msem{\varnothing : \varnothing \partition \Delta, \Delta_0 \vdash e : T}\right)\right)
\end{align*}
\caption{Mixed program semantics}
\label{fig:sem-mixed-prog}
\end{figure}

\subsection{Metatheory}

We are defining semantics as recursive functions on typing judgments rather than on programs and expressions directly. Qunity's typing relation is not syntax-directed, which means that there are multiple ways to type---and give semantics to---the same expression. Thus, we must prove that the semantics is truly a function, i.e., that the different proofs of the same judgment lead to the same semantic denotation.

As an example, note the two deduction trees in Figure~\ref{fig:trees} prove the same typing judgment, so they should give rise to the same semantics.
The semantics takes a tensor product of two state vectors for \textsc{T-PurePair}, but a tensor product of two density operators for \textsc{T-MixedPair}.
This is not a problem because a tensor product can equivalently be taken before or after the conversion to density operators: $(\ket{\psi} \otimes \ket{\phi})(\bra{\psi} \otimes \bra{\phi}) = \op{\psi}{\psi} \otimes \op{\phi}{\phi}$ for any states $\ket \psi, \ket\phi$.
We prove that this kind of equivalence always holds.

\begin{figure}[t]
\[
	\inference{\inference{\varnothing \partition \Gamma_0 \vdash e_0 : T_0 \qquad \varnothing \partition \Gamma_1 \vdash e_1 : T_1}{\varnothing \partition \Gamma_0, \Gamma_1 \vdash \pair{e_0}{e_1} : T_0 \otimes T_1}[\textsc{T-PurePair}]}{\Gamma_0, \Gamma_1 \Vdash \pair{e_0}{e_1} : T_0 \otimes T_1}[\textsc{T-Mix}]
	\]
	\vspace{2mm}
	\[
		\inference{
		\inference{\varnothing \partition \Gamma_0 \vdash e_0 : T_0}{\Gamma_0 \Vdash e_0 : T_0}[\textsc{T-Mix}]
		\qquad
	  \inference{\varnothing \partition \Gamma_1 \vdash e_1 : T_1}{\Gamma_1 \Vdash e_1 : T_1}[\textsc{T-Mix}]
  	}{\Gamma_0, \Gamma_1 \Vdash \pair{e_0}{e_1} : T_0 \otimes T_1}[\textsc{T-MixedPair}]
\]
\caption{Two proofs of the same typing judgment}%
\label{fig:trees}
\end{figure}

\begin{theorem}[well-defined semantics]
	\label{thm:well-defined}
	Qunity has a well-defined semantics:
	\begin{itemize}
		\item Whenever $(\Gamma \partition \Delta \vdash e : T)$ is valid and $\sigma \in \V(\Gamma)$, the denotation $\msem{\sigma : \Gamma \partition \Delta \vdash e : T}$ is uniquely defined.
		\item Whenever $(\Delta \Vdash e : T)$ is valid, the denotation $\msem{\Delta \Vdash e : T}$ is uniquely defined.
		\item Whenever $(\vdash f : T \rightsquigarrow T')$ is valid, the denotation $\msem{\vdash f : T \rightsquigarrow T'}$ is uniquely defined.
		\item Whenever $(\vdash f : T \Rrightarrow T')$ is valid, the denotation $\msem{\vdash f : T \Rrightarrow T'}$ is uniquely defined.
	\end{itemize}
	Whenever a typing judgment has more than one proof of validity, the derived semantics is independent of the proof.
\end{theorem}

The proof of Theorem~\ref{thm:well-defined},
\ifshort
given in the supplemental report,
\else
given in Appendix~\ref{app:well-defined},
\fi
was largely inspired by the proof of Newman's lemma \cite{newmans-lemma,newmans-lemma-proof}, a standard tool for proving global confluence from local confluence when using (operational) reduction semantics.
Qunity's semantics is denotational rather than operational, but one can imagine the operational procedure of \emph{evaluating} the denotational semantics by repeatedly rewriting denotations in terms of the denotations of subexpressions.
In this view, we have a terminating sequence of rewrite rules, and global confluence is exactly what is needed to prove the semantics well-defined.
We do not use Newman's lemma directly, but the induction strategy is essentially the same, and this allows us to focus on particular cases of equivalence like the one in Figure~\ref{fig:trees}, which are easy to verify algebraically.

Not all pure Qunity programs have a norm-preserving (isometric) semantics.
Rather, programs that may throw an exception are norm-\emph{decreasing} instead.
We can characterize our semantics in terms of \emph{Kraus operators}.

\begin{definition}[Kraus operator]
	A Kraus operator is a linear operator $E$ such that $E^\dagger E \sqsubseteq I$, where ``$\sqsubseteq$'' denotes the L\"owner order \cite[p.~17]{mingsheng-ying}.
\end{definition}

Note that Kraus operators are typically defined as \emph{sets} of operators $E$ whose sum satisfies the property above. For our purposes, a single operator will suffice.

\begin{theorem}
	$\msem{\sigma : \Gamma \partition \Delta \vdash e : T}$ and $\msem{\vdash f : T \rightsquigarrow T'}$ are Kraus operators whenever well-defined.
	$\msem{\Delta \Vdash e : T}$ and $\msem{\vdash f : T \Rrightarrow T'}$ are completely positive, trace non-increasing superoperators whenever well-defined.
\end{theorem}

We do not prove this theorem directly, but it is a direct consequence of the correctness of our compilation procedure discussed in Section~\ref{sec:compilation}.

We can also be more precise about the semantics of the orthogonality and spanning judgments.
In the absence of variables, the orthogonality judgment describes a set of orthogonal basis states, and the spanning judgment describes a spanning set of orthogonal basis states.
With variables, the picture is a bit more complicated, because we are dealing with projectors onto orthogonal \emph{subspaces} rather than orthogonal states.
An algebraic description of these semantics is given by Lemmas~\ref{lem:ortho-sem} and \ref{lem:spanning-sem}.
\ifshort
\else
We do not have a nice semantic representation of the erases judgment except for the one used by the compiler in Lemma~\ref{lem:erasure}.
\fi

\begin{lemma}[spanning semantics]
	\label{lem:spanning-sem}
	Suppose $\varnothing \partition \Delta_j \vdash e_j : T$ for all $j \in \{1, \ldots, n\}$ and suppose $\spanning{T}{e_1, \ldots, e_n}$ is true.
	Then $\sum_{j = 1}^n \msem{\varnothing \partition \Delta_j \vdash e_j : T}\msem{\varnothing \partition \Delta_j \vdash e_j : T}^\dagger = \mathbb{I}$.
\end{lemma}

\begin{lemma}[orthogonality semantics]
	\label{lem:ortho-sem}
	Suppose $\varnothing \partition \Delta_j \vdash e_j : T$ for all $j \in \{1, \ldots, n\}$ and suppose $\ortho{T}{e_1, \ldots, e_n}$ is true.
	Then $\bra{v} \msem{\varnothing \partition \Delta_j \vdash e_j : T} \ket{\tau_j} \in \{0, 1\}$ for all $\tau_j \in \V(\Delta_j), v \in \V(T)$, and $\sum_{j = 1}^n \msem{\varnothing \partition \Delta_j \vdash e_j : T}\msem{\varnothing \partition \Delta_j \vdash e_j : T}^\dagger \sqsubseteq \mathbb{I}$.
\end{lemma}

Finally, Qunity's semantics uses norm-decreasing operators and trace-decreasing superoperators, which can be interpreted operationally as a sort of ``exception.''
\ifshort
In the supplemental report,
\else
In Appendix~\ref{app:iso},
\fi
we present an additional judgment ``iso'' that can be used to statically determine whether a program belongs to a checkable class of exception-free programs.

\section{Examples}
\label{sec:examples}

We have already shown two examples of programs we can write in Qunity: Deutsch's algorithm (Section~\ref{sec:deutsch}) and Grover's algorithm (end of Section~\ref{sec:syntax}). This section presents two more---the quantum Fourier transform and the quantum walk---aiming to further illustrate Qunity's expressiveness.
The quantum walk depends on \emph{specialized erasure}, a technique that reverses \texttt{ctrl} expressions to implement a general form of reversible pattern-matching.
We give a third example---the Deutsch-Jozsa algorithm---in
\ifshort
the supplemental report.
\else
Appendix~\ref{sec:dj}.
\fi

\subsection{Quantum Fourier transform}

\begin{figure}[ht]
\begin{alignat*}{3}
	\texttt{and} &\defeq&& \lambda x \xmapsto{{\color{gray}\Bit\otimes\Bit}} \\ &&&\cntrl{x}{\Bit\otimes\Bit}{\pair{\zero}{\zero} &\mapsto \pair{x}{\zero} \\ \pair{\zero}{\one} &\mapsto \pair{x}{\zero} \\ \pair{\one}{\zero} &\mapsto \pair{x}{\zero} \\ \pair{\one}{\one} &\mapsto \pair{x}{\one}}{(\Bit\otimes\Bit)\otimes\Bit} \triangleright \snd{(\Bit\otimes\Bit)}{\Bit} \\
	\texttt{couple}(k) &\defeq&\;& \lambda \pair{x_0}{x_1} \xmapsto{{\color{gray}\Bit\otimes\Bit}}  \\
										 &&&\cntrl{\texttt{and} \pair{x_0}{x_1}}{\Bit}{\zero &\mapsto \pair{x_1}{x_0} \\ \one &\mapsto \pair{x_1}{x_0} \triangleright \gphase{\Bit\otimes\Bit}{2 \pi \texttt{ / } 2^\texttt{k}}}{\Bit\otimes\Bit} \\
  \texttt{rotations}(0) &\defeq&& \lambda \unit \xmapsto{{\color{gray}\Bit^{\otimes 0}}} \unit \\
  \texttt{rotations}(1) &\defeq&& \lambda \pair{x}{\unit} \xmapsto{{\color{gray}\Bit^{\otimes 1}}} \pair{\had\:x}{\unit} \\
	\texttt{rotations}(n+2) &\defeq&& \lambda \pair{x_0}{x} \xmapsto{{\color{gray}\Bit^{\otimes (n+2)}}} \\
													&&&\left( \begin{aligned} &\texttt{let } \pair{x_0}{\pair{y_0'}{y}} \texttt{ =}_{\color{gray}\Bit^{\otimes(n+2)}} \pair{x_0}{x \triangleright \texttt{rotations}(n+1)} \texttt{ in} \\ &\texttt{let } \pair{\pair{y_0}{y_1}}{y} \texttt{ =}_{\color{gray}(\Bit \otimes \Bit)\otimes \Bit^{\otimes n}} \pair{\pair{x_0}{y_0'} \triangleright \texttt{couple}(n+2)}{y} \texttt{ in} \\ &\pair{y_0}{\pair{y_1}{y}} \end{aligned} \right) \\
	\texttt{qft}(0) &\defeq&& \lambda \unit \xmapsto{{\color{gray}\Bit^{\otimes 0}}} \unit \\
	\texttt{qft}(n+1) &\defeq&& \lambda x \xmapsto{{\color{gray}\Bit^{\otimes(n+1)}}} \texttt{let } \pair{x_0}{x'} \texttt{ =}_{\color{gray}\Bit^{\otimes n}} x \triangleright \texttt{rotations}(n+1) \texttt{ in } \pair{x_0}{x' \triangleright \texttt{qft}(n)}
\end{alignat*}
\caption{The quantum Fourier transform implemented in Qunity}
\label{fig:qft}
\end{figure}

In Figure~\ref{fig:qft}, we use a presentation of the quantum Fourier transform that has a symmetric circuit diagram \cite{semiclassical-fourier}.
This version uses a two-qubit ``coupling'' gate that swaps two qubits and coherently induces a particular global phase conditional on both qubits being $\ket{\one}$.

\subsection{Specialized erasure}
\label{sec:match}

Quantum control---that is, programming with conditionals while maintaining quantum coherence---is a common feature of quantum algorithms, and Qunity's \texttt{ctrl} construct can be a powerful tool for implementing this pattern.
However, the ``erases'' requirement of the \textsc{T-Ctrl} typing rule can be a frustrating limitation on practical quantum control.
It effectively requires that any part of the input used for quantum control (the $\Delta$ context in the \textsc{T-Ctrl} typing rule) must be present in the output as well.
We show here how programmers can get around this limitation with a general approach for ``erasing'' controlled data in a more customizable way, without adding any new primitives to the Qunity language.

The basic approach is fairly simple, and it is already used in the implementation of existing quantum algorithms \cite[p.~8]{childs-van-dam-review}.
Suppose one has two Hilbert spaces $\Hilb$ and $\Hilb'$ with orthonormal sets $\{\ket{1}, \ldots, \ket{n}\} \subset \Hilb$ and $\{\ket{1'}, \ldots, \ket{n'}\} \subset \Hilb'$, and one would like to implement the operator $E \defeq \sum_{j=1}^n \op{j'}{j} \in \linear(\Hilb, \Hilb')$.
If one can implement operators $E_1 \defeq \sum_{j=1}^n \op{j, j'}{j} \in \linear(\Hilb, \Hilb \otimes \Hilb')$ and $E_2 \defeq \sum_{j=1}^n \op{j, j'}{j'} \in \linear(\Hilb', \Hilb \otimes \Hilb')$, then one can implement $E = E_2^\dagger E_1$.
This pattern can be useful in Qunity, where the operators $E_1$ and $E_2$ can be easy to implement using a \texttt{ctrl} expression, but the erases judgment prevents $E$ from being implemented more directly.
The $E_2^\dagger$ operator serves as a sort of ``specialized erasure,'' erasing the input state $\ket{j}$.

This pattern can be used to implement a ``\texttt{match}'' program, shown in Figure~\ref{fig:match}.
The typing and semantics of this program are very similar to the symmetric pattern matching language \cite{symmetric-pattern-matching}.
In particular, the typing rule and semantics in Figure~\ref{fig:match-type} can be inferred, where ``$\text{classical}(e)$'' means that $e$ includes no \texttt{u3} or \texttt{rphase}.

\begin{figure}[t]
\begin{alignat*}{3}
	\match{T}{e_1 &\mapsto e_1' \\ &\cdots \\ e_n &\mapsto e_n'}{T'} &\defeq&& \lambda x \xmapsto{{\color{gray}T}} \cntrl{x}{T}{e_1 &\mapsto \pair{x}{e_1'} \\ &\cdots \\ e_n &\mapsto \pair{x}{e_n'}}{T \otimes T'} \\ &&&\triangleright \lambda \left( \cntrl{x'}{T'}{e_1' &\mapsto \pair{e_1}{x'} \\ &\cdots \\ e_n' &\mapsto \pair{e_1}{x'}}{T \otimes T'} \right) \xmapsto{{\color{gray}T \otimes T'}} x'
\end{alignat*}
\caption{Reversible pattern matching via specialized erasure}
\label{fig:match}
\end{figure}

\begin{figure}[t]
\[
	\inference{\text{classical}(e_j) \text{ for all } j \qquad \ortho{T}{e_1, \ldots, e_n} \qquad \varnothing\partition \Delta_j \vdash e_j : T \text{ for all } j \\ \text{classical}(e_j') \text{ for all } j \qquad \ortho{T'}{e_1', \ldots, e_n'} \qquad \varnothing \partition \Delta_j \vdash e_j' : T' \text{ for all } j}{\vdash \match{T}{e_1 &\mapsto e_1' \\ &\cdots \\ e_n &\mapsto e_n'}{T'} : T \rightsquigarrow T'}
\]
\[
	\msem{\vdash \match{T}{e_1 &\mapsto e_1' \\ &\cdots \\ e_n &\mapsto e_n'}{T'} : T \rightsquigarrow T'} = \sum_{j=1}^n \msem{\varnothing \partition \Delta_j \vdash e_j' : T'} \msem{\varnothing\partition \Delta_j \vdash e_j : T}^\dagger
\]
\caption{Typing and semantics for the \texttt{match} construct}
\label{fig:match-type}
\end{figure}

As another example, Figure~\ref{fig:dsum} shows how one can use this pattern to implement the direct sum of linear operators, defined so that
$\msem{\vdash f_0 \oplus f_1 : T_0 \oplus T_1 \rightsquigarrow T_0' \oplus T_1'} = \msem{\vdash f_0 : T_0 \rightsquigarrow T_0'} \oplus \msem{\vdash f_1 : T_1 \rightsquigarrow T_1'}$.

\begin{figure}[t]
\begin{alignat*}{3}
	\tagsum{T_0}{T_1} &\defeq&& \lambda x \xmapsto{{\color{gray}T_0 \oplus T_1}} \cntrl{x}{T_0 \oplus T_1}{\lef{T_0}{T_1}\: x_0 &\mapsto \pair{\zero}{x} \\ \rit{T_0}{T_1}\:x_1 &\mapsto \pair{\one}{x}}{\Bit \otimes (T_0 \oplus T_1)} \\
	f_0 \oplus f_1 &\defeq&& \lambda x \xmapsto{{\color{gray}T_0 \oplus T_1}} \\
								 &&&\texttt{let } \pair{x_i}{x} \texttt{ =}_{{\color{gray}\Bit\otimes(T_0\oplus T_1)}} \tagsum{T_0}{T_1} x \;\texttt{ in} \\
								 &&& \cntrl{x_i}{\Bit}{\zero &\mapsto \pair{x_i}{x \triangleright \lef{T_0}{T_1}^\dagger \triangleright f_0} \\ \one &\mapsto \pair{x_i}{x \triangleright \rit{T_0}{T_1}^\dagger \triangleright f_1}}{\Bit \otimes (T_0' \oplus T_1')} \\
								 &&& \triangleright \tagsum{T_0'}{T_1'}^\dagger
\end{alignat*}
\caption{A Qunity implementation of the direct sum of linear operators}
\label{fig:dsum}
\end{figure}

These examples are admittedly fairly lengthy, and would not lead to an efficient implementation if compiled using the compilation procedure described in Section~\ref{sec:compilation} without optimizations.
Our purpose here is to demonstrate the expressiveness of Qunity despite its small number of language primitives, but additional language primitives may make efficient compilation easier.

\subsection{Quantum walk}
\label{sec:quantum-walk}

We use Qunity to implement a quantum walk algorithm for boolean formula evaluation \cite{booleanformula2007,booleanformula2010}.
The simplest version of this algorithm treats a \textsc{nand} formula as a balanced binary tree where each leaf corresponds to a variable in the formula and each vertex corresponds to a \textsc{nand} application.
Given black-box oracle access to a function $f : \texttt{Variable} \to \{\zero,\one\}$, the task is to evaluate the formula.

This algorithm performs a quantum walk, repeatedly applying a \emph{diffusion step} and a \emph{walk step} to two quantum registers: a vertex index and a qutrit ``three-sided coin.'' %
The diffusion step uses coherent control to apply a different operator to the coin depending on the vertex index.
If the vertex index is a leaf $v$, then the oracle is used to induce a conditional phase flip of $(-1)^{f(v)}$.
Otherwise, a reflection operator $(2\op{\texttt{u}}{\texttt{u}} - I)$ is applied to the coin register, where $\ket{\texttt{u}}$ is a coin state whose value depends on whether the vertex index is one of two special root nodes.
The walk step then performs a coherent permutation on the two registers, ``walking'' the vertex index to the direction specified by the coin and setting the coin to be the direction traveled from.

This algorithm presents some interesting challenges for implementation, and Quipper's \cite{quipper} implementation of this algorithm is quite long.\footnote{See \href{https://www.mathstat.dal.ca/~selinger/quipper/doc/src/Quipper/Algorithms/BF/BooleanFormula.html}{\texttt{Quipper.Algorithms.BF.BooleanFormula}} for reference.} %
One challenge is the representation of the vertex index.
The traversed graph is a tree, so a convenient representation is the path taken to get to this vertex from the root of the tree, a list like [\texttt{vleft}, \texttt{vright}, \ldots] indicating the sequence of child directions.
The problem is that this list has an unknown length, so one cannot simply use an array of qubits where each qubit corresponds to a direction taken.
Quipper has no sum types, and every (quantum) Quipper type is effectively a fixed-length array of qubits, so Quipper's implementation must directly manipulate an encoding of variable-length lists into bitstrings, which requires developers to work at a lower level of abstraction than they would prefer.
Qunity's sum types make it convenient to coherently manipulate variable-length lists by managing the qubit encoding automatically. %

First, we must define the types and values used in this program in terms of existing ones.
We use a recursively-defined \texttt{Vertex} type to represent a vertex in the tree as a variable-length list.
This type and others are defined in Figure~\ref{fig:walk-sugar} in terms of Qunity's base types.
Here, the parameter $n$ is the height of the tree, a bound on the depth.
A $\texttt{Vertex}_{n+1}$ is thus either empty or a $\texttt{Vertex}_n$ with an additional $\Child$ appended and can store a superposition of lists of different lengths.
The algorithm augments the tree with two special vertices at the root of the tree, $\rootone{n}$ and $\rooot{n+1}$, both of type $\Vertex_{n+2}$, where the ``$+2$'' comes from the extra depth incurred by these vertices.
The leaves of the traversed tree are those where the path from the root is maximal, so we can use the fixed-length \texttt{Leaf} type to describe leaves separately from arbitrary vertices.

\begin{figure}[ht]
    \centering
  \begin{minipage}{.35\textwidth}

\begin{align*}
	\Child &\defeq \Bit \\
	\texttt{vleft} &\defeq \zero \\
	\texttt{vright} &\defeq \one \\
	\Coin &\defeq \Maybe(\Child) \\
	\texttt{cdown} &\defeq \nothing{\Child} \\
	\texttt{cleft} &\defeq \just{\Child} \vleft \\
	\texttt{cright} &\defeq \just{\Child} \vright 
\end{align*}
  \end{minipage}
  \hspace{.1in}
  \begin{minipage}{.4\textwidth}
	\begin{align*}
  \Vertex_0 &\defeq \Void \\
  \Vertex_{n+1} &\defeq \Unit \oplus (\Vertex_n \otimes \Child) \\
	\rooot{n} &\defeq \texttt{left}_{\color{gray}\Vertex_{n+1}} \unit \\
	(e \consarrow{n} e_0) &\defeq \texttt{right}_{\color{gray}\Vertex_{n+1}} \pair{e}{e_0} \\
	\rootone{n} &\defeq \rooot{n} \consarrow{n+1} \vleft \\
	\Leaf_n &\defeq \Child^{\otimes n} \\
\end{align*}
  \end{minipage}
	\caption{Types and values used in the quantum walk algorithm (also see Fig.~\ref{fig:notation})}%
\label{fig:walk-sugar}
\end{figure}

\begin{figure}[ht]
\begin{alignat*}{3}
	\asleaf{0} &\defeq&\;& \lambda \; \rootone{0} \xmapsto{{\color{gray}{\Vertex_2}}} \unit \\
	\asleaf{n+1} &\defeq&& \lambda (x \consarrow{n+2} x_0) \xmapsto{{\color{gray}{\Vertex_{n+3}}}} \pair{x_0}{x \triangleright \asleaf{n}} \\
	\downcast{0} &\defeq&& \lambda v \xmapsto{{\color{gray}{\Vertex_1}}} \cntrl{v}{\Vertex_1}{}{\Vertex_0} \\
	\downcast{n+1} &\defeq&& \lambda v \xmapsto{{\color{gray}{\Vertex_{n+2}}}} \\
								 &&& \cntrl{v\hspace{-1cm}}{\Vertex_{n+2}}{\rooot{n+1} &\mapsto \pair{v}{\rooot{n}} \\ v' \consarrow{n+1} x &\mapsto \pair{v}{\downcast{n} v' \consarrow{n} x}}{\Vertex_{n+2} \otimes \Vertex_{n+1}} \\
								 &&& \triangleright \lambda \quad \cntrl{v\hspace{-1cm}}{\Vertex_{n+1}}{\rooot{n} &\mapsto \pair{\rooot{n+1}}{v} \\ v' \consarrow{n} x &\mapsto \pair{\downcast{n}^\dagger v' \consarrow{n+1} x}{v}}{\Vertex_{n+2} \otimes \Vertex_{n+1}} \hspace{-23mm} \xmapsto{{\color{gray}{\Vertex_{n+2} \otimes \Vertex_{n+1}}}} v \\
	\leftchild{n} &\defeq&& \lambda v \xmapsto{{\color{gray}{\Vertex_{n}}}} \downcast{n} ( v \consarrow{n} \vleft ) \\
	\rightchild{n} &\defeq&& \lambda v \xmapsto{{\color{gray}{\Vertex_{n}}}} \downcast{n} ( v \consarrow{n} \vright ) \\
\end{alignat*}
\caption{Programs and expressions used in the quantum walk implementation}
\label{fig:walk-subexpressions}
\end{figure}

Figure~\ref{fig:walk-subexpressions} defines some additional programs used in this algorithm.
The \texttt{asleaf} program is essentially a projector onto the subspace of vertices spanned by leaves; 
$\asleaf{n}$ has type $(\Vertex_{n+2} \rightsquigarrow \Leaf_n)$. %
This program converts the variable-length \texttt{Vertex} type into the fixed-length \texttt{Leaf} type, terminating exceptionally if the length is insufficient.
The $\downcast{n}$ program has type $(\Vertex_{n+1} \rightsquigarrow \Vertex_n)$, and is used to convert a vertex to a type with a smaller maximum length, terminating exceptionally if the length is maximal.
This program, defined using the specialized erasure technique described in the previous section, is used by the $\leftchild{n}$ and $\rightchild{n}$ programs, both of type $(\Vertex_n \rightsquigarrow \Vertex_n)$, which reversibly ``step'' a vertex to a child in the tree.

We implement the diffusion step of the quantum walk in Figure~\ref{fig:diffusion}.
This implementation resembles the original algorithm description \cite[p.~4]{booleanformula2007} much more closely than Quipper's implementation, which involves more complicated qubit encodings and explicit uncomputation.
The diffusion program has the following type: %
\[
	\inference{\vdash f : \Leaf_n \Rrightarrow \Bit}{\vdash \diffusion{n}(f) : \Coin \otimes \Vertex_{n+2} \rightsquigarrow \Coin \otimes \Vertex_{n+2}}
\]

\begin{figure}[ht]
\begin{align*}
	&\diffusion{n}(f) \defeq \\
	&\lambda \pair{c}{v} \xmapsto{{\color{gray}\Coin \otimes \Vertex_{n+2}}} \\
	&\texttt{ctrl } v \\
	&\left\{ \begin{aligned} v' \consarrow{n} x \consarrow{n+1} x' \mapsto \; &\texttt{ctrl} \left( \begin{aligned} & \texttt{try}\; \just{\Leaf(n)} ((v' \consarrow{n} x \consarrow{n+1} x') & \triangleright \asleaf{n}) \\ & \texttt{catch}\; \nothing{\Leaf(n)} \end{aligned} \right) \\ &\left\{ \begin{aligned} \nothing{\Leaf(n)} &\mapsto \pair{c \triangleright \reflect{\Coin}(\texttt u)}{v} \\ \just{\Leaf(n)} \ell &\mapsto \texttt{ctrl } (f\;\ell) \left\{ \begin{aligned} \zero &\mapsto \pair{c}{v} \\ \one &\mapsto \pair{c}{v} \triangleright \gphase{}{\pi} \end{aligned} \right\} \end{aligned} \right\} \\ \rootone{n} \mapsto \; &\pair{c \triangleright \reflect{\Coin}(\texttt{u\textquotesingle}_{\color{gray}n})}{v} \\
\rooot{n+1} \mapsto \; &\pair{c}{v} \end{aligned} \right\}
	\end{align*}
	\caption{The diffusion step of the boolean formula algorithm. For brevity, we omit some type annotations and the implementation of the expressions \texttt{u} and $\texttt{u\textquotesingle}_{\color{gray}n}$. See the Coq implementation \cite{typechecker} for details.}
	\label{fig:diffusion}
	\end{figure}

We implement the walk step in Figure~\ref{fig:walk}.
This program must update both the vertex and the coin while maintaining quantum coherence.
To make this work, we use a program $\nextcoin{n}$ of type $(\Coin \otimes \Vertex_{n+1} \rightsquigarrow (\Coin \otimes \Vertex_{n+1}) \otimes \Coin)$ to both compute the updated coin and to uncompute the previous coin.
The walk program $\walk{n}$ then has type $(\Coin \otimes \Vertex_{n+1} \rightsquigarrow \Coin \otimes \Vertex_{n+1})$.

\begin{figure}[ht]
	\begin{alignat*}{2}
		\nextcoin{n} &\defeq&\;& \lambda x \mapsto \cntrl{x\hspace{-17mm}}{\Coin \otimes \Vertex_{n+1}}{\pair{\cdown}{v \consarrow{n} \vleft} &\mapsto \pair{x}{\cleft} \\ \pair{\cdown}{v \consarrow{n} \vright} &\mapsto \pair{x}{\cleft} \\ \pair{\cright}{v} &\mapsto \pair{x}{\cdown} \\ \pair{\cleft}{v} &\mapsto \pair{x}{\cdown}}{(\Coin \otimes \Vertex_{n+1}) \otimes \Coin} \\[2mm]
	\walk{n} &\defeq&& \lambda x \xmapsto{{\color{gray} \Coin \otimes \Vertex_{n+1}}} \\
					 &&& \nextcoin{n} x \\
					 &&& \triangleright \lambda \pair{\pair{c}{v}}{c'} \xmapsto{{\color{gray} (\Coin \otimes \Vertex_{n+1}) \otimes \Coin}} \\
					 &&& \cntrl{\pair{c}{c'}}{}{\pair{\cdown}{\cleft} &\mapsto \pair{\pair{c'}{\leftchild{n+1}^\dagger v}}{c} \\ \pair{\cdown}{\cright} &\mapsto \pair{\pair{c'}{\rightchild{n+1}^\dagger v}}{c} \\ \pair{\cleft}{\cdown} &\mapsto \pair{\pair{c'}{\leftchild{n+1} v}}{c} \\ \pair{\cright}{\cdown} &\mapsto \pair{\pair{c'}{\rightchild{n+1} v}}{c}}{} \\
					 &&& \triangleright \nextcoin{n}^\dagger
	\end{alignat*}
		\caption{The walk step of the boolean formula algorithm}
		\label{fig:walk}
\end{figure}

\section{Compilation}
\label{sec:compilation}

We have developed an algorithm for compiling Qunity to a low-level qubit circuit language such as Open\textsc{Qasm} \cite{openqasm3}. This section provides an overview of the algorithm; the details are given in
\ifshort
the supplemental report \cite{supplement}.
\else
Appendix~\ref{app:compilation}.
\fi

Our compilation algorithm serves two purposes.
First, it makes clear that Qunity is realizable on a quantum computer. Realizability is not immediately clear from the definitions of Qunity's semantics; some quantum languages like Lineal \cite{lineal} and the \textsc{zx}-calculus \cite{zx} allow for programs to be written with norm-\emph{increasing} semantics and thus have no physical interpretation.
Second, it gives an intuition on what is operationally happening step-by-step in the physical computer and answers questions about data representation, automatic uncomputation, and the interplay between Qunity's pure and mixed modes of computation.

There are some obvious challenges in compiling Qunity to low-level circuits.
In particular, the \texttt{ctrl} construct can make irreversible programs reversible, and the \texttt{try}-\texttt{catch} construct can make trace-decreasing programs trace-preserving.
We solve these problems by augmenting the circuit with ancillary ``garbage'' and ``flag'' qubits, and then compiling it into a unitary circuit.
The garbage qubits are used to store discarded information, taking advantage of the deferred measurement principle \cite[p.~186]{mike-and-ike}, and the \texttt{ctrl} expression performs a reverse computation to ``uncompute'' this garbage and produce a reversible circuit. 
The flag qubits are used as a kind of ``assertion,'' where a nonzero flag qubit corresponds to an error that can be caught by the \texttt{try}-\texttt{catch} expression.
Our main result,
\ifshort
proven in the supplemental report,
\else
proven in Appendix~\ref{app:compilation},
\fi
is summarized by the following theorem:
\begin{theorem}
	There is a recursive procedure that compiles any well-typed Qunity program to a qubit-based circuit consisting only of single-qubit gates and controlled operators.
	The unitary semantics of this low-level circuit will correspond closely to the (potentially norm-decreasing) semantics of the Qunity program, as made precise by Definitions~\ref{def:implement-kraus} and \ref{def:implement-superop} in Section~\ref{sec:hl2ll} below.
\end{theorem}

\subsection{Overview}

Compiling Qunity programs to qubit circuits happens in two stages: from Qunity programs to high-level circuits, and from high-level circuits to low-level qubit circuits.
More precisely:

The \textbf{input} is a valid typing judgment, as defined in Section~\ref{sec:types}, for either an expression or program, pure or mixed.
		We assume that the compiler has access to the proof of validity for the typing judgment, and our compiler (like our semantics) is defined as a recursive function of this proof.

		An \textbf{intermediate representation} uses quantum circuits of the sort found in most quantum computing textbooks, with two notable exceptions:
		\begin{enumerate}
			\item The wires in the circuit correspond to the Hilbert spaces defined in Definition~\ref{def:typed-hilbert}, rather than single qubits.
				In our circuit diagrams, we will label wires with the Hilbert space they represent and a ``slash'' pointing to the labeled wire.
				This is useful because it is easier to analyze circuit semantics if we can algebraically manipulate vectors in the direct sum of Hilbert spaces rather than the corresponding qubit encodings of these vectors.
				Using wires to represent the direct sum of Hilbert spaces is unconventional, but not unheard of in graphical quantum languages \cite{many-worlds-calculus}.
			\item
				The boxes in the circuit correspond to norm-non-increasing operators and trace-non-increasing superoperators rather than unitary operators and trace-preserving superoperators (\textsc{cptp} maps).
				Again, this is not unheard of: Fault-tolerant quantum circuits are often drawn and analyzed with norm-decreasing ``$r$-filter'' projectors \cite{qecc}.
		\end{enumerate}
		Some of our circuits will be ``pure,'' involving no measurement and described by norm-non-increasing linear operators.
		Others will be impure, potentially involving measurement and described by trace-non-increasing superoperators.
		In both cases, a vertical stacking of boxes represents a tensor product, a horizontal stacking of boxes represents function composition, and a bare wire represents the identity.
		Whenever we use a pure component described by the operator $E$ within an impure circuit, it has superoperator semantics $\rho \mapsto E\rho E^\dagger$.

		The \textbf{output} is a low-level qubit-based unitary quantum circuit of the sort standard in quantum computing literature.
		We provide circuit diagrams, and it should be obvious how to implement these circuits in a runnable quantum assembly language such as Open\textsc{Qasm} \cite{openqasm3}.
		Here, wire labels indicate the number of qubits in a particular register, and unlabelled wires can be assumed to represent a single qubit.
		All of these circuits will be ``pure'' (unitary and measurement-free), and we assume that the controlled gate $\op{\zero}{\zero} \otimes \mathbb{I} + \op{\one}{\one} \otimes U$ is implementable whenever the unitary $U$ can be, as this sort of quantum control is a primitive ``gate modifier'' in Open\textsc{Qasm}.

In what follows, we will refer to our intermediate circuit representation as ``high-level circuits'' and the target qubit language as ``low-level circuits.''
High-level circuits semantically described by superoperators will be referred to as ``decoherence-based,'' while those that are measurement-free will be referred to as ``pure.''
Section~\ref{sec:qunity2hl} outlines the compilation from Qunity programs to high-level circuits, and Section~\ref{sec:hl2ll} outlines the compilation of these high-level circuits to low-level circuits.

\subsection{Qunity Programs to High-Level Circuits}
\label{sec:qunity2hl}

Our compiler, like our denotational semantics, is a recursive function of a valid typing judgment.
Sometimes we do not depend on the compiled subcircuit directly, but rather a transformed version of this circuit; we make it clear when we do this and show how to implement these circuit transformers in
\ifshort
the supplemental report.
\else
Appendix~\ref{app:ll-hl}.
\fi
For each kind of judgment, we describe below the correspondence between the Qunity semantics and the high-level circuit semantics and give an example of a compiled circuit. 

\paragraph*{Pure expressions}
		Given a judgment $(\Gamma \partition \Delta \vdash e : T)$, the compiler produces a pure circuit with input space $\Hilb(\Gamma) \otimes \Hilb(\Delta)$ and output space $\Hilb(\Gamma) \otimes \Hilb(T)$.
		The semantics of this circuit $C$ corresponds to Qunity expression semantics in the following way:
		\[
			\bra{\sigma, v} C \ket{\sigma, \tau} = \bra{v} \msem{\sigma : \Gamma \partition \Delta \vdash e : T} \ket{\tau}
		\]
		Recall that the typing context $\Gamma \partition \Delta$ is partitioned in two: left of the $\partition$ is the ``classical'' context $\Gamma$ and right of it is the ``quantum'' context $\Delta$.
		The semantics makes the distinction clear -- $\Gamma$ determines the classical parameter to the denotation, while $\Delta$ determines input Hilbert space.
		In the compilation setting, the story is slightly different.
		Because even classical computation may need to be done reversibly on quantum data, our compiler uses $\Hilb(\Gamma)$ as an auxiliary Hilbert space holding the variables used in a classical way.
		Practically, ``used in a classical way'' means that the wires in this ancilla register are used only as ``control'' wires, so any data on this register remains unchanged in the standard basis.
		For this reason, we depict these gates in circuit diagrams with a control on the classical $\Gamma$ register:
		
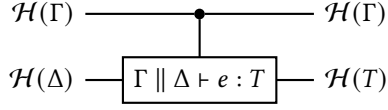
\[
		\begin{quantikz}
			\lstick{$\Hilb(\Gamma)$} & \ctrl{1} & \rstick{$\Hilb(\Gamma)$} \qw \\
			\lstick{$\Hilb(\Delta)$} & \gate{{\Gamma \partition \Delta \vdash e : T}} & \rstick{$\Hilb(T)$} \qw
		\end{quantikz}
		\]

		Seven inference rules define the pure expression typing relation, and in
\ifshort
                the supplemental report
\else
                Appendix~\ref{app:compilation}
\fi
                we give a compiled circuit for each of these cases.
		As an example, Figure~\ref{fig:compile-ex-purepair} shows the circuit for the \textsc{T-PurePair} typing rule, which includes as subcircuits the compiled subexpressions.
		It uses a ``share'' gate graphically represented by the ``controlled cloud'' component that copies quantum data in the standard basis, mapping $\ket{\tau} \mapsto \ket{\tau, \tau}$.
		We show how to implement this gate later, in Figure~\ref{fig:compile-share}. %

\begin{figure}[ht]
\[
	\inference{\Gamma\partition \Delta, \Delta_0 \vdash e_0 : T_0 \qquad \Gamma\partition \Delta, \Delta_1 \vdash e_1 : T_1}{\Gamma\partition \Delta, \Delta_0, \Delta_1 \vdash \pair {e_0} {e_1} : T_0 \otimes T_1}[\textsc{T-PurePair}]
\]
\medskip
\[
\begin{quantikz}
	\lstick{$\Hilb(\Gamma)$} & \qw & \qw & \ctrl{1} & \ctrl{3} & \qw \rstick{$\Hilb(\Gamma)$} \\
	\lstick{$\Hilb(\Delta)$} & \ctrl{2} & \qwbundle{\Hilb(\Delta)} & \gate[2]{{\Gamma \partition \Delta, \Delta_0 \vdash e_0 : T_0}} & \qw & \qw \rstick{$\Hilb(T_0)$} \\
	\lstick{$\Hilb(\Delta_0)$} & \qw & \qw & & \\
														 & \gate[style={cloud},nwires=1]{} & \qwbundle{\Hilb(\Delta)} & \qw & \gate[2]{{\Gamma \partition \Delta, \Delta_1 \vdash e_1 : T_1}} & \qw \rstick{$\Hilb(T_1)$} \\
	\lstick{$\Hilb(\Delta_1)$} & \qw & \qw & \qw & & 
\end{quantikz}
\]
\caption{\textsc{T-PurePair} compilation}
	\label{fig:compile-ex-purepair}
\end{figure}

\paragraph*{Pure programs}
		Given a judgment $(\vdash f : T \rightsquigarrow T')$, the compiler produces a pure circuit with input space $\Hilb(T)$ and output space $\Hilb(T')$.
		This circuit's semantics is the same as $\msem{\vdash f : T \rightsquigarrow T'}$.
		As an example, Figure~\ref{fig:compile-ex-pureabs} shows the circuit for the \textsc{T-PureAbs} typing rule.
		Note that this circuit does not use the compiled circuit for $(\varnothing \partition \Delta \vdash e : T)$ directly; rather, we use a transformed version: its \emph{adjoint}. %
		As described in Section~\ref{sec:hl2ll}, norm-decreasing operators in our high-level circuits are implemented by unitary operators in our low-level circuits by treating some of the wires as ``prep wires'' (initialized to zero) and some of the wires as ``flag wires'' (asserted to be zero upon termination).
		Whenever a norm-decreasing operator can be implemented, its adjoint can also be implemented by taking the adjoint of the underlying unitary and swapping the prep and flag wires, reversing the circuit and swapping the processes of qubit initialization and qubit termination.
		We describe this circuit transformer (and others) in more detail in
\ifshort
                the supplemental report.
\else
                Appendix~\ref{app:ll-hl}.
\fi

                                \begin{figure}[ht]
\[
	\inference{\varnothing\partition \Delta \vdash e : T \qquad \varnothing\partition \Delta \vdash e' : T'}{\vdash \lambda e \xmapsto{{\color{gray}T}} e' : T \rightsquigarrow T'}[\textsc{T-PureAbs}]
\]
\medskip
\[
\begin{quantikz}
	\lstick{$\Hilb(T)$} & \gate{{\msem{\varnothing \partition \Delta \vdash e : T}^\dagger}} & \qwbundle{\Hilb(\Delta)} & \gate{{\varnothing \partition \Delta \vdash e' : T'}} &\rstick{$\Hilb(T')$} \qw
\end{quantikz}
\]
\caption{\textsc{T-PureAbs} compilation}
	\label{fig:compile-ex-pureabs}
\end{figure}

\paragraph*{Mixed expressions}                
		Given a judgment $(\Delta \Vdash e : T)$, the compiler produces a decoherence-based circuit with input space $\Hilb(\Delta)$ and output space $\Hilb(T)$.
		This circuit's semantics is the same as $\msem{\Delta \Vdash e : T}$.
		As an example, Figure~\ref{fig:compile-ex-mixedswap} shows the circuit for the \textsc{T-MixedPerm} typing rule.
		The $\pi$ gate here is a series of \textsc{swap} gates that permutes the data in $\Delta$ according to the permutation function $\pi$.
		This example demonstrates the benefit of the explicit exchange rules---exchange of variables corresponds to \textsc{swap} gates in the quantum circuit.

\begin{figure}[ht]
	\[
	\inference{\Delta \Vdash e : T}{\pi(\Delta) \Vdash e : T}[\textsc{T-MixedPerm}]
	\qquad
\begin{quantikz}
	\lstick{$\Hilb(\pi(\Delta))$} & \gate{\pi^{-1}} & \qwbundle{\Hilb(\Delta)} & \qw & \gate{e} & \rstick{$\Hilb(T)$} \qw \\
\end{quantikz}
\]
\caption{\textsc{T-MixedPerm} compilation}
	\label{fig:compile-ex-mixedswap}
\end{figure}

\paragraph*{Mixed programs}                
		Given a judgment $(\vdash f : T \Rrightarrow T')$, the compiler produces a decoherence-based circuit with input space $\Hilb(T)$ and output space $\Hilb(T')$.
		This circuit's semantics is the same as $\msem{\vdash f : T \Rrightarrow T'}$.
		As an example, Figure~\ref{fig:compile-ex-mixedabs} shows the circuit for the \textsc{T-MixedAbs} typing rule.
		This is the same as the \textsc{T-PureAbs} circuit, except that there is an extra context $\Delta_0$ for unused variables, which are discarded.

\begin{figure}[ht]
	\[
	\inference{\varnothing\partition \Delta, \Delta_0 \vdash e : T \qquad \Delta \Vdash e' : T'}{\vdash \lambda e \xmapsto{{\color{gray}T}} e' : T \Rrightarrow T'}[\textsc{T-MixedAbs}]
	\]
\medskip
	\[
\begin{quantikz}
	\lstick{$\Hilb(T)$} & \gate[2,nwires=2]{{\msem{\varnothing\partition \Delta, \Delta_0 \vdash e : T}^\dagger}} & \qwbundle{\Hilb(\Delta)} & \gate{{\Delta \Vdash e' : T'}} & \rstick{$\Hilb(T')$} \qw \\
											&& \trash{\Hilb(\Delta_0)}
\end{quantikz}
\]
\caption{\textsc{T-MixedAbs} compilation}
	\label{fig:compile-ex-mixedabs}
\end{figure}

\subsection{High-Level Circuits to Low-Level Circuits}
\label{sec:hl2ll}

The high-level circuits constructed in the previous section are convenient for analysis, in particular for our proofs of 
\ifshort
correctness.
\else
correctness (in Appendix~\ref{app:compilation}).
\fi
For these circuits to be runnable on quantum hardware, however, an additional compilation stage is necessary, transforming these high-level circuits into low-level ones.
In particular, our high-level circuits manipulate values in $\V(T)$ in Hilbert space $\Hilb(T)$ with $\dim(\Hilb(T)) = |\V(T)|$, but low-level circuits use only qubits, working in the Hilbert space $\complex^{2^n}$ where $n$ is the number of qubits on the quantum computer.
Our compiler must then translate Qunity programs written for the Hilbert space $\Hilb(T)$ into Open\textsc{Qasm} programs that use some $|\V(T)|$-dimensional subspace of $\complex^{2^n}$, and values in $\V(T)$ will be encoded into bitstrings in $\{\zero, \one\}^n$.
This number of qubits $n$ is determined from $T$ by a function $\size()$, shown below. Values in $\V(T)$ will end up encoded into bitstrings whose length is $\size(T)$ using the function $\encode()$, where ``$\doubleplus$'' denotes bitstring concatenation, and we omit right-zero-padding for brevity.

\smallskip
\begin{align*}
	\size(\Void) &\defeq 0 & \encode(\unit) &= \texttt{""} \\
	\size(\Unit) &\defeq 0 & \encode(\lef{T_0}{T_1}\: v) &= \texttt{"0"} \doubleplus \encode(v) \\
	\size(T_0 \oplus T_1) &\defeq 1 + \max \{ \size(T_0), \size(T_1) \} & \encode(\rit{T_0}{T_1}\: v) &= \texttt{"1"} \doubleplus \encode(v) \\
	\size(T_0 \otimes T_1) &\defeq \size(T_0) + \size(T_1) & \encode(\pair{v_0}{v_1}) &= \encode(v_0) \doubleplus \encode(v_1)
\end{align*}
\smallskip

We can now be precise about what it means to implement a norm-non-increasing operator with a low-level circuit.
\begin{definition}
	\label{def:implement-kraus}
	We say that it is possible to implement a norm-non-increasing operator $E : \Hilb(T) \to \Hilb(T')$ if there is a low-level circuit implementing a unitary operator $U : \complex^{2^{\size(T) + \nprep}} \to \complex^{2^{\size(T') + \nflag}}$ for some integers $\nprep$ and $\nflag$ such that for all $v \in \V(T), v' \in \V(T')$:
	\[
		\bra{\encode(v'), \zero^{\otimes \nflag}} U \ket{\encode(v), \zero^{\otimes \nprep}} = \bra{v'} E \ket{v}
	\]
	Here, $\nprep$ is the number of ``prep'' qubits initialized to $\ket{\zero}$, and $\nflag$ is the number of ``flag'' qubits asserted to be in the $\ket{\zero}$ state upon termination.
	This definition requires that $\size(T) + \nprep = \size(T') + \nflag$.
\end{definition}

For example, consider the ``share'' gate graphically represented by the ``controlled cloud'' in Figure~\ref{fig:compile-ex-purepair}.
This gate can be defined for any type $T$, where it is mathematically represented by the isometry $\sum_{v \in \V(T)} \op{v, v}{v}$, copying a value in the standard basis.
By Definition~\ref{def:implement-kraus}, one can implement this operator with $\nprep = \size(T)$ prep qubits and $\nflag = 0$ flag qubits, using the series of \textsc{cnot} gates shown in Figure~\ref{fig:compile-share}.

\begin{figure}
\[
\begin{quantikz}
	\lstick{$\Hilb(T)$} & \ctrl{1} & \rstick{$\Hilb(T)$} \qw \\
										& \gate[style={cloud},nwires=1]{} & \rstick{$\Hilb(T)$} \qw
\end{quantikz}
\quad\mapsto
\begin{quantikz}
	\lstick[wires=4]{$\size(T)$} & \ctrl{4} & \qw & \ \ldots\ \qw & \qw & \rstick[wires=4]{$\size(T)$} \qw\\
															 & \qw & \ctrl{4} &  \ \ldots\ \qw & \qw & \qw \\
	\wave&&&&& \\
			 & \qw & \qw & \ \ldots\ \qw & \ctrl{4} & \qw \\
	\lstick[wires=4]{$\nprep = \size(T)$} & \targ{} & \qw & \ \ldots\ \qw & \qw & \rstick[wires=4]{$\size(T)$} \qw \\
																				& \qw & \targ{} & \ \ldots\ \qw & \qw & \qw \\
	\wave&&&&& \\
	& \qw & \qw & \ \ldots\ \qw & \targ{} & \qw
\end{quantikz}
\]
\caption{A (high-level) ``share'' gate implemented by (low-level) \textsc{cnot} gates}
	\label{fig:compile-share}
\end{figure}

Definition~\ref{def:implement-kraus} can be adapted to the setting of decoherence and superoperators by including an additional ``garbage'' register.
This is essentially the Stinespring dilation \cite[p.~186]{stinespring-dilation}, using the principle of deferred measurement to purify our computation and organize all of our measurements onto one segment of our output.

\begin{definition}
	\label{def:implement-superop}
	We say that it is possible to implement a trace-non-increasing superoperator $\mathcal{E} : \linear(\Hilb(T)) \to \linear(\Hilb(T'))$ if there is a qubit circuit implementing a unitary operator $U : \complex^{2^{\size(T) + \nprep}} \to \complex^{2^{\size(T') + \nflag + \ngarb}}$ for some integers $\nprep$, $\nflag$, and $\ngarb$ such that for all $v_1, v_2 \in \V(T)$, $v_1', v_2' \in \V(T')$:
	\begin{DIFnomarkup}
	\begin{alignat*}{3}
		&&& \bra{v_1'} \mathcal{E}(\op{v_1}{v_2}) \ket{v_2'} \\
		&=&\;& \sum_{b \in \{\zero, \one\}^{\ngarb}} \bra{\encode(v_1'), \zero^{\otimes \nflag}, b} U \left(\op{\encode{v_1}}{\encode{v_2}} \otimes \op{\zero}{\zero}^{\otimes \nprep}\right) U^\dagger \ket{\encode(v_2'), \zero^{\otimes \nflag}, b}
	\end{alignat*}
	\end{DIFnomarkup}
	Like before, $\nprep$ is the number of ``prep'' qubits initialized to $\ket{\zero}$, and $\nflag$ is the number of ``flag'' qubits asserted to be in the $\ket{\zero}$ state upon termination.
	The new parameter $\ngarb$ is the number of ``garbage'' qubits discarded after use.
	This definition requires that $\size(T) + \nprep = \size(T') + \nflag + \ngarb$.
\end{definition}

For example, under this definition the ``discard'' used for the $\Hilb(\Delta_0)$ space in Figure~\ref{fig:compile-ex-mixedabs} can be implemented by an empty (identity) circuit by setting $\ngarb = \size(T)$.
Any implementable pure circuit is also implementable as a decoherence-based circuit by setting $\ngarb=0$.

\subsection{Example}

As a concrete example of compiling a simple program, consider a ``coin flip'' expression defined below.
This program is defined in terms of a \texttt{meas} program that measures its argument in the standard basis, by copying in the standard basis and then discarding the copy.

\begin{align*}
	\meas{T} &\defeq \lambda x \xmapsto{{\color{gray}T}} \pair{x}{x} \triangleright \fst{\Bit}{\Bit} \\
	\texttt{coin} &\defeq \meas{\Bit} \left( \had\: \zero \right) \\
								&= \unit \triangleright \texttt{left}_{\color{gray}\Bit} \triangleright \had \triangleright \lambda x \xmapsto{{\color{gray}\Bit}} \pair{x}{x} \triangleright \lambda \pair{x_0}{x_1} \xmapsto{{\color{gray}T\otimes T}} x_0
\end{align*}

The typing judgment $(\varnothing \Vdash \texttt{coin} : \Bit)$ is compiled into the following high-level circuit.
Here you can see that the $\pair{x}{x}$ expression implements the isometry $(\op{\zero, \zero}{\zero} + \op{\one, \one}{\one})$, and the $\texttt{fst}$ program implements a partial trace.

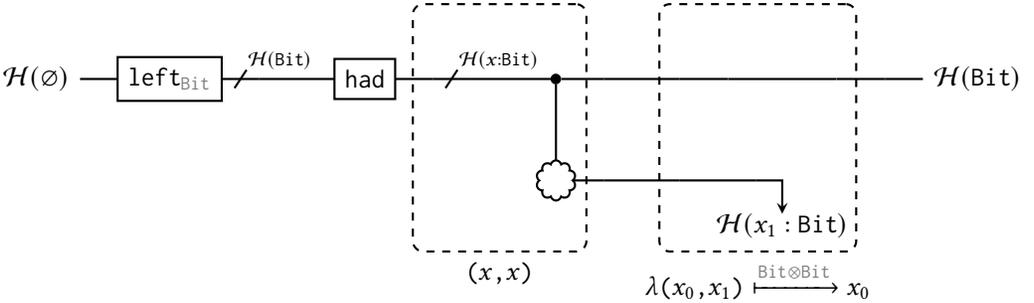
\[
	\begin{quantikz}[row sep=-30]
		\lstick{$\Hilb(\varnothing)$} & \gate{\texttt{left}_{\color{gray}\Bit}} & \qwbundle{\Hilb(\Bit)} &[0.5cm] \gate{\had} & \qw \gategroup[wires=2,steps=3,style={dashed,rounded corners},label style={label position=below,anchor=north,yshift=-0.2cm}]{$\pair{x}{x}$} & \qwbundle{\Hilb(x:\Bit)} &[0.5cm] \ctrl{1} & \qw &[0.5cm] \qw \gategroup[wires=2,steps=2,style={dashed,rounded corners},label style={label position=below,anchor=north,yshift=-2mm}]{$\lambda \pair{x_0}{x_1} \xmapsto{{\color{gray}\Bit \otimes \Bit}} x_0$} & \hphantomgate{quite wide} & \qw & \rstick{$\Hilb(\Bit)$} \qw & \\[1cm]
																	&&&&&& \gate[style={cloud},nwires=1]{} & \qw & \ghost{\Bigg\{} \qw & \trash{\Hilb(x_1 : \Bit)} & & &
	\end{quantikz}
\]

This is then compiled into the following low-level circuit, with $n\subcap{prep}=2$, $n\subcap{flag} = 0$, $n\subcap{garb} = 1$, $\size(\varnothing) = 0$ contextual input wires, and $\size(\Bit) = 1$ data output wire.
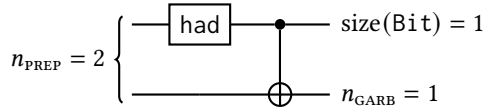
\[
	\begin{quantikz}
		\lstick[2]{$n\subcap{prep} = 2$} & \gate{\had} & \ctrl{1} & \rstick{$\size(\Bit) = 1$} \qw \\
																		 & \qw & \targ{} & \rstick{$n\subcap{garb} = 1$} \qw
	\end{quantikz}
\]

\section{Conclusion}

Qunity is designed to unify classical and quantum computing through an expressive generalization of classical programming constructs.
Its syntax allows programmers to write quantum algorithms using familiar classical programming constructs, like exception handling and pattern matching.
Our type system leverages algebraic data types and relevant (substructural) types, differentiating between unitary maps and quantum channels but allowing them to be usefully nested.
Qunity's semantics brings constructions commonly used in algorithm \emph{analysis}---such as bounded-error quantum subroutines, projectors, and direct sums---into the realm of algorithm \emph{implementation}.

We have formally defined the Qunity programming language, proven that its semantics is well-defined, and shown how it can be used to implement some complicated quantum algorithms.
We have demonstrated a strategy for compiling Qunity programs to low-level qubit-based unitary circuits and proven that our procedure preserves the semantics, demonstrating that this language does indeed have a physical interpretation and could be run on quantum hardware.
While classical computers can be modeled by logic gates acting on classical bits, it is far more convenient to use higher-level programming constructs for most tasks.
Qunity's design similarly abstracts away low-level qubit-based gates using techniques from quantum algorithms that are overlooked in existing languages.
We hope that Qunity's features can ease the implementation and analysis of complicated quantum algorithms written at a high level of abstraction.

\begin{acks}
We thank Andrew Childs for helping us to understand the quantum algorithm designer's perspective and for originally pointing us toward the \textsc{bqp} subroutine theorem, and we thank the anonymous referees for their helpful comments on a draft of this paper.
We thank Mikhail Mints for helping us with numerous corrections to the published version of this paper.
This material is based upon work supported by the U.S. Department of Energy, Office of Science, Office of \grantsponsor{ascr}{Advanced Scientific Computing Research}{https://doi.org/10.13039/100006192}, Quantum Testbed Pathfinder Program under Award Number
\grantnum{ascr}{DE-SC0019040}, and the \grantsponsor{afosr}{Air Force Office of Scientific Research}{https://doi.org/10.13039/100000181} under Grant No.
\grantnum{afosr}{FA95502110051}.
This work is funded in part by EPiQC, an \grantsponsor{ccf}{NSF}{10.13039/100000143} Expedition in Computing, under award \grantnum{ccf}{CCF-1730449}.
\end{acks}

\bibliographystyle{ACM-Reference-Format}
\begin{DIFnomarkup}
\bibliography{references.bib}
\end{DIFnomarkup}

\ifshort
\else
\appendix

\section{Quantum computing background}

Quantum computing is an emerging field at the intersection of theoretical computer science and quantum physics.
Quantum computers are designed to exploit physical phenomena that are uniquely quantum mechanical, such as superposition and entanglement.
Modern quantum computers are largely impractical and outperformed by ``classical'' (non-quantum) computers, but some quantum algorithms scale exponentially better than their best known classical counterparts, meaning that a practical quantum advantage may be possible in the near future if the hardware can be made scalable and precise.
There are plenty of textbooks that provide a good introduction to the field \cite{mike-and-ike,klm,mermin} and survey its current state \cite{2019-progress-prospects}; our purpose in this appendix is to provide a more streamlined overview of some of the lesser-known mathematical formalisms used throughout this work.

The denotational semantics of a deterministic program run on a finite classical computer can often be described by a function $f : S \to T$, where $S$ is some finite set of possible input states and $T$ is some finite set of possible output states, perhaps using a \emph{partial} function to model non-termination or error.
The size of the sets $S$ and $T$ scales exponentially with the physical size of the device: 64 physical bits of input means $2^{64}$ possible inputs.
The main difference in the quantum setting is that it is easier to model the input and output as finite-dimensional vector spaces rather than finite sets, and the function $f$ as a linear operator on these vector spaces.
Basis vectors correspond to classical states, and non-basis vectors represent a \emph{superposition}.
The exponential advantage of quantum computing comes from the ability to leverage these superpositions to enact useful computations.

\subsection{Hilbert spaces}
\label{app:hilbert}

We assume that the reader is familiar with \emph{inner product spaces}, vector spaces equipped with an inner product \cite[p.~167]{linear-algebra}.
A finite-dimensional \emph{Hilbert} space is simply a finite-dimensional inner product space \cite[p.~66]{mike-and-ike}.
(The definition of Hilbert space is a bit more restricted in the infinite-dimensional case, but we are only concerned with finite dimensions in this work.)
All of our Hilbert spaces will be denoted with the letter $\Hilb$ and will be over the field $\complex$ of complex numbers.

Bra-ket notation is a convenient notation for describing vectors and operators in Hilbert spaces.
A vector is written using ket notation, writing ``$\ket{\zero}$'' rather than ``$\vec{x}_\zero$'' used in other areas of mathematics.
The inner product of any vectors $\ket{\phi}$ and $\ket{\psi}$ is written $\ip{\phi}{\psi}$, and is linear in the second argument $\ket{\psi}$.
The outer product, written $\op{\phi}{\psi}$, is the linear operator defined such that $\ket{\phi}\!\!\ip{\psi}{\psi'} = \ip{\psi}{\psi'} \cdot \bra{\phi}$, where ``$\cdot$'' denotes scalar multiplication.
It is often useful to treat a ``ket'' $\ket{\psi}$ as a \emph{column vector} and a ``bra'' $\bra{\psi} = \ket{\psi}^\dagger$ as a \emph{row vector} (where ``$\dagger$'' denotes the conjugate transpose), and then both the inner and outer product are simply matrix multiplication.

The trivial Hilbert space is the set $\{0\}$, a zero-dimensional Hilbert space with only one vector, the zero vector.
The set $\complex$ forms a one-dimensional Hilbert space over the complex numbers, spanned by the one-element orthonormal basis $\{1\}$.
The inner product of two complex numbers $z_1$ and $z_2$ is $z_1^{*} z_2$.

Our algebraic data types in this paper rely on two ways of composing Hilbert spaces: (external) direct sums $\oplus$ and tensor products $\otimes$.
Given two Hilbert spaces $\Hilb\subcap{m}$ and $\Hilb\subcap{n}$ with $\dim(\Hilb\subcap{m}) = M$ and $\dim(\Hilb\subcap{n}) = N$, the direct sum and tensor product are defined so that $\dim(\Hilb\subcap{m} \oplus \Hilb\subcap{n}) = M + N$ and $\dim(\Hilb\subcap{m} \otimes \Hilb\subcap{n}) = M \cdot N$.
These operators apply to Hilbert spaces, as well as vectors and operators within these spaces.
For example, if $\ket{\phi} \in \Hilb\subcap{m}$ and $\ket{\psi} \in \Hilb\subcap{n}$, then $\ket{\phi} \oplus \ket{\psi} \in \Hilb\subcap{m} \oplus \Hilb\subcap{n}$ and $\ket{\phi} \otimes \ket{\psi} \in \Hilb\subcap{m} \otimes \Hilb\subcap{n}$, and if $E\subcap{m} : \Hilb\subcap{m,1} \to \Hilb\subcap{m,2}$ and $E\subcap{n} : \Hilb\subcap{n,1} \to \Hilb\subcap{n,2}$ are linear operators, then $E\subcap{m} \oplus E\subcap{n} : \Hilb\subcap{m,1} \oplus \Hilb\subcap{n,1} \to \Hilb\subcap{m,2} \oplus \Hilb\subcap{n,2}$ is also a linear operator, as well as $E\subcap{m} \otimes E\subcap{n} : \Hilb\subcap{m,1} \otimes \Hilb\subcap{n,1} \to \Hilb\subcap{m,2} \otimes \Hilb\subcap{n,2}$.
Following standard notational conventions, we sometimes write $\ket{\phi, \psi}$ for $\ket{\phi} \otimes \ket{\psi}$.

The tensor product will be familiar to anyone with a basic understanding of quantum computing.
It is used to describe \emph{joint states}.
If a qubit inhabits a two-dimensional Hilbert space $\Hilb_2$, and a qutrit inhabits a three-dimensional Hilbert space $\Hilb_3$, then the six-dimensional Hilbert space $\Hilb_2 \otimes \Hilb_3$ describes a physical system containing a qubit and a qutrit, possibly entangled.
In the context of programming languages, tensor products correspond nicely to \emph{product types}.

Qubit arrays serve as a sufficient physical model for most of the interesting results in quantum computing.
An $n$-qubit system is described by a $2^n$-dimensional Hilbert space, which is what allows quantum computing to achieve its famed exponential advantages.
For a programmer, however, this can be fairly limiting for expressiveness, effectively restricting one's data types to those whose cardinality is a power of two.
To allow for arbitrary finite data types, the direct sum is useful, serving as an ``\textsc{or}'' where the tensor product serves as ``\textsc{and}.''
The five-dimensional Hilbert space $\Hilb_2 \oplus \Hilb_3$ describes quantum states that are \emph{either} a qubit \emph{or} a qutrit, or a superposition of the two.

These constructions may be more intuitive given matrix representations.
The direct sum of two vectors is their concatenation, while the direct sum of two matrices is their block diagonalization, as shown by the examples in Figure~\ref{fig:dsum-matrix}.

\begin{figure}[ht]
\begin{align*}
	\ket{\phi} &\defeq \begin{bmatrix} a \\ b \end{bmatrix} & \ket{\psi} &\defeq \begin{bmatrix} c \\ d \\ e \end{bmatrix} \\
	\ket{\phi} \oplus \ket{\psi} &= \begin{bmatrix} a \\ b \\ c \\ d \\ e \end{bmatrix} & \ket{\phi} \otimes \ket{\psi} &= \begin{bmatrix} a c \\ a d \\ a e \\ b c \\ b d \\ b e \end{bmatrix} \\
\end{align*}
\[
\begin{aligned}
\begin{aligned}
	E\subcap{a} &\defeq \begin{bmatrix} f & g & h & i \\ j & k & l & m \\ n & o & p & q \end{bmatrix} \\ E\subcap{b} &\defeq \begin{bmatrix} r & s \\ t & u \\ v & w \end{bmatrix} \\
\end{aligned}
\qquad
	E\subcap{a} \oplus E\subcap{b} = \begin{bmatrix} f & g & h & i & 0 & 0 \\ j & k & l & m & 0 & 0 \\ n & o & p & q & 0 & 0 \\ 0 & 0 & 0 & 0 & r & s \\ 0 & 0 & 0 & 0 & t & u \\ 0 & 0 & 0 & 0 & v & w \end{bmatrix} \\
\end{aligned}
\]
\vspace{5mm}
\[
	E\subcap{a} \otimes E\subcap{b} = \begin{bmatrix} f r & f s & g r & g  s & h  r & h  s & i  r & i  s \\ f  t & f  u & g  t & g  u & h  t & h  u & i  t & i  u \\ f  v & f  w & g  v & g  w & h  v & h  w & i  v & i  w \\
j r & j s & k r & k  s & l  r & l  s & m  r & m  s \\ j  t & j  u & k  t & k  u & l  t & l  u & m  t & m  u \\ j  v & j  w & k  v & k  w & l  v & l  w & m  v & m  w \\
n r & n s & o r & o  s & p  r & p  s & q  r & q  s \\ n  t & n  u & o  t & o  u & p  t & p  u & q  t & q  u \\ n  v & n  w & o  v & o  w & p  v & p  w & q  v & q  w
	\end{bmatrix}
	\]
\caption{Matrix examples for direct sums and tensor products}
\label{fig:dsum-matrix}
\end{figure}

Suppose Hilbert space $\Hilb\subcap{m}$ is spanned by an orthonormal basis $B\subcap{m}$, and Hilbert space $\Hilb\subcap{n}$ is spanned by an orthonormal basis $B\subcap{n}$.
Then $\Hilb\subcap{m} \oplus \Hilb\subcap{n}$ is spanned by an orthonormal basis $\{\ket{u} \oplus 0 : \ket{u} \in B\subcap{m}\} \cup \{ 0 \oplus \ket{v} : \ket{v} \in B\subcap{n}\}$, and $\Hilb\subcap{m} \otimes \Hilb\subcap{n}$ is spanned by an orthonormal basis $\{\ket{u} \otimes \ket{v} : \ket{u} \in B\subcap{m}, \ket{v} \in B\subcap{n} \}$.
The following identities may be useful for understanding operations on these spaces.

\begin{align*}
	\left( \ket{\phi_\zero} \oplus \ket{\psi_\zero} \right) + \left( \ket{\phi_\one} \oplus \ket{\psi_\one} \right) &= \left( \ket{\phi_\zero} + \ket{\phi_\one} \right) \oplus \left( \ket{\psi_\zero} + \ket{\psi_\one} \right) \\
	z \left( \ket{\phi} \oplus \ket{\psi} \right) &= z \ket{\phi} \oplus z \ket{\psi} \\
	\left( \ket{\phi} \oplus \ket{\psi} \right)^\dagger &= \bra{\phi} \oplus \bra{\psi} \\
	\left( \bra{\phi_\zero} \oplus \bra{\psi_\zero} \right) \left( \ket{\phi_\one} \oplus \ket{\psi_\one} \right) &= \ip{\phi_\zero}{\phi_\one} + \ip{\psi_\zero}{\psi_\one} \\
	(E\subcap{m} \oplus E\subcap{n}) \left(\ket{\phi} \oplus \ket{\psi} \right) &= E\subcap{m} \ket{\phi} \oplus E\subcap{n} \ket{\psi} \\
	z \ket{\phi} \otimes \ket{\psi} &= \ket{\phi} \otimes z \ket{\psi} \\
																	&= z \left( \ket{\phi} \otimes \ket{\psi} \right) \\
	\left( \ket{\phi} \otimes \ket{\psi} \right)^\dagger &= \bra{\phi} \otimes \bra{\psi} \\
	\left( \bra{\phi_\zero} \otimes \bra{\psi_\zero} \right) \left( \ket{\phi_\one} \otimes \ket{\psi_\one} \right) &= \ip{\phi_\zero}{\phi_\one} \cdot \ip{\psi_\zero}{\psi_\one} \\
	(E\subcap{m} \otimes E\subcap{n}) \left(\ket{\phi} \otimes \ket{\psi} \right) &= E\subcap{m} \ket{\phi} \otimes E\subcap{n} \ket{\psi}
\end{align*}

Category theory can help describe the relationship between the direct sum and tensor product.
Together, they make the category of finite-dimensional vector spaces and linear operators into a \emph{symmetric bimonoidal category} \cite[p.~130]{symmetric-bimonoidal}, also known as a ``rig category.''
This means that the category of finite-dimensional vector spaces is a monoidal category with respect to both the direct sum and the tensor product, and there are isomorphisms with some well-characterized properties similar to what one would expect from addition and multiplication.
In particular, there is a distributivity isomorphism from $\Hilb\subcap{a} \otimes \left(\Hilb\subcap{b} \oplus \Hilb\subcap{c}\right)$ to $\Hilb\subcap{a} \otimes \Hilb\subcap{b} \oplus \Hilb\subcap{a} \otimes \Hilb\subcap{c}$, mapping $\ket{a} \otimes \left( \ket{b} \oplus \ket{c} \right)$ to $\ket{a} \otimes \ket{b} \oplus \ket{a} \otimes \ket{c}$.

Qunity's semantics is defined in terms of norm-non-increasing operators, operators $E$ such that $\mathbb{I} - E^\dagger E$ is a positive operator.
It is not hard to see that these operators are closed under composition as well as the direct sum and tensor product, meaning that there is also a symmetric bimonoidal category of finite-dimensional Hilbert spaces and norm-non-increasing operators.
This is the category on which Qunity's pure denotational (categorical) semantics is defined.

\subsection{Operators and superoperators}
\label{app:superop}

Quantum physicists have two main ways to mathematically represent a quantum state: a \emph{pure} state is represented by a state \emph{vector}, while a \emph{mixed} state is represented by a density \emph{operator} (or density \emph{matrix}).
A computation on a pure state is described by a linear operator from the input Hilbert space to the output Hilbert space, while a computation on mixed states is described by a \emph{superoperator}, a sort of higher-order function that acts as a linear operator from density operators to density operators.\footnote{Superoperators are technically operators by definition, but throughout this work we often contrast them with the ``pure'' operators used to describe actions on pure states, describing the latter simply as ``operators.''}

Existing textbooks \cite[Section~3.5]{klm} have a good overview of the mathematical formalisms involved; our purpose here is to justify the need for both of these kinds of computations in quantum programming languages.
Some readers may be concerned that while Qunity's type system breaks down the distinction between quantum and classical types, it creates a new distinction between pure and mixed programs.
In Appendix~\ref{app:classical-semantics}, we illustrate that the pure/mixed distinction is fairly different from the classical/quantum one, and in this section we hope to justify that the expressiveness gained from including the pure/mixed distinction is worth the more complicated typing rules.

One main distinction between operators and superoperators is \emph{reversibility}: operators are reversible, while superoperators are not.
By ``reversible,'' we do not mean strictly \emph{invertible}, as Qunity's semantics allows for non-invertible norm-decreasing operators.
Rather, we mean that a linear operator $E : \Hilb_1 \to \Hilb_2$ has an \emph{adjoint} operator $E^\dagger : \Hilb_2 \to \Hilb_1$ \cite[p.~204]{linear-algebra}, and the adjoint is implementable as a circuit whenever the original is implementable (as shown later in Lemma~\ref{lem:compile-adj}).
If one is focused only on unitary operators, then the adjoint is the inverse.

It is useful to draw an analogy with classical reversible computing.
For Qunity's classical sublanguage (Appendix~\ref{app:sublanguage}), the semantics of a program can be described by a partial function over values.
For a purely-typed program, this partial function is injective.
Any injective partial function is invertible, and this inverse partial function is implementable in Qunity as $f^{\dagger {\color{gray}{T}}}$ in Figure~\ref{fig:notation}.
By encoding reversibility into the typing judgment, Qunity informs the programmer which programs are reversible and allows only those programs to be reversed.
Many quantum algorithms depend on the ability to implement adjoint operators, and Qunity allows this to be relatively seamless.

The other main advantage of operators over superoperators has no classical analog: global phase.
In particular, the identity operator $\mathbb{I}$ is distinct from the negative identity operator $-\mathbb{I}$, even though both correspond to the same (identity) superoperator.
This is true for any two operators that differ by a global phase, and it raises an important question for the semantics of any quantum programming language: should these two programs be semantically identical (using superoperator semantics), or should they be distinct (using operator semantics)?
Qunity solves this problem in a somewhat unusual way, defining two distinct (but interrelated) semantics: one of operators, and one of superoperators.
Ying's Qu\textsc{Gcl} \cite[Chapter~6]{mingsheng-ying} takes the superoperator semantics approach, but this is a problem for a language with quantum control.
While the operators $\mathbb{I}$ and $-\mathbb{I}$ correspond to the same superoperator, the controlled operators $\mathbb{I} \oplus \mathbb{I}$ and $\mathbb{I} \oplus (-\mathbb{I})$ (or equivalently, $\ket{\zero} \otimes \mathbb{I} + \ket{\one} \otimes \mathbb{I}$ and $\ket{\zero} \otimes \mathbb{I} - \ket{\one} \otimes \mathbb{I}$) do not.
This leads to an unfortunate problem where Qu\textsc{Gcl}'s semantics is non-compositional, and ``equivalent'' expressions are not substitutable.\footnote{Qu\textsc{Gcl}'s non-compositional superoperator semantics is defined in terms of a compositional semantics based on lists of Kraus operators. If one restricts attention solely to the Kraus operator semantics, then compositionality is not an issue, but one loses the benefits of the superoperator formalism. Somewhat confusingly, Ying uses equals notation ``='' to denote equivalence of the non-compositional superoperator semantics.}

Some languages have taken the opposite approach, ignoring superoperators and restricting attention to pure computation only.
For example, semantics of the symmetric pattern matching language \cite{symmetric-pattern-matching} consists solely of unitary operators.
In one sense, no computational power is lost by taking this route.
Any irreversible computation can be made reversible by keeping track of all intermediate ``garbage'' data.
As a classical example, the irreversible function $\textsc{and} : \{\zero,\one\}^2 \to \{\zero,\one\}$ can be converted into a bijective function $\textsc{reversible-and} : \{\zero,\one\}^3 \to \{\zero, \one\}^3$ defined as $\textsc{reversible-and}(a, b, c) \defeq \left(\textsc{xor}\left(c, \textsc{and}(a, b)\right), a, b\right)$.
This function is bijective, and the behavior of the original \textsc{and} function can be recovered by inputting $c=0$ and discarding the garbage outputs $a$ and $b$.
The quantum generalization of this idea is the \emph{Stinespring dilation} \cite[p.~186]{stinespring-dilation}, allowing any superoperator to be represented as an operator, using extra inputs and discarded outputs.
Some interpretations of quantum mechanics assert that \emph{all} interactions are unitary, describing even seemingly nonunitary decoherence as really unitary dynamics in some larger Hilbert space \cite{kinematics-dynamics}.

In a programming language, the main problem with requiring all programs to be reversible is that it forces the programmer to perform a lot of unnecessary bookkeeping. 
Classical programmers are used to composing programs without regard to reversibility, and requiring programs like \textsc{reversible-and} is cumbersome and unintuitive.
Even without Qunity's goal of unifying quantum and classical computing, mixed states are often the most convenient way for quantum algorithm designers to describe intermediate states in quantum programs \cite[p.~25]{childs-van-dam-review}, and this is impossible in a language that only describes reversible operations on pure states.

One other reason for including superoperator semantics is the fact that it does a good job of capturing the aspects of a computation that are physically distinguishable.
Though the $\mathbb{I}$ and $-\mathbb{I}$ operators described previously have different behavior when subject to quantum control, they are indistinguishable when uncontrolled.
Given black-box access, it is impossible to determine which is which by simply feeding states in and measuring the output.
The Stinespring dilation is not unique in general, and radically different unitary operators can give rise to the same superoperator.
The superoperator formalism allows one to treat more programs as semantically identical, as long as they are not subject to quantum control (which in Qunity is only allowed for pure programs).
This canonicity can have practical advantages, like a wider range of valid compiler optimizations.

It is for these reasons that we made the unusual choice of giving Qunity two different (but interrelated) semantics, one of operators and one of superoperators.
An operator semantics is necessary for expressing reversibility and control, but insufficient for implementing any function that is not bijective.
Both pure states and mixed states are widely used in the analysis of quantum algorithms, and one or the other may be more convenient depending on the context.
Quantum algorithm designers are used to switching between the two representations, and Qunity is designed to accommodate this fluidity.

\section{Classical Sublanguage}
\label{app:sublanguage}

In this appendix, we consider the sublanguage of Qunity defined by removing the two uniquely quantum constructs: \texttt{u3} and \texttt{rphase}.
\begin{align*}
	e &\defeqq \unit \mid x \mid \pair{e}{e} \mid \cntrl{e}{T}{e &\mapsto e \\ &\cdots \\ e &\mapsto e}{T} \mid \trycatch{e}{e} \mid f\: e \\
	f &\defeqq \lef T T \mid \rit T T \mid \lambda e \xmapsto{{\color{gray}{T}}} e
\end{align*}

We will show that this sublanguage is reasonably expressive, has a classical semantics that closely corresponds to full Qunity's quantum semantics, and can be executed on classical hardware, proving Theorems~\ref{thm:mlpi} and \ref{thm:classical-generalization}.

\subsection{Classical semantics}
\label{app:classical-semantics}

First, we will show how the semantics of this classical sublanguage can be classically interpreted.
We will use the same typing judgment as in full Qunity but without the \textsc{T-Gate} and \textsc{T-Rphase} rules.
Like Qunity's full semantics, this classical semantics is a recursive function of the typing judgment, but we will use $\psem{\cdot}$ to denote these partial functions on values rather than $\msem{\cdot}$ to represent linear operators on Hilbert spaces.
The use of \emph{partial} functions corresponds to our use of \emph{norm-decreasing} operators and \emph{trace-decreasing} superoperators, where an input ouside the domain of definition corresponds to a failure, a norm or trace of zero.
(The \texttt{iso} judgment defined in Appendix~\ref{app:iso} would guarantee a total function.)
The distinction between ``pure'' and ``mixed'' typing relations in the quantum setting corresponds to the distinction between coherent operators and superoperators (which may discard information), but in our restricted classical setting, the corresponding notion will be whether a partial function is \emph{injective} (that is, reversible).
Though we are working in a classical setting, we will continue to refer to the ``pure semantics'' and ``mixed semantics'' to draw the connection to Qunity's existing typing judgment and to avoid the rather wordy ``not necessarily injective semantics.''
Precisely:
\begin{itemize}
	\item
		If $\Gamma\partition \Delta \vdash e : T$ and $\sigma \in \V(\Gamma)$, then $\psem{\sigma : \Gamma \partition \Delta \vdash e : T} : \V(\Delta) \rightharpoonup \V(T)$ is an injective partial function that defines the pure classical semantics of expression $e$.
		We give the denotation in Figure~\ref{fig:sem-pure-c-expr}.
		The $\sigma$ is a sort of ``optional data,'' so the pure classical expression semantics may be viewed as a two-parameter partial function $\V(\Gamma) \times \V(\Delta) \rightharpoonup \V(T)$, injective in its second argument.  \item If $\Delta \Vdash e : T$, then $\psem{\Delta \Vdash e : T} : \V(\Delta) \rightharpoonup \V(T)$ is a (not necessarily injective) partial function that defines the mixed classical semantics of expression $e$.
		We give the denotation in Figure~\ref{fig:sem-mixed-c-expr}.
	\item
		If $\vdash f : T \rightsquigarrow T'$, then $\psem{\vdash f : T \rightsquigarrow T'} : \V(T) \rightharpoonup \V(T')$ is an injective partial function that defines the pure classical semantics of program $f$.
		We give the denotation in Figure~\ref{fig:sem-pure-c-prog}.
	\item
		If $\vdash f : T \Rrightarrow T'$, then $\psem{\vdash f : T \Rrightarrow T'} : \V(T) \rightharpoonup \V(T')$ is a (not necessarily injective) partial function that defines the mixed semantics of program $f$.
		We give the denotation in Figure~\ref{fig:sem-mixed-c-prog}.
\end{itemize}

\begin{figure}[th]
\begin{alignat*}{2}
	\psem{\sigma : \Gamma \partition \varnothing \vdash \unit : \Unit} (\varnothing) &\defeq&\;& \unit \\
	\psem{\sigma : \Gamma \partition \varnothing \vdash x : T} (\varnothing) &\defeq&& \sigma(x) \\
	\psem{\sigma : \Gamma \partition x : T \vdash x : T} (x \mapsto v) &\defeq&& v \\
 \psem{\sigma : \Gamma \partition \Delta, \Delta_0, \Delta_1 \vdash \pair{e_0}{e_1} : T_0 \otimes T_1} (\tau, \tau_0, \tau_1) &\defeq&& { \texttt{\LARGE(}} \psem{\sigma : \Gamma \partition \Delta, \Delta_0 \vdash e_0 : T_0} (\tau, \tau_0) {{\texttt{\LARGE,}}} \\
																																																																&&&\psem{\sigma : \Gamma \partition \Delta, \Delta_1 \vdash e_1 : T_1}(\tau, \tau_1) { \texttt{\LARGE)}} \\
 \psem{\sigma, \sigma' : \Gamma, \Gamma' \partition \Delta, \Delta' \vdash \cntrl{e}{T}{e_1 &\mapsto e_1' \\ &\cdots \\ e_n &\mapsto e_n'}{T'} : T'} (\tau,\tau') &\defeq&&  \psem{\sigma, \sigma_j : \Gamma, \Gamma_j \partition \Delta,\Delta' \vdash e_j' : T'} (\tau, \tau') \\
 \text{\ldots for } j, \sigma_j \text{ such that } \psem{\Gamma, \Delta \Vdash e : T} (\sigma, \tau) &=&& \psem{\varnothing : \varnothing \partition \Gamma_j \vdash e_j : T}(\sigma_j) \\
 \psem{\sigma : \Gamma \partition \Delta \vdash f\; e : T} (\tau) &\defeq&& \psem{\vdash f : T \rightsquigarrow T'} \left( \psem{\sigma : \Gamma \partition \Delta \vdash e : T} (\tau) \right) \\
 \psem{\pi\subcap{g}(\sigma) : \pi\subcap{g}(\Gamma) \partition \pi\subcap{d}(\Delta) \vdash e : T} \left(\pi\subcap{d}(\tau)\right) &\defeq&& \psem{\sigma : \Gamma \partition \Delta \vdash e : T} (\tau)
\end{alignat*}
\caption{Pure classical expression semantics}
\label{fig:sem-pure-c-expr}
\end{figure}

\begin{figure}[th]
\begin{alignat*}{2}
	\psem{\Delta, \Delta_{0} \Vdash e : T}(\tau,\tau_0) &\defeq&\;& \psem{\varnothing : \varnothing \partition \Delta \vdash e : T} (\tau) \\
	\psem{\Delta, \Delta_0, \Delta_1 \Vdash \pair{e_0}{e_1} : T_0 \otimes T_1}(\tau, \tau_0, \tau_1) &\defeq&& { \texttt{\LARGE(}} \psem{\Delta, \Delta_0 \Vdash e_0 : T_0}(\tau, \tau_0) {{\texttt{\LARGE,}}} \\
																																					 &&&\psem{\Delta, \Delta_1 \Vdash e_1 : T_1}(\tau, \tau_1) { \texttt{\LARGE)}} \\
\psem{\Delta_0,\Delta_1 \Vdash \trycatch{e_0}{e_1} : T}( \tau_0,\tau_1 ) &\defeq&& \begin{cases} \psem{\Delta_0 \Vdash e_0 : T}( \tau_0 ) &\text{ if defined} \\ \psem{\Delta_1 \Vdash e_1 : T}( \tau_1 ) &\text{ otherwise} \end{cases} \\
\psem{\Delta \Vdash f\;e : T'}(\tau) &\defeq&& \psem{\vdash f : T \Rrightarrow T'}\left(\psem{\Delta \Vdash e : T}(\tau)\right) \\
 \psem{\pi(\Delta) \Vdash e : T} (\pi(\tau) &\defeq&& \psem{\Delta \Vdash e : T} (\tau)
\end{alignat*}
\caption{Mixed classical expression semantics}
\label{fig:sem-mixed-c-expr}
\end{figure}

\begin{figure}[th]
\begin{align*}
	\psem{\vdash \lef{T_0}{T_1} : T_0 \rightsquigarrow T_0 \oplus T_1}(v) &\defeq \lef{T_0}{T_1}\: v \\
	\psem{\vdash \rit{T_0}{T_1} : T_1 \rightsquigarrow T_0 \oplus T_1}(v) &\defeq \rit{T_0}{T_1}\: v \\
	\psem{\vdash \lambda e \xmapsto{{\color{gray}T}} e' : T \rightsquigarrow T'}(v) &\defeq \psem{\varnothing : \varnothing \partition \Delta \vdash e' : T'} (\tau) \\
	\text{\ldots for } \tau \text{ such that } \psem{\varnothing : \varnothing \partition \Delta \vdash e : T} (\tau) &= v
\end{align*}
\caption{Pure classical program semantics}
\label{fig:sem-pure-c-prog}
\end{figure}

\begin{figure}[th]
\begin{align*}
	\psem{\vdash f : T \Rrightarrow T'}(v) &\defeq \psem{\vdash f : T \rightsquigarrow T'}(v) \\
	\psem{\Vdash \lambda e \xmapsto{{\color{gray}T}} e' : T \Rrightarrow T'}(v) &\defeq \psem{\Delta \Vdash e' : T'} (\tau) \\
	\text{\ldots for } \tau, \tau_0 \text{ such that } \psem{\varnothing : \varnothing \partition \Delta, \Delta_0 \vdash e : T} (\tau, \tau_0) &= v
\end{align*}
\caption{Mixed classical program semantics}
\label{fig:sem-mixed-c-prog}
\end{figure}

It is not immediately obvious that each of these definitions corresponds to a well-defined partial function.
In particular, the classical denotation of lambdas are defined ``for $\tau$ such that \ldots,'' and the classical denotation of \texttt{ctrl} similarly depends on some $j$ and $\sigma_j$, and these partial functions will be ill-defined if there is no satisfying assignment.
The uniqueness of the $\tau$ and $\sigma_j$ comes from the fact that the pure classical semantics is injective by definition.
We could have equivalently defined the pure classical denotation of lambdas in terms of inverse (partial) functions, which always exist for injective partial functions:
\[
\psem{\vdash \lambda e \xmapsto{{\color{gray}T}} e' : T \rightsquigarrow T'} = \psem{\varnothing : \varnothing \partition \Delta \vdash e' : T'} \circ \psem{\varnothing : \varnothing \partition \Delta \vdash e : T}^{-1}
\]
The fact that the $j$ used in \texttt{ctrl} semantics can be chosen uniquely comes from the \texttt{ortho} judgment, which ensures that all of the patterns $e_j$ are non-overlapping.

To be clear, these are \emph{partial} functions and may be undefined on some inputs.
In some cases, there may be no satisfying $\tau$.
For the partial function to be well-defined, it is only required for there to be \emph{at most} one output for every unique input.
Following existing conventions \cite[p.~16]{tapl}, we use the symbol $\uparrow$ to indicate when a partial function is undefined.
For example, $\vdash \lambda \zero \xmapsto{{\color{gray}\Bit}} \zero : \Bit \rightsquigarrow \Bit$ is a valid typing judgment because $\varnothing \partition \varnothing \vdash \zero : \Bit$ is valid, but $\psem{\vdash \lambda \zero \xmapsto{{\color{gray}\Bit}} \zero : \Bit \rightsquigarrow \Bit}(\one) \uparrow$ because there is no valuation $\tau$ such that $\psem{\varnothing : \varnothing \partition \varnothing \vdash \zero : \Bit}(\tau) = \one$.
Our definitions above assume implicitly that the semantic function is undefined if any of its subexpressions is undefined; for example, the semantics of a pair depends on the semantics of the subexpressions, and if either of these is undefined, then the pair semantics is undefined as well.

Because the denotational semantics is defined as a function of the typing judgment, we do not have an explicit type soundness theorem, as a valid semantics is already defined for any well-typed program.
Sometimes, one may desire a stronger notion of type soundness, ensuring not only that the semantics exists, but that it is defined for all possible inputs; that is, it is a \emph{total} function.
In these cases, one can use the \texttt{iso} judgment defined in Appendix~\ref{app:iso} to check for a class of error-free programs.
Conjecture~\ref{thm:iso} asserts that each classical program matched by the \texttt{iso} judgment has a total function as its classical semantics.

These definitions should make it clear that Qunity's pure/mixed distinction is very different from the classical/quantum distinction that it is designed to eliminate.
The pure/mixed distinction describes \emph{reversibility}, which can matter in both classical and quantum settings.

\subsection{Classical Sublanguage Proofs}
\label{app:mlpi}

In this appendix, we prove the theorems from Section~\ref{sec:classical-semantics}.
Using the definitions from Section~\ref{sec:denotation} and Appendix~\ref{app:classical-semantics}, we can restate Theorem~\ref{thm:classical-generalization} a bit more precisely:

\begingroup
\def\thetheorem{\ref{thm:classical-generalization}}
\begin{theorem}
	The semantics $\psem{\cdot}$ of the classical sublanguage of Qunity corresponds to the quantum semantics $\msem{\cdot}$ of the full language in the following way, where ``$\Longleftrightarrow$'' denotes ``if and only if'':
	\begin{align*}
		\psem{\sigma : \Gamma \partition \Delta \vdash e : T}(\tau) = v &\iff \msem{\sigma : \Gamma \partition \Delta \vdash e : T}\ket{\tau} = \ket{v} \\
		\psem{\sigma : \Gamma \partition \Delta \vdash e : T}(\tau) \uparrow &\iff \msem{\sigma : \Gamma \partition \Delta \vdash e : T}\ket{\tau} = 0 \\
		\psem{\Delta \Vdash e : T}(\tau) = v &\iff \msem{\Delta \Vdash e : T}\left(\op{\tau}{\tau}\right) = \op{v}{v} \\
		\psem{\Delta \Vdash e : T}(\tau) \uparrow &\iff \msem{\Delta \Vdash e : T}\left(\op{\tau}{\tau}\right) = 0 \\
		\psem{\vdash f : T \rightsquigarrow T'}(v) = v' &\iff \msem{\vdash f : T \rightsquigarrow T'}\ket{v} = \ket{v'} \\
		\psem{\vdash f : T \rightsquigarrow T'}(v) \uparrow &\iff \msem{\vdash f : T \rightsquigarrow T'}\ket{v} = 0 \\
		\psem{\vdash f : T \Rrightarrow T'}(v) = v' &\iff \msem{\vdash f : T \Rrightarrow T'}\left(\op{v}{v}\right) = \op{v'}{v'} \\
		\psem{\vdash f : T \Rrightarrow T'}(v) \uparrow &\iff \msem{\vdash f : T \Rrightarrow T'}\left(\op{v}{v}\right) = 0 \\
	\end{align*}
\end{theorem}
\addtocounter{theorem}{-1}
\endgroup
\begin{proof}
	We use (mutual) structural induction on the typing judgment.
	Because of the way that denotation $\msem{\cdot}$ is defined on the standard basis and extended by linearity, in most cases it is straightforward to see that this follows directly from the definitions.
	Here, we will consider the one case that is slightly more complicated: \textsc{T-Ctrl}.

			First, consider the case where $\psem{\sigma, \sigma' : \Gamma, \Gamma' \partition \Delta, \Delta' \vdash \cntrl{e}{T}{ \cdots }{T'} : T'} (\tau,\tau')$ is undefined.
			In this case, there exists no unique $j, \sigma_j$ such that $\psem{\Gamma, \Delta \Vdash e : T} (\sigma, \tau) = \psem{\varnothing : \varnothing \partition \Gamma_j \vdash e_j : T}(\sigma_j)$, so the summation in the \texttt{ctrl} denotation produces zero, as required.
			We can then assume that there exists a unique $j, \sigma_j, v$ such that $\psem{\Gamma, \Delta \Vdash e : T} (\sigma, \tau) = \psem{\varnothing : \varnothing \partition \Gamma_j \vdash e_j : T}(\sigma_j) = v$.
			Then:
			\begin{alignat*}{2}
				\bra{v} \msem{\Gamma, \Delta \Vdash e : T}\left( \op{\sigma, \tau}{\sigma, \tau} \right) \ket{v} &=&& 1 \\
				\bra{v'} \msem{\Gamma, \Delta \Vdash e : T}\left( \op{\sigma, \tau}{\sigma, \tau} \right) \ket{v'} &=&& 0 \text{ for all } v' \neq v \\ 
				\bra{\sigma_j} \msem{\varnothing : \varnothing \partition \Gamma_j \vdash e_j : T}^\dagger \ket v &=&& \bra{v} \msem{\varnothing : \varnothing \partition \Gamma_j \vdash e_j : T} \ket{\sigma_j} \\
																																																					&=&& 1 \\
				\bra{\sigma_k} \msem{\varnothing : \varnothing \partition \Gamma_k \vdash e_k : T}^\dagger \ket v &=&& 0 \text{ for all other } k, \sigma_k \\
				\msem{\sigma, \sigma_j : \Gamma, \Gamma_j \partition \Delta,\Delta' \vdash e_j' : T'} \ket{\tau, \tau'} &=&& \psem{\sigma, \sigma_j : \Gamma, \Gamma_j \partition \Delta,\Delta' \vdash e_j' : T'} (\tau, \tau') \\
	\msem{\sigma, \sigma' : \Gamma, \Gamma' \partition \Delta, \Delta' \vdash \cntrl{e}{T}{\cdots}{T'} : T'} \ket{\tau,\tau'} &=&& \sum_{v \in \V(T)} \bra{v} \msem{\Gamma, \Delta \Vdash e : T}\left( \op{\sigma, \tau}{\sigma, \tau} \right) \ket v \\
																																														 &&& \cdot \sum_{j=1}^n \sum_{\sigma_j \in \V(\Gamma_j)} \bra{\sigma_j} \msem{\varnothing : \varnothing \partition \Gamma_j \vdash e_j : T}^\dagger \ket v \\
																																														 &&&\cdot \msem{\sigma, \sigma_j : \Gamma, \Gamma_j \partition \Delta,\Delta' \vdash e_j' : T'} \ket{\tau, \tau'} \\
																																														 &=&& \psem{\sigma, \sigma_j : \Gamma, \Gamma_j \partition \Delta,\Delta' \vdash e_j' : T'} (\tau, \tau') \\
																																														 &=&& \psem{\sigma, \sigma' : \Gamma, \Gamma' \partition \Delta, \Delta' \vdash \cntrl{e}{}{\cdots}{} : T'} (\tau,\tau')
		 \end{alignat*}
\end{proof}

This classical sublanguage thus has a reasonable semantics that naturally generalizes to the standard basis in the quantum case.
The compilation procedure described in Section~\ref{sec:compilation} and Appendix~\ref{app:compilation} can similarly be adapted to compiling this sublanguage to classical hardware.
By applying the compilation procedure to this sublanguage, the resulting circuit consists entirely of classical gates: \textsc{swap}, \textsc{not}, and controlled gates.
The only exceptions are the \texttt{rphase} and \texttt{u3} constructs, which are excluded from the sublanguage.
The resulting classical logic circuit can be run (or simulated) on classical hardware.

One may be concerned that this classical sublanguage lacks the expressiveness expected from a typical classical programming language, especially with the strict \emph{relevant} typing requirements of the \texttt{ctrl} expression, the main way to perform computation conditional on some intermediate state.
To make it clear that this is not an issue, we now show that any program written in an existing classical language can be translated into Qunity.
We use the existing $\Pi$ and $\MLPi$ languages \cite{mlpi}, as the (ir)reversibility of these languages corresponds well to Qunity's pure and mixed typing, and a procedure already exists for translating a more typical classical functional language into these.
A combinator {\color{teal} $c$} in the $\Pi$ language is a (reversible) program that can be written as a pure classical Qunity program, while an \emph{arrow computation} {\color{teal} $a$} in the $\MLPi$ language is a (potentially irreversible) program that can be written as a mixed classical Qunity program.
Figure~\ref{fig:mlpi-syntax} shows the syntax of these languages.
Throughout this section, we color the text of these terms to make them more visually distinct from Qunity's syntax, and we use $\MLPi$ types {\color{teal} $b$} interchangeably with Qunity types $T$.

\begin{figure}[ht]
	\begin{align*}
		{\color{teal} b} &\defeqq {\color{teal} 1} \mid {\color{teal} b + b} \mid {\color{teal} b \times b} \\
		{\color{teal} iso} &\defeqq {\color{teal} swap^{+}} \mid {\color{teal} assocl^{+}} \mid {\color{teal} assocr^{+}} \mid {\color{teal} unite} \mid {\color{teal} uniti} \mid {\color{teal} swap^{\times}} \mid {\color{teal} assocl^{\times}} \mid {\color{teal} assocr^{\times}} \mid {\color{teal} distrib} \mid {\color{teal} factor} \\
		{\color{teal} c} &\defeqq {\color{teal} iso} \mid {\color{teal} id} \mid {\color{teal} sym\: c} \mid {\color{teal} c; c} \mid {\color{teal} c + c} \mid {\color{teal} c \times c} \\
		{\color{teal} a} &\defeqq {\color{teal} iso} \mid {\color{teal} a + a} \mid {\color{teal} a \times a} \mid {\color{teal} a; a} \mid {\color{teal} \texttt{arr}\: a} \mid {\color{teal} a >\!>\!> a} \mid {\color{teal} \texttt{left}\: a} \mid {\color{teal} \texttt{first}\: a} \mid {\color{teal} \texttt{create}_b} \mid {\color{teal} \texttt{erase}} \\
	\end{align*}
	\caption{$\Pi$ and $\MLPi$ languages \cite{mlpi}}
	\label{fig:mlpi-syntax}
\end{figure}

First, we show how a $\Pi$ combinator is translated into a pure classical Qunity program.
We will write $\qiso{T}{T'}({\color{teal}c})$ for the Qunity translation of a combinator ${\color{teal}c}$.
For all of the primitive isomorphisms ${\color{teal}iso}$, the \texttt{match} construct described in Section~\ref{sec:match} can be used to implement $\qiso{T}{T'}({\color{teal}iso})$ directly from the semantic definitions \cite[p.~76]{mlpi}.
See Figure~\ref{fig:qiso} for some representative examples.
\begin{figure}[ht]
\[
	\match{T_1 \oplus (T_2 \oplus T_3)}{
	\lef{T_1}{(T_2 \oplus T_3)}\: v_1 &\mapsto \lef{(T_1 \oplus T_2)}{T_3}\left( \lef{T_1}{T_2}\: v_1 \right) \\
  \rit{T_1}{(T_2 \oplus T_3)}\left(\lef{T_2}{T_3}\: v_2\right) &\mapsto \lef{(T_1 \oplus T_2)}{T_3}\left( \rit{T_1}{T_2}\: v_2 \right) \\
  \rit{T_1}{(T_2 \oplus T_3)}\left(\rit{T_2}{T_3}\: v_3\right) &\mapsto \rit{(T_1 \oplus T_2)}{T_3}\: v_3}
	{(T_1 \oplus T_2) \oplus T_3}
\]
\vspace{5mm}
\[
	\match{(T_1 \oplus T_2) \otimes T_3}{
	\pair{\lef{T_1}{T_2}\: v_1}{v_3} &\mapsto \lef{(T_1 \otimes T_3)}{(T_2 \otimes T_3)} \pair{v_1}{v_3} \\
	\pair{\rit{T_1}{T_2}\: v_2}{v_3} &\mapsto \rit{(T_1 \otimes T_3)}{(T_2 \otimes T_3)} \pair{v_2}{v_3}}
	{(T_1 \otimes T_3) \oplus (T_2 \otimes T_3)}
\]
\caption{Definitions of $\qiso{T_1 \oplus (T_2 \oplus T_3)}{(T_1 \oplus T_2) \oplus T_3}({\color{teal}assocl^{+}})$ and $\qiso{(T_1 \oplus T_2) \otimes T_3}{(T_1 \otimes T_3) \oplus (T_2 \otimes T_3)}({\color{teal}distrib})$}
	\label{fig:qiso}
\end{figure}

The other combinators are also straightforward to implement in Qunity, using the direct sum defined in Section~\ref{sec:match}:
\begin{align*}
	\qiso{T}{T}({\color{teal}id}) &\defeq \lambda x \xmapsto{{\color{gray}T}} x \\
	\qiso{T}{T'}({\color{teal}sym\: c}) &\defeq \qiso{T'}{T}({\color{teal}c})^\dagger \\
	\qiso{T_1}{T_2}({\color{teal}c_1 ; c_2}) &\defeq \lambda x \xmapsto{{\color{gray}T_1}} x \triangleright \qiso{T_1}{T}({\color{teal}c_1}) \triangleright \qiso{T}{T_2}({\color{teal}c_2}) \\
	\qiso{T_0 \oplus T_1}{T_0' \oplus T_1'}({\color{teal}c_0 + c_1}) &\defeq \qiso{T_0}{T_0'}({\color{teal}c_0}) \oplus \qiso{T_1}{T_1'}({\color{teal}c_1}) \\
	\qiso{T_0 \otimes T_1}{T_0' \otimes T_1'}({\color{teal}c_0 \times c_1}) &\defeq \lambda \pair{x_0}{x_1} \xmapsto{{\color{gray}T_0 \otimes T_1}} \pair{x_0 \triangleright \qiso{T_0}{T_0'}({\color{teal}c_0})}{x_1 \triangleright \qiso{T_1}{T_1'}({\color{teal}c_1})}
\end{align*}

These definitions show that purely-typed programs in Qunity's classical sublanguage are powerful enough to express any of $\Pi$'s reversible programs:

\begin{lemma}[pure classical sublanguage expressiveness]
	\label{lem:pi}
	Let ${\color{teal}c}$ be a $\Pi$ combinator and let $f = \qiso{T}{T'}({\color{teal}c})$.
	This transformation preserves typing and semantics:
	\begin{itemize}
		\item Using the $\Pi$ typing, ${\color{teal}c} : T \leftrightarrow T'$ implies $\vdash f : T \rightsquigarrow T'$
		\item Using the $\Pi$ evaluation rules, ${\color{teal}c}\: v \mapsto v'$ implies that $\psem{\vdash f : T \rightsquigarrow T'}(v) = v'$.
	\end{itemize}
\end{lemma}
\begin{proof}
	For the primitive isomorphisms, the typing and semantics described in Section~\ref{sec:match} should make both of these fairly clear.
	All of the patterns on the left and right sides of the \texttt{match} expressions are fairly simple, and it is easy to verify that they have the required types and meet the orthogonality requirements for the typing to be preserved.
	Our transformation comes directly from the $\Pi$ isomorphism semantics, so the \texttt{match} semantics preserve this behavior.
	The typing and semantics of the other combinators also follows directly from the definitions.
\end{proof}

Next, we can show that irreversible programs from $\MLPi$ can also be translated into Qunity, this time as \emph{mixed} programs.
We write $\qarr{T}{T'}({\color{teal}a})$ for the translation of $\MLPi$ arrow computation ${\color{teal}a}$ into a Qunity program of type $T \Rrightarrow T'$, defined as follows on all $\MLPi$ terms for which the semantics is defined:
\begin{alignat*}{2}
	\qarr{T}{T'}(\texttt{arr}\: {\color{teal}a}) &\defeq&\;& \qiso{T}{T'}({\color{teal}a}) \\
	\qarr{T_1}{T_2}({\color{teal}a_1 >\!>\!> a_2}) &\defeq&& \lambda x \xmapsto{{\color{gray}T_1}} x \triangleright \qarr{T_1}{T}({\color{teal}a_1}) \triangleright \qarr{T}{T_2}({\color{teal}a_2}) \\
	\qarr{T_1 \otimes T_3}{T_2 \otimes T_3}({\color{teal}\texttt{first}\: a}) &\defeq&& \lambda \pair{x_1}{x_3} \xmapsto{{\color{gray}T_1 \otimes T_3}} \pair{x_1 \triangleright \qarr{T_1}{T_2}({\color{teal}a})}{x_3} \\
	\qarr{T_1 \oplus T_3}{T_2 \oplus T_3}({\color{teal}\texttt{left}\: a}) &\defeq&& \lambda x \xmapsto{{\color{gray}T_1 \oplus T_3}} \\
																													 &&&x \triangleright \match{T_1 \oplus T_3}{\lef{T_1}{T_3}\: x_1 &\mapsto \pair{\zero}{\pair{x_1}{\phi(T_3)}} \\ \rit{T_1}{T_3}\: x_3 &\mapsto \pair{\one}{\pair{\phi(T_1)}{x_3}}}{\Bit \otimes (T_1 \otimes T_3)} \\
																													 &&&\triangleright \lambda \pair{x_b}{\pair{x_1}{x_3}} \xmapsto{{\color{gray}\Bit \otimes (T_1 \otimes T_3)}} \\
																													 &&&\pair{x_b}{\pair{x_1 \triangleright \qarr{T_1}{T_2}({\color{teal}a})}{x_3}} \\
																													 &&&\triangleright \lambda \pair{x_b}{\pair{x_2}{x_3}} \xmapsto{{\color{gray}\Bit \otimes (T_2 \otimes T_3)}} \\
																													 &&& \cntrl{x_b}{}{\zero &\mapsto \pair{\lef{T_2}{T_3}\:x_2}{\pair{x_b}{\rit{T_2}{T_3}\: x_3}} \\ \one &\mapsto \pair{\rit{T_2}{T_3}\: x_3}{\pair{x_b}{\lef{T_2}{T_3}\: x_2}}}{} \\
																													 &&&\triangleright \fst{(T_2 \oplus T_3)}{(\Bit \otimes (T_2 \oplus T_3))} \\
\qarr{\Unit}{T}({\color{teal}\texttt{create}_T}) &\defeq&& \lambda \unit \xmapsto{{\color{gray}\Unit}} \phi(T) \\
\qarr{T}{\Unit}({\color{teal}\texttt{erase}}) &\defeq&& \lambda x \xmapsto{{\color{gray}T}} \unit
\end{alignat*}

Because the \texttt{ctrl} expression disallows any discarding of information, the {\color{teal}\texttt{left}} definition has to be a bit complicated.
The program $\qarr{T_1}{T_2}({\color{teal}a})$ would be ill-typed if placed directly in the branch of a \texttt{ctrl} expression, so instead we use a \texttt{match} expression to convert the sum type into a product type, using James and Sabry's $\phi$ function to create dummy values for the absent type.
This allows us to then run the $\qarr{T_1}{T_2}({\color{teal}a})$ program within a pair construct instead.
We then use a \texttt{ctrl} expression and the \texttt{fst} function to rearrange these values, discarding the unneeded ones.

We can finally prove the other theorem from Section~\ref{sec:classical-semantics}, demonstrating Qunity's power for implementing \emph{irreversible} classical programs as well:

\begingroup
\def\thetheorem{\ref{thm:mlpi}}
\begin{theorem}
	Let ${\color{teal}a}$ be an $\MLPi$ arrow computation and let $f = \qarr{T}{T'}({\color{teal}a})$.
	This transformation preserves typing and semantics:
	\begin{itemize}
		\item Using the $\MLPi$ typing, ${\color{teal}a} : T \rightsquigarrow T'$ implies $\vdash f : T \Rrightarrow T'$
		\item Using the $\MLPi$ evaluation rules, ${\color{teal}a}\: v \mapsto_{\textsc{ml}} v'$ implies $\psem{\vdash f : T \Rrightarrow T'}(v) = v'$.
	\end{itemize}
\end{theorem}
\addtocounter{theorem}{-1}
\endgroup
\begin{proof}
	Using structural induction, in most of these cases, it is easy to verify that the required properties follow directly from the definitions.
	We will show  the evaluation of the one case that is more complicated: {\color{teal}\texttt{left}}.
	The input to ${\color{teal}\texttt{left}\:a}$ is a sum type, so we will consider the two input cases separately.
	First, suppose ${\color{teal}a}\:v_1 \mapsto_{\textsc{ml}} v_2$ by $\MLPi$ evaluation rules and by induction hypothesis that $\psem{\vdash \qarr{T_1}{T_2}({\color{teal}a}) : T_1 \Rrightarrow T_2}(v_1) = v_2$, and we will show that $\psem{\vdash \qarr{T_1 \oplus T_3}{T_2 \oplus T_3}({\color{teal}\texttt{left}\:a}) : T_1 \oplus T_3 \Rrightarrow T_2 \oplus T_3}(\lef{T_1}{T_3}\: v_1) = \lef{T_2}{T_3}\: v_2$.
	This program is effectively sequential, and we will use the $\Rightarrow$ symbol to describe evaluating one step of the program, reserving the $\mapsto$ symbol for bindings within valuations.
	\begin{align*}
		\lef{T_1}{T_3}\:v_1
		&\Rightarrow (x \mapsto \lef{T_1}{T_3}\:v_1) \\
		&\Rightarrow \lef{T_1}{T_3}\:v_1 \\
		&\Rightarrow \pair{\zero}{\pair{v_1}{\phi(T_3)}} \\
		&\Rightarrow (x_b \mapsto \zero, x_1 \mapsto v_1, x_3 \mapsto \phi(T_3)) \\
		&\Rightarrow \pair{\zero}{\pair{v_2}{\phi(T_3)}} \\
		&\Rightarrow (x_b \mapsto \zero, x_2 \mapsto v_2, x_3 \mapsto \phi(T_3)) \\
		&\Rightarrow \pair{\lef{T_2}{T_3}\:v_2}{\pair{\zero}{\rit{T_2}{T_3}\:v_3}} \\
		&\Rightarrow \lef{T_2}{T_3}\:v_2
	\end{align*}

	Next, we show that $\psem{\vdash \qarr{T_1 \oplus T_3}{T_2 \oplus T_3}({\color{teal}\texttt{left}\:a}) : T_1 \oplus T_3 \Rrightarrow T_2 \oplus T_3}(\rit{T_1}{T_3}\: v_3) = \rit{T_2}{T_3}\: v_3$.
	For clarity, let $v_2 = \psem{\vdash \qarr{T_1}{T_2}({\color{teal}a}) : T_1 \Rrightarrow T_2}(\phi(T_1))$, which must be defined because the $\MLPi$ combinator semantics is defined on all inputs.
	\begin{align*}
		\rit{T_1}{T_3}\:v_3
		&\Rightarrow (x \mapsto \rit{T_1}{T_3}\:v_3) \\
		&\Rightarrow \rit{T_1}{T_3}\:v_3 \\
		&\Rightarrow \pair{\one}{\pair{\phi(T_1)}{v_3}} \\
		&\Rightarrow (x_b \mapsto \one, x_1 \mapsto \phi(T_1), x_3 \mapsto v_3) \\
		&\Rightarrow \pair{\one}{\pair{v_2}{v_3}} \\
		&\Rightarrow (x_b \mapsto \one, x_2 \mapsto v_2, x_3 \mapsto v_3) \\
		&\Rightarrow \pair{\rit{T_2}{T_3}\:v_3}{\pair{\one}{\lef{T_2}{T_3}\:v_2}} \\
		&\Rightarrow \rit{T_2}{T_3}\:v_3
	\end{align*}
\end{proof}

There is already a way of translating programs from a more typical \texttt{let}-based first-order classical functional programming language into their $\MLPi$ language \cite{mlpi}, so programs from this \texttt{let}-based language can be translated into Qunity as well.
Because Qunity is restricted to types of finite cardinality, it cannot implement the arbitrarily-sized natural numbers and lists one is used to in classical programming languages, but we have shown throughout this paper how a thin layer of metaprogramming can be used to create types parameterized by a compile-time constant, like the bounded-but-variable-length list used in Section~\ref{sec:quantum-walk}.

\section{Spanning judgments}
\label{app:spanning}

Qunity's orthogonality judgment, used to type \texttt{ctrl} expressions, depends on a ``spanning'' judgment, defined by the inference rules in Figure~\ref{fig:spanning}.
In these rules we write $\FV(e)$ to compute the set of free variables in $e$, which is straightforward since expressions $e$ have no binders.
The semantics of a spanning judgment is described by Lemma~\ref{lem:spanning-sem}.

\begin{figure}[t]
\[
	\inference{}{\spanning{\Void}{}}[\textsc{S-Void}]
	\quad
	\inference{}{\spanning{\Unit}{\unit}}[\textsc{S-Unit}]
	\quad
	\inference{}{\spanning{T}{x}}[\textsc{S-Var}]
	\]
\vspace{2mm}
	\[
	\inference{\spanning{T}{e_1, \ldots, e_n} \qquad \spanning{T'}{e_1', \ldots, e_{n'}'}}{\spanning{T \oplus T'}{\lef{T}{T'}e_1, &\ldots, \lef{T}{T'}e_n, \\ \rit{T}{T'}e_1', &\ldots, \rit{T}{T'}e_{n'}' }}[\textsc{S-Sum}]
\]
\vspace{2mm}
\[
	\inference{\spanning{T}{e_1, \ldots, e_m} \qquad \spanning{T'}{e_{j,1}', \ldots, e_{j,n_j}'} \text{ for all } j \\ \FV(e_j) \cap \bigcup_{k=1}^{n_j} \FV(e_{j,k}') = \varnothing \text{ for all } j}{\spanning{T\otimes T'}{ \pair{e_1}{e_{1,1}'}, &\ldots, \pair{e_1}{e_{1,n_1}'}, \\ &\ldots, \\ \pair{e_m}{e_{m,1}'}, &\ldots, \pair{e_m}{e_{m,n_m}'} }}[\textsc{S-Pair}]
\]
\vspace{2mm}
\[
	\inference{\spanning{T}{ e_1, \ldots, e_n } }{\spanning{T}{ \pi(e_1, \ldots, e_n) }}[\textsc{S-Perm}]
\]
\caption{Spanning inference rules}
\label{fig:spanning}
\end{figure}

\section{Isometry Checking and Type Soundness}
\label{app:iso}

Qunity's semantics uses norm-decreasing operators and trace-decreasing superoperators, which can be interpreted operationally as a sort of ``error.'' 
This means that even well-typed Qunity programs are not necessarily error-free, which seems to indicate an unsound type system.
To assuage these concerns, we use this appendix to discuss an additional judgment ``iso'' that can be used to statically determine whether a program belongs to a checkable class of error-free programs.

We write $\iso(e)$ or $\iso(f)$ to denote that expression $e$ or program $f$ is isometric.
As will be made precise by Theorem~\ref{thm:iso}, this indicates that an expression or program has norm-preserving pure semantics or trace-preserving mixed semantics.
This efficiently checkable property is defined by the inference rules in Figure~\ref{fig:iso}.

Our isometry judgment depends on another: we write ``$\classical(e)$'' to denote that expression $e$ is an element of the classical sublanguage of Qunity defined in the previous appendix.

\begin{figure}[t]
\[
	\inference{}{\iso(\unit)}[\textsc{I-Unit}]
	\quad
	\inference{}{\iso(x)}[\textsc{I-Var}]
	\quad
	\inference{\iso(e_0) \qquad \iso(e_1)}{\iso(\pair{e_0}{e_1})}[\textsc{I-Pair}]
\]
\vspace{2mm}
\[
\inference{\classical(e) \qquad \iso(e) \qquad \spanning{T}{e_1, \ldots, e_n} \qquad \iso(e_1') \quad \cdots \quad \iso(e_n')}{\iso\left(\cntrl{e}{T}{e_1 &\mapsto e_1' \\ &\cdots \\ e_n &\mapsto e_n'}{T'}\right)}[\textsc{I-Ctrl}]
\]
\vspace{2mm}
\[
	\inference{\iso(e)}{\iso(\trycatch{e}{e'})}[\textsc{I-Try}]
	\quad
	\inference{\iso(e')}{\iso(\trycatch{e}{e'})}[\textsc{I-Catch}]
	\quad
	\inference{\iso(f) \qquad \iso(e)}{\iso(f \: e)}[\textsc{I-App}]
\]
\vspace{2mm}
\[
	\inference{}{\iso(\uthree{r_\theta}{r_\phi}{r_\lambda})}[\textsc{I-Gate}]
	\quad
	\inference{}{\iso(\lef{T_0}{T_1})}[\textsc{I-Left}]
	\quad
	\inference{}{\iso(\rit{T_0}{T_1})}[\textsc{I-Right}]
\]
\vspace{2mm}
\[
	\inference{\spanning{T}{e} \qquad \iso(e')}{\iso(\lambda e \xmapsto{{\color{gray}T}} e')}[\textsc{I-Abs}]
	\quad
	\inference{}{\iso\left(\rphase{T}{e}{r}{r'}\right)}[\textsc{I-Rphase}]
\]
\caption{Isometry inference rules}
\label{fig:iso}
\end{figure}

\begin{conjecture}[norm- and trace-preserving semantics]
	\label{thm:iso}
	The iso judgment corresponds to norm- or trace-preserving semantics:
	\begin{itemize}
		\item Whenever $(\Gamma \partition \Delta \vdash e : T)$ and $\iso(e)$ are valid, $\msem{\sigma : \Gamma \partition \Delta \vdash e : T}^\dagger \msem{\sigma : \Gamma \partition \Delta \vdash e : T} = \mathbb{I}$.
			If $e$ is classical, then its classical semantics (Appendix~\ref{app:sublanguage}) is a total function rather than a partial function; that is, $\psem{\sigma : \Gamma \partition \Delta \vdash e : T}(\tau) \downarrow$ for all $\tau \in \V(\Delta)$.
		\item Whenever $(\vdash f : T \rightsquigarrow T')$ and $\iso(f)$ are valid, $\msem{\vdash f : T \rightsquigarrow T'}^\dagger \msem{\vdash f : T \rightsquigarrow T'} = \mathbb{I}$.
			If $f$ is classical, then $\psem{\vdash f : T \rightsquigarrow T'}$ is a total function.
		\item Whenever $(\Delta \Vdash e : T)$ and $\iso(e)$ are valid, $\tr(\msem{\Delta \Vdash e : T}(\rho)) = \tr(\rho)$.
			If $e$ is classical, then $\psem{\Delta \Vdash e : T}$ is a total function.
		\item Whenever $(\vdash f : T \Rrightarrow T')$ and $\iso(f)$ are valid, $\tr(\msem{\vdash f : T \Rrightarrow T'}(\rho)) = \tr(\rho)$.
			If $f$ is classical, then $\psem{\vdash f : T \Rrightarrow T'}$ is a total function.
	\end{itemize}
\end{conjecture}

\clearpage

\section{Additional Example: The Deutsch-Jozsa algorithm}
\label{sec:dj}

Given an oracle computing the (constant or balanced) function $f : \{\zero, \one\}^n \to \{\zero, \one\}$ for some $n$, the Deutsch-Jozsa algorithm can be framed as three steps:
\begin{enumerate}
	\item
		Prepare the state $\ket{+}^{\otimes n}$.
	\item
		Coherently perform the map:
		\[
			\ket{x} \mapsto (-1)^{f(x)} \ket{x}
		\]
	\item
		Perform a measurement that determines whether the state is still $\ket{+}^{\otimes n}$.
		If $\ket{+}^{\otimes n}$ is measured, then output ``\one'' to indicate that $f$ is constant.
		If some orthogonal state is measured, then output ``\zero'' to indicate that $f$ is balanced.
\end{enumerate}

\begin{figure}[th]
\begin{alignat*}{3}
	\texttt{deutsch-jozsa}_n(f) &\defeq&\;& \left( \begin{aligned} &\texttt{let } x \texttt{ =}_{\color{gray} \Bit^{\otimes n}}\; \texttt{plus}^{\otimes n} \texttt{ in } \\
															 &\cntrl{(f\: x)}{\Bit^{\otimes n}}{\zero &\mapsto x \\ \one &\mapsto x \triangleright \gphase{\Bit^{\otimes n}}{\pi}}{\Bit^{\otimes n}} \end{aligned} \right) \\
															 &&& \triangleright \equals{\Bit^{\otimes n}}(\texttt{plus}^{\otimes n}) \\
\end{alignat*}
\caption{The Deutsch-Jozsa algorithm implemented in Qunity}
\label{fig:dj}
\end{figure}

Like in Deutsch's algorithm in Section~\ref{sec:deutsch}, this succinct interpretation is straightforward to implement in Qunity even if $f$ is computed via an arbitrary quantum algorithm.
Figure~\ref{fig:dj} shows our Qunity implementation, and Qunity's type system should assign the following type to this program:
\begin{conjecture}
\[
	\inference{\vdash f : \Bit \Rrightarrow \Bit^{\otimes n}}{\varnothing \Vdash \texttt{deutsch-jozsa}_n(f) : \Bit}
\]
\end{conjecture}

\clearpage

\section{Real constants}
\label{app:real}

Qunity syntax, defined in Figure~\ref{fig:syntax}, depends on some representation of real constants $r$.
These constants may involve computation using standard (classical) scientific calculator operations, but their values cannot depend on the results of quantum computation because they are compile-time constants.
In Figure~\ref{fig:real}, we present one possible definition for these real constants.
This definition is largely inspired by Open\textsc{Qasm}'s \cite{openqasm3} real constants, making compilation to Open\textsc{Qasm} more straightforward.
Like Open\textsc{Qasm}, we use the letter ``$\varepsilon$'' to stand for Euler's number $e \approx 2.7$, in our case to avoid conflict with the letter ``$e$'' used to represent expressions.

\begin{figure}[ht]
\begin{alignat*}{3}
						n \in&&& \{0, 1, \ldots \} &\quad \textit{(natural)} \\
	r \defeqq &\quad&&&\textit{(real)} \\
						|&&& \pi &\quad \textit{(pi)} \\
						|&&& \varepsilon &\quad \textit{(Euler's number)} \\
						|&&& n &\quad \textit{(natural)} \\
	|&&& \texttt{-} r &\quad\textit{(negation)} \\
	|&&& r \texttt{ + } r &\quad\textit{(addition)} \\
	|&&& r \texttt{ * } r &\quad\textit{(multiplication)} \\
	|&&& r \texttt{ / } r &\quad\textit{(division)} \\
	|&&& \texttt{sin(} r \texttt{)} &\quad\textit{(sine)} \\
	|&&& \texttt{cos(} r \texttt{)} &\quad\textit{(cosine)} \\
	|&&& \texttt{tan(} r \texttt{)} &\quad\textit{(tangent)} \\
	|&&& \texttt{arcsin(} r \texttt{)} &\quad\textit{(inverse sine)} \\
	|&&& \texttt{arccos(} r \texttt{)} &\quad\textit{(inverse cosine)} \\
	|&&& \texttt{arctan(} r \texttt{)} &\quad\textit{(inverse tangent)} \\
	|&&& \texttt{exp(} r \texttt{)} &\quad\textit{(exponential function)} \\
	|&&& \texttt{ln(} r \texttt{)} &\quad\textit{(natural logarithm)} \\
	|&&& \texttt{sqrt(} r \texttt{)} &\quad\textit{(square root)} \\
\end{alignat*}
\caption{Real constants}
\label{fig:real}
\end{figure}

This syntax does allow for invalid constants to be constructed, like ``$1 / 0$.''
It may seem problematic that our semantic proofs ignore this possibility and assume that any real constant $r$ is a valid real number.
However, Qunity's design ensures that these real computations never must be done at runtime, so in practice any invalid constants can be caught before running on quantum hardware.

\clearpage

\section{Well-Defined Semantics}
\label{app:well-defined}

Qunity's semantics $\msem{\cdot}$ is defined as a function of a valid typing judgment.
More precisely, it is defined as a recursive function of the \emph{proof of validity} for a typing judgment.
However, our notation treats it as a function of the typing judgment itself.
For this notation to be well-defined, we must show that the derived semantics is independent of the particular proof used to prove a typing judgment.

The induction strategy used in our proof depends on the notion of \emph{proof depth}, and on semantics being well-defined up to a particular proof depth.

\begin{definition}[proof depth]
	Suppose $P$ is a proof of any typing judgment.
	If $P$ is not proven in terms of any other typing judgments, then $\depth(P) = 0$.
	If $P$ is proven in terms of typing proofs $\{P_1, \ldots, P_n\}$, then $\depth(P) = 1 + \max\{\depth(P_1), \ldots, \depth(P_n)\}$.
\end{definition}

\begin{definition}[equivalent proofs]
	Suppose $P_1$ and $P_2$ are two typing proofs.
	We say that $P_1$ and $P_2$ are equivalent proofs if any of the following are true:
	\begin{itemize}
		\item $P_1$ is a proof of $(\pi\subcap{g1}(\Gamma) \partition \pi\subcap{d1}(\Delta) \vdash e : T)$ and $P_2$ is a proof of $(\pi\subcap{g2}(\Gamma) \partition \pi\subcap{d2}(\Delta) \vdash e : T)$ for some list permutations $\pi\subcap{g1}, \pi\subcap{d1}, \pi\subcap{g2}, \pi\subcap{d2}$.
		\item $P_1$ is a proof of $(\pi_1(\Delta) \Vdash e : T)$ and $P_2$ is a proof of $(\pi_2(\Delta) \Vdash e : T)$ for some list permutations $\pi_1$ and $\pi_2$.
		\item $P_1$ and $P_2$ are two proofs of the same program typing judgment $(\vdash f : F)$.
	\end{itemize}
\end{definition}

\begin{definition}[equivalent semantics]
	Suppose $P_1$ and $P_2$ are equivalent proofs.
	We say that $P_1$ and $P_2$ have equivalent semantics if the following are true:
	\begin{itemize}
		\item If $P_1$ is a proof of $(\pi\subcap{g1}(\Gamma) \partition \pi\subcap{d1}(\Delta) \vdash e : T)$ and $P_2$ is a proof of $(\pi\subcap{g2}(\Gamma) \partition \pi\subcap{d2}(\Delta) \vdash e : T)$, then $\msem{\pi\subcap{g1}(\sigma) : P_1}\ket{\pi\subcap{d1}(\tau)} = \msem{\pi\subcap{g2}(\sigma) : P_2}\ket{\pi\subcap{d2}(\tau)}$ for all $\sigma \in \V(\Gamma), \tau \in \V(\Delta)$.
		\item If $P_1$ is a proof of $(\pi_1(\Delta) \Vdash e : T)$ and $P_2$ is a proof of $(\pi_2(\Delta) \Vdash e : T)$, then $\msem{P_1}\left(\op{\pi_1(\tau')}{\pi_1(\tau)}\right) = \msem{P_2}\left(\op{\pi_2(\tau')}{\pi_2(\tau)}\right)$ for all $\tau, \tau' \in \V(\Delta)$.
		\item If $P_1$ and $P_2$ are two proofs of $(\vdash f : F)$, then $\msem{P_1} = \msem{P_2}$.
	\end{itemize}
	We say that Qunity semantics is well-defined up to depth $n$ if $P_1$ and $P_2$ have equivalent semantics for any equivalent proofs with $\depth(P_1) + \depth(P_2) \leq n$.
\end{definition}

We now use these definitions to prove that Qunity's semantics is well-defined.

\begin{proof}[Proof of Theorem~\ref{thm:well-defined}]
	We will prove by induction that Qunity semantics is well-defined up to depth $n$ for any natural number $n$.
	\begin{itemize}
		\item Base case: Qunity semantics is well-defined up to depth 0.
			Consider two typing proofs $P_1$ and $P_2$ such that $\depth(P_1) = \depth(P_2) = 0$.
			In this case, both $P_1$ and $P_2$ are proven using a single inference rule.
			The possible cases for the rules used to prove $P_1$ and $P_2$ are: \textsc{T-Unit}, \textsc{T-Cvar}, \textsc{T-Qvar}, \textsc{T-Gate}, \textsc{T-Left}, \textsc{T-Right}.
			We can safely assume that both $P_1$ and $P_2$ were proven using the same inference rule, as each rule applies to a different syntax, except for \textsc{T-Cvar} and \textsc{T-Qvar}, which apply to different contexts.
			$P_1$ and $P_2$ are thus identical proofs (up to permutation of contexts) and thus have equivalent semantics.
		\item Qunity semantics is well-defined up to depth $(n+1)$. Assume by inductive hypothesis that Qunity semantics is well-defined up to depth $n$.
			Suppose $P_1$ and $P_2$ are two typing proofs such that $\depth(P_1) + \depth(P_2) \leq n + 1$, and consider all of the cases for the rules used to prove $P_1$ and $P_2$:
			\begin{itemize}
				\item The same rule is used to prove both proofs.
					Both proofs will also use equivalent subproofs because our typing relation is deterministic and fixes the typing context.
					Our inductive hypothesis says that these subproofs have equivalent semantics, and thus $P_1$ and $P_2$ have equivalent semantics.
				\item Either proof uses \textsc{T-PurePerm}. Without loss of generality, assume it is $P_1$ that proves $(\pi\subcap{g1}(\Gamma) \partition \pi\subcap{d1}(\Delta) \vdash e : T)$ using the \textsc{T-PurePerm} rule.
					$P_1$ is then proven in terms of another proof $P_0$ of $(\pi\subcap{g0}(\Gamma) \partition \pi\subcap{d0}(\pi\subcap{d1}(\Delta)) \vdash e : T)$ for some list permutations $\pi\subcap{g0}$ and $\pi\subcap{d0}$.
					The semantics of \textsc{T-PurePerm} is defined so that $\msem{\pi\subcap{g0}(\sigma) : P_0}\ket{\pi\subcap{d0}(\pi\subcap{d1}(\tau))} = \msem{P_1}\ket{\pi\subcap{d1}(\tau)}$.
					Our inductive hypothesis lets us assume that $\msem{\pi \subcap{g0}(\sigma) : P_0}\ket{\pi\subcap{d0}(\pi\subcap{d1}(\tau))} = \msem{\pi\subcap{g2}(\sigma) : P_2}\ket{\pi\subcap{d2}(\tau)}$, which then implies that $\msem{\pi\subcap{g1}(\sigma) : P_1}\ket{\pi\subcap{d1}(\tau)} = \msem{\pi\subcap{g2}(\sigma) : P_2}\ket{\pi\subcap{d2}(\tau)}$.
				\item Either proof uses \textsc{T-MixedPerm}. The denotations are then equivalent for the same reason as the \textsc{T-PurePerm} case.
				\item \textsc{T-Mix} and \textsc{T-MixedPair}.
					The judgment being proven is of the form $(\pi(\Delta, \Delta_0, \Delta_1) \Vdash \pair{e_0}{e_1} : T_0 \otimes T_1)$.
					Without loss of generality, assume $P_1$ uses \textsc{T-Mix} to prove $(\pi_1(\Delta, \Delta_0, \Delta_1) \Vdash \pair{e_0}{e_1} : T_0 \otimes T_1)$ and $P_2$ uses \textsc{T-MixedPair} to prove $(\pi_2(\Delta), \pi_{20}(\Delta_0), \pi_{21}(\Delta_1)) \Vdash \pair{e_0}{e_1} : T_0 \otimes T_1)$.
					Then $P_1$ is proven in terms of some proof $P_3$ of $(\varnothing \partition \pi_1(\Delta, \Delta_0, \Delta_1) \vdash \pair{e_0}{e_1} : T_0 \otimes T_1)$, and $P_2$ is proven in terms of some proofs $P_{40}$ of $(\pi_2(\Delta), \pi_{20}(\Delta_0) \Vdash e_0 : T_0)$ and $P_{41}$ of $(\pi_2(\Delta), \pi_{21}(\Delta_1) \Vdash e_1 : T_1)$.
					It is not hard to see that both the $P_3$ proof and the $P_4$ proofs can be proven in terms of some proofs $P_{50}$ of $(\varnothing \partition \pi_5(\Delta), \pi_{50}(\Delta_0) \vdash e_0 : T_0)$ and $P_{51}$ of $(\varnothing \partition \pi_5(\Delta), \pi_{51}(\Delta_1) \vdash e_1 : T_1)$, as the type system guarantees that $\pair{e_0}{e_1}$ can be purely typed if and only if both $e_0$ and $e_1$ are.
					Further, these proofs are necessarily of smaller depth, meaning that our inductive hypothesis can be applied, and we can conclude that:
					\begin{alignat*}{2}
						&&\;& \msem{P_1}\left(\op{\pi_1(\tau', \tau_0', \tau_1')}{\pi_1(\tau, \tau_0, \tau_1)}\right) \\
						&=&& \msem{P_3} \op{\pi_1(\tau', \tau_0', \tau_1')}{\pi_1(\tau, \tau_0, \tau_1)} \msem{P_3}^\dagger \\
						&=&& \left( \msem{P_{50}} \ket{\pi_5(\tau'), \pi_{50}(\tau_0')} \otimes \msem{P_{51}} \ket{\pi_5(\tau'), \pi_{51}(\tau_1')} \right) \left( \bra{\pi_5(\tau), \pi_{50}(\tau_0)} \msem{P_{50}}^\dagger \otimes \bra{\pi_5(\tau), \pi_{51}(\tau_1)} \msem{P_{51}}^\dagger \right) \\
						&=&& \msem{P_{50}} \op{\pi_5(\tau'), \pi_{50}(\tau_0')}{\pi_5(\tau), \pi_{50}(\tau_0)} \msem{P_{50}}^\dagger \otimes \msem{P_{51}} \op{\pi_5(\tau'), \pi_{51}(\tau_1')}{\pi_5(\tau), \pi_{51}(\tau_1)} \msem{P_{51}}^\dagger \\
						&=&& \msem{P_{40}} \left( \op{\pi_2(\tau'), \pi_{20}(\tau_0')}{\pi_2(\tau), \pi_{20}(\tau_0)}\right)  \otimes \msem{P_{41}} \left( \op{\pi_2(\tau'), \pi_{21}(\tau_1')}{\pi_2(\tau), \pi_{21}(\tau_1)} \right) \\
						&=&& \msem{P_2}\left(\op{\pi_2(\tau'), \pi_{20}(\tau_0'), \pi_{21}(\tau_1')}{\pi_2(\tau), \pi_{20}(\tau_0), \pi_{21}(\tau_1)}\right) \\
					\end{alignat*}
				\item \textsc{T-Mix} and \textsc{T-MixedApp}.
					This case is very similar.
					The judgment being proven is of the form $(\pi(\Delta) \Vdash f\: e : T)$.
					Assume without loss of generality that $P_1$ proves $(\pi_1(\Delta) \Vdash f\: e : T)$ using rule \textsc{T-Mix} in terms of a proof $P_3$ of $(\varnothing \partition \pi_1(\Delta) \vdash f\: e : T)$, and that $P_2$ proves $(\pi_2(\Delta) \Vdash f\: e : T)$ using rule \textsc{T-MixedApp} in terms of a proof $P_{4e}$ of $(\pi_2(\Delta) \Vdash e : T')$ and a proof $P_{4f}$ of $(\vdash f : T' \Rrightarrow T)$.
					It is again not hard to see that both the $P_3$ proof and the $P_4$ proofs can be proven in terms of a proof $P_{5e}$ of $(\varnothing \partition \pi_2(\Delta) \vdash e : T')$ and a proof $P_{5f}$ of $(\vdash f : T' \rightsquigarrow T)$.
					These proofs are of smaller depth, so we can use our inductive hypothesis to conclude:
					\begin{alignat*}{2}
						&&\;& \msem{P_1}\left(\op{\pi_1(\tau')}{\pi_1(\tau)}\right) \\
						&=&& \msem{P_3} \op{\pi_1(\tau')}{\pi_1(\tau)} \msem{P_3}^\dagger \\
						&=&& \msem{P_{5f}} \msem{P_{5e}} \op{\pi_5(\tau')}{\pi_5(\tau)} \msem{P_{5e}}^\dagger \msem{P_{5f}}^\dagger \\
						&=&& \msem{P_{4f}}\left( \msem{P_{4e}} \left( \op{\pi_2(\tau')}{\pi_2(\tau)} \right) \right) \\
						&=&& \msem{P_2}\left(\op{\pi_2(\tau')}{\pi_2(\tau)}\right) \\
					\end{alignat*}
				\item \textsc{T-Channel} and \textsc{T-MixedAbs}.
					The judgment being proven is $(\vdash \lambda e \xmapsto{{\color{gray}T}} e' : T \Rrightarrow T')$
					Assume without loss of generality that $P_1$ is proven using rule \textsc{T-Channel} in terms of a proof $P_3$ of $(\vdash \lambda e \xmapsto{{\color{gray}T}} e' : T \rightsquigarrow T')$, and that $P_2$ is proven using rule \textsc{T-MixedAbs} in terms of a proof $P_4$ of $(\varnothing \partition \Delta \vdash e : T)$ and a proof $P_4'$ of $(\Delta \Vdash e' : T')$.
					(We don't have to include a discarded context $\Delta_0$ because we know this function is also purely typed.)
					Then these subproofs can be proven in terms of $P_4$ and a proof $P_5'$ of $(\varnothing \partition \Delta e' : T')$, and we can use our inductive hypothesis to prove $P_1$ and $P_2$ semantics equivalent:
					\begin{alignat*}{2}
						&&\;& \msem{P_1}\left(\op{v'}{v}\right) \\
						&=&& \msem{P_3} \op{v'}{v} \msem{P_3}^\dagger \\
						&=&& \msem{P_{5}'} \msem{P_{4}}^\dagger \op{v'}{v} \msem{P_{4}} \msem{P_{5}'}^\dagger \\
						&=&& \msem{P_{4}'} \left( \msem{P_{4}}^\dagger \op{v'}{v} \msem{P_{4}} \right) \\
						&=&& \msem{P_2}\left(\op{v'}{v}\right) \\
					\end{alignat*}
				\item Any other cases are impossible because the typing rules target conflicting syntax or conflicting contexts. The cases above are the only ones where the same judgment can be proven using different rules.
			\end{itemize}
	\end{itemize}
\end{proof}

\clearpage

\section{Full Compilation}
\label{app:compilation}

In Section~\ref{sec:compilation}, we outline our procedure for compiling Qunity programs to low-level quantum circuits.
In this appendix, we assume the reader has already read that section and cover the details we previously omitted, justifying that this strategy does indeed output correct circuits.
This appendix is organized in the reverse order of Section~\ref{sec:compilation}: we construct Qunity semantics from the bottom-up, starting from low-level circuits.
Note that the validity checking components (Lemmas~\ref{lem:validate}, \ref{lem:invalidflag}, and \ref{lem:invalidflagmix}) were not present at the time of original publication.
These additions correct a mistake in the original version of this appendix.

\subsection{Low-level circuits and higher-level circuits}
\label{app:ll-hl}

First, we construct a circuit component used for flagging invalid value encodings.
We use the following notation for the projectors onto valid and invalid encodings:
\[
	\projv{T} \defeq \sum_{v \in \V(T)} \op{\enc{v}}{\enc{v}} \in \linear(\complex^{2^{\size(T)}})
	\qquad
	\projvx{T} \defeq \mathbb{I} - \projv{T}
\]
\begin{lemma}
	\label{lem:validate}
	For any type $T$, we can construct a low-level circuit implementing the self-adjoint unitary operator $\validate{T}: \linear(\complex^{2^{\size(T) + n(T)}})$ for some $n(T) \in \mathbb{N}$ with the following properties:
	\begin{align*}
		\left(\mathbb{I} \otimes \op{\zero}{\zero}^{\otimes n(T)}\right) \validate{T} \left(\projv{T} \otimes \op{\zero}{\zero}^{\otimes n(T)}\right) &= \projv{T} \otimes \op{\zero}{\zero}^{\otimes n(T)} \\
		\left(\mathbb{I} \otimes \op{\zero}{\zero}^{\otimes n(T)}\right) \validate{T} \left(\projvx{T} \otimes \op{\zero}{\zero}^{\otimes n(T)}\right) &= 0
	\end{align*}
	In other words, the $\validate{T}$ circuit uses $n(T)$ qubits to flag invalid encodings.
\end{lemma}
\begin{proof}
	Our construction is recursive on $T$.
	In the \Void{} case, there is \emph{no} valid encoding, so $n(\Void)=1$ and the circuit is a single Pauli-X gate.
	In the \Unit{} case, there is no \emph{invalid} encoding, so $n(\Unit)=0$ and the empty circuit enacts an identity operator.

	For the $T_0 \oplus T_1$ case, suppose that $\size(T_0) \leq \size(T_1)$.
	We can make this assumption without loss of generality because of the direct sum swap map isomorphism $T_0 \oplus T_1 \cong T_1 \oplus T_0$, which can be implemented with a single Pauli-X gate on the first qubit.
	Set $n(T_0 \oplus T_1) = 1 + \max\{n(T_0), n(T_1)\}$, and then the following circuit suffices:
\begin{DIFnomarkup}
	\[
	\begin{quantikz}
		\lstick{1}                         & \octrl{2} & \octrl{1}                                & \ctrl{1}                                 & \rstick{1} \qw \\
		\lstick{$\size(T_0)$}              & \qw       & \gate{\validate{T_0}} \vqw{3}            & \gate[2]{\validate{T_1}}                 & \rstick{$\size(T_1)$} \qw \\
		\lstick{$\size(T_1) - \size(T_0)$} & \ctrl{1}  & \qw                                      & \vqw{2} \\
		\lstick{1}                         & \targ{}   & \qw                                      & \qw                                      & \rstick{1} \qw \\
		\lstick{$\max\{n(T_0), n(T_1)\}$}  & \qw       & \gate{\validate{T_0} \otimes \mathbb{I}} & \gate{\validate{T_1} \otimes \mathbb{I}} & \rstick{$\max\{n(T_0), n(T_1)\}$} \qw
	\end{quantikz}
	\]
\end{DIFnomarkup}
	The first control on the $\size(T_1) - \size(T_0)$ wires should be controlled on the condition that \emph{any} of the wires are in the $\ket{\one}$ state, detecting whether the sum type encoding is improperly zero-padded.
	The control wires are designed so that they leave valid encodings unchanged, but flag invalid encodings.

	In the $T_0 \otimes T_1$ case, we use $n(T_0 \otimes T_1) = n(T_0) + n(T_1)$, with the following circuit:
\begin{DIFnomarkup}
	\[
	\begin{quantikz}
		\lstick{$\size(T_0)$} & \gate{\validate{T_0}} \vqw{2} & \qw                           & \qw \rstick{$\size(T_0)$} \\
		\lstick{$\size(T_1)$} & \qw                           & \gate{\validate{T_1}} \vqw{2} & \qw \rstick{$\size(T_1)$} \\
		\lstick{$n(T_0)$}     & \gate{\validate{T_0}}         & \qw                           & \qw \rstick{$n(T_0)$} \\
		\lstick{$n(T_1)$}     & \qw                           & \gate{\validate{T_1}}         & \qw \rstick{$n(T_1)$}
	\end{quantikz}
	\]
\end{DIFnomarkup}
A pair encoding is valid if an only if each of its components is valid, and so this circuit flags the two components independently.
\end{proof}

Definition~\ref{def:implement-kraus} requires that a low-level unitary circuit $U$ implementing a Kraus operator $E$ satisfy $\bra{\encode{v'}, \zero} U \ket{\encode{v}, \zero} = \bra{v'}U\ket{v}$.
Some constructions below will depend on implementations with a slightly stricter requirement, formalized in the following lemma.
\begin{lemma}
	\label{lem:invalidflag}
	Suppose that it is possible to implement a Kraus operator $E : \Hilb(T) \to \Hilb(T')$ with a unitary $U : \complex^{2^{\size(T) + \nprep}} \to \complex^{2^{\size(T') + \nflag}}$ that meets the requirements of Defintion~\ref{def:implement-kraus}.
	Then, it is also possible to implement $E$ with an operator $U' : \complex^{2^{\size(T) + \nprep'}} \to \complex^{2^{\size(T') + \nflag'}}$ for integers $\nprep'$ and $\nflag'$ such that:
	\[
		(\mathbb{I} \otimes \op{\zero}{\zero}^{\otimes \nflag'}) U' (\mathbb{I} \otimes \op{\zero}{\zero}^{\otimes \nprep'})	= (\projv{T'} \otimes \op{\zero}{\zero}^{\otimes \nflag'}) U' (\mathbb{I} \otimes \op{\zero}{\zero}^{\otimes \nprep'})
	\]
\end{lemma}
\begin{proof}
	We construct $U'$ simply by appending the validation circuit to $U$, using $\nprep' = \nprep + n(T)$ and $\nflag' = \nflag + n(T)$, using $n(T)$ from Lemma~\ref{lem:validate}.
\begin{DIFnomarkup}
	\[
	\begin{quantikz}
		\lstick{$\size(T)$} & \gate[2]{U} & \gate{\validate{T}} \vqw{2} & \qw \rstick{$\size(T)$} \\
		\lstick{$\nprep$} &  & \qw & \qw \rstick{$\nflag$} \\
		\lstick{$n(T)$} & \qw & \gate{\validate{T}} & \qw \rstick{$n(T)$}
	\end{quantikz}
	\]
\end{DIFnomarkup}
Using the second condition of Lemma~\ref{lem:validate}, we have:
\begin{align*}
		(\mathbb{I} \otimes \op{\zero}{\zero}^{\otimes \nflag'}) U' (\mathbb{I} \otimes \op{\zero}{\zero}^{\otimes \nprep'}	&= ((\projv{T'} + \projvx{T'}) \otimes \op{\zero}{\zero}^{\otimes \nflag'}) U' (\mathbb{I} \otimes \op{\zero}{\zero}^{\otimes \nprep'}) \\
																																																												&= (\projv{T'} \otimes \op{\zero}{\zero}^{\otimes \nflag'}) U' (\mathbb{I} \otimes \op{\zero}{\zero}^{\otimes \nprep'})
	\end{align*}

	To show that it still meets the requirements of Definition~\ref{def:implement-kraus}:
\begin{align*}
	\bra{\enc(v'), \zero^{\otimes \nflag'}} U' \ket{\enc(v), \zero^{\otimes \nprep'}}
	&= \bra{\enc(v'), \zero^{\otimes \nflag'}} \left(\projv{T'} \otimes \op{\zero}{\zero}^{\otimes \nflag'}\right) U' \ket{\enc(v), \zero^{\otimes \nprep'}} \\
	&= \bra{\enc(v'), \zero^{\otimes \nflag}} \left(\projv{T'} \otimes \op{\zero}{\zero}^{\otimes \nflag}\right) U \ket{\enc(v), \zero^{\otimes \nprep}} \\
	&= \bra{\enc(v'), \zero^{\otimes \nflag}} U \ket{\enc(v), \zero^{\otimes \nprep}}
	\end{align*}
	Thus, it behaves the same way on valid inputs as the original $U$.
\end{proof}

The same idea can also be used when implementing superoperators to ensure that valid inputs lead to valid outputs.
\begin{lemma}
	\label{lem:invalidflagmix}
	Suppose that it is possible to implement a completely positive trace-non-increasing linear superoperator $\mathcal{E} : \linear(\Hilb(T)) \to \linear(\Hilb(T'))$ with a unitary $U : \complex^{2^{\size(T) + \nprep}} \to \complex^{2^{\size(T') + \nflag + \ngarb}}$ that meets the requirements of Defintion~\ref{def:implement-superop}.
	Then, it is also possible to implement $E$ with an operator $U' : \complex^{2^{\size(T) + \nprep'}} \to \complex^{2^{\size(T') + \nflag' + \ngarb}}$ for integers $\nprep', \nflag'$ with the additional requirement that that for all $v \in \V(T)$,
	\[
		\left(\projvx{T'} \otimes \op{\zero}{\zero}^{\otimes \nflag'} \otimes \mathbb{I}\right) U' \ket{\encode(v), \zero^{\otimes \nprep'}} = 0
	\]
\end{lemma}
\begin{proof}
	As in Lemma~\ref{lem:invalidflag}, we append the validation circuit to $U$, using $\nprep' = \nprep + n(T)$ and $\nflag' = \nflag + n(T)$, with $n(T)$ from Lemma~\ref{lem:validate}.
\begin{DIFnomarkup}
	\[
	\begin{quantikz}
		\lstick{$\size(T)$}  & \gate[3,nwires=3]{U} & \qw & \gate{\validate{T}} \vqw{2} & \qw \rstick{$\size(T)$} \\
		\lstick{$\nprep$}  & & \qw & \qw & \qw \rstick{$\nflag$} \\
		& & \gate[swap]{} & \gate{\validate{T}} & \qw \rstick{$n(T)$} \\
		\lstick{$n(T)$} & \qw &  & \qw & \qw \rstick{$\ngarb$}
	\end{quantikz}
	\]
\end{DIFnomarkup}
The first condition of Lemma~\ref{lem:validate} ensures that the requirements of Definition~\ref{def:implement-superop} still hold.
The second condition of Lemma~\ref{lem:validate} ensures the new requirement.
\end{proof}

Next, we construct the direct sum injections, used in several of the higher-level circuits.
Throughout the rest of this section, we write the number of prep, flag, and garbage wires as $n\subcap{p}$, $n\subcap{f}$, and $n\subcap{g}$, respectively, and we use $C(E)$ to denote the low-level unitary qubit operator that implements the high-level Kraus operator $E$.

\begin{lemma}
	\label{lem:compile-inj}
	It is always possible to implement the direct sum injections $\msem{\lef{T_0}{T_1}}$ and $\msem{\rit{T_0}{T_1}}$.
\end{lemma}
\begin{proof}
First, one can implement the operator $\lef{T_0}{T_1} : \Hilb(T_0) \to \Hilb(T_0) \oplus \Hilb(T_1)$ using $n\subcap{p} = 1 + \max\{\size(T_0), \size(T_1)\} - \size(T_0)$ and $n\subcap{f} = 0$ with this circuit:
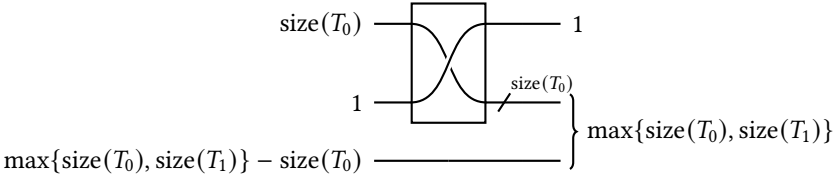
\[
\begin{quantikz}
	\lstick{$\size(T_0)$} & \gate[swap]{} & \qw & \rstick{1} \qw \\
	\lstick{1} & & \qwbundle{\size(T_0)} & \rstick[wires=2]{$\max\{\size(T_0),\size(T_1)\}$} \qw \\
	\lstick{$\max\{\size(T_0),\size(T_1)\} - \size(T_0)$} & \qw & \qw & \qw
\end{quantikz}
\]
Using $U$ to represent the qubit-based unitary implemented by this circuit, we can show that it meets the needed criteria:
	\begin{alignat*}{2}
		&&\;& \bra{\enc(\lef{T_0}{T_1} v'), 0} U \ket{\enc(v), 0, 0} \\
		&=&& \bra{0, \enc(v'), 0} U \ket{\enc(v), 0, 0} \\
		&=&& \ip{0, \enc(v'), 0}{0, \enc(v), 0} \\
		&=&& \ip{\enc(v')}{\enc(v)} \\
		&=&& \ip{v'}{v} \\
		&=&& \left(\bra{v'} \oplus 0\right)\left(\ket{v} \oplus 0\right) \\
		&=&& \ip{\lef{T_0}{T_1} v'}{\lef{T_0}{T_1} v} \\
		&&\;& \bra{\enc(\rit{T_0}{T_1} v'), 0} U \ket{\enc(v), 0, 0} \\
		&=&& \bra{1, \enc(v'), 0} U \ket{\enc(v), 0, 0} \\
		&=&& \ip{1, \enc(v'), 0}{0, \enc(v), 0} \\
		&=&& 0 \\
		&=&& \left(0 \oplus \bra{v'}\right)\left(\ket{v} \oplus 0\right) \\
		&=&& \ip{\rit{T_0}{T_1} v'}{\lef{T_0}{T_1} v} \\
	\end{alignat*}

Similarly, one can implement the operator $\rit{T_0}{T_1} : \Hilb(T_1) \to \Hilb(T_0) \oplus \Hilb(T_1)$ using $n\subcap{p} = 1 + \max\{\size(T_0), \size(T_1)\} - \size(T_1)$ and $n\subcap{f} = 0$ like this:
\[
\begin{quantikz}
	\lstick{$\size(T_1)$} & \gate[swap]{} & \gate{X} & \rstick{1} \qw \\
	\lstick{1} & & \qwbundle{\size(T_1)} & \rstick[wires=2]{$\max\{\size(T_0),\size(T_1)\}$} \qw \\
	\lstick{$\max\{\size(T_0),\size(T_1)\} - \size(T_1)$} & \qw & \qw & \qw
\end{quantikz}
\]
	\begin{alignat*}{2}
		&&\;& \bra{\enc(\lef{T_0}{T_1} v'), 0} U \ket{\enc(v), 0, 0} \\
		&=&& \bra{0, \enc(v'), 0} U \ket{\enc(v), 0, 0} \\
		&=&& \ip{0, \enc(v'), 0}{1, \enc(v), 0} \\
		&=&& \ip{\enc(v')}{\enc(v)} \\
		&=&& 0 \\
		&=&& \left(\bra{v'} \oplus 0\right)\left(0 \oplus \ket{v}\right) \\
		&=&& \ip{\lef{T_0}{T_1} v'}{\rit{T_0}{T_1} v} \\
		&&\;& \bra{\enc(0 \oplus v),0} U \ket{\enc(v), 0, 0} \\
		&=&& \bra{1, \enc(v'), 0} U \ket{\enc(v), 0, 0} \\
		&=&& \ip{1, \enc(v'), 0}{1, \enc(v), 0} \\
		&=&& \ip{\enc(v')}{\enc(v)} \\
		&=&& \ip{v'}{v} \\
		&=&& \left(0 \oplus \bra{v'}\right)\left(0 \oplus \ket{v}\right) \\
		&=&& \ip{\rit{T_0}{T_1} v'}{\rit{T_0}{T_1} v}
	\end{alignat*}
\end{proof}

When we construct circuits of these Kraus operators, we are implicitly using three things: identity operators (represented by bare wires), tensor products (represented by vertically-stacked gates), and operator composition (represented by horizontally-stacked gates).
To justify this form of circuit diagram, we must show that these constructions are always possible.
It is always possible to implement the identity operator $\mathbb{I} : \Hilb(T) \to \Hilb(T)$, by using $n\subcap{p} = n\subcap{f} = 0$ and an identity circuit.
The next two lemmas demonstrate that tensor products and compositions are possible.

\begin{lemma}
	Suppose it is possible to implement the operators $E_0 : \Hilb_0 \to \Hilb_0'$ and $E_1 : \Hilb_1 \to \Hilb_1'$.
	Then it is possible to implement the operator $E_0 \otimes E_1 : \Hilb_0 \otimes \Hilb_1 \to \Hilb_0' \otimes \Hilb_1'$.
\end{lemma}
\begin{proof}
	Assume $E_0$ is implemented by $U_0$ with $n_{\zero,\textsc{p}}$ prep wires and $n_{\zero,\textsc{f}}$ flag wires, and assume that $E_1$ is implemented by $U_1$ with $n_{\one,\textsc{p}}$ prep wires and $n_{\one,\textsc{f}}$ flag wires.
	The following qubit circuit $U$ then implements $E_0 \otimes E_1$ with $(n_{\zero,\textsc{p}} + n_{\one,\textsc{p}})$ prep wires and $(n_{\zero,\textsc{f}} + n_{\one,\textsc{f}})$ flag wires: 

	\[
	\begin{quantikz}
		\lstick{$\size(\Hilb_0)$} & \qw & \gate[2]{U_0} & \qw & \rstick{$\size(\Hilb_0')$} \qw \\
		\lstick{$\size(\Hilb_1)$} & \gate[swap]{} & & \gate[swap]{} & \rstick{$\size(\Hilb_1')$} \qw \\
		\lstick{$n_{\zero,\textsc{p}}$} & & \gate[2]{U_1} & & \rstick{$n_{\zero,\textsc{f}}$} \qw \\
		\lstick{$n_{\one,\textsc{p}}$} & \qw & & \qw & \rstick{$n_{\one,\textsc{f}}$} \qw \\
	\end{quantikz}
	\]

	\begin{alignat*}{2}
		&&\;& \bra{\enc(v_0'), \enc(v_1'), 0, 0} U \ket{\enc(v_0), \enc(v_1), 0, 0} \\
		&=&& \bra{\enc(v_0'), 0, \enc(v_1'), 0} (U_0 \otimes U_1) \ket{\enc(v_0), 0, \enc(v_1), 0} \\
		&=&& \bra{\enc(v_0'), 0} U_0 \ket{\enc(v_0), 0} \cdot \bra{\enc(v_1'), 0} U_1 \ket{\enc(v_1), 0} \\
		&=&& \bra{v_0'} E_0 \ket{v_0} \cdot \bra{v_1'} E_1 \ket{v_1} \\
		&=&& \bra{v_0', v_1'} (E_0 \otimes E_1) \ket{v_0, v_1}
	\end{alignat*}
\end{proof}

\begin{lemma}
	\label{lem:compose}
	Suppose it is possible to implement the operators $E_0 : \Hilb_0 \to \Hilb'$ and $E_1 : \Hilb' \to \Hilb_1$.
	Then it is possible to implement the operator $E_1 E_0 : \Hilb_0 \to \Hilb_1$.
\end{lemma}
\begin{proof}
	Assume $E_0$ is implemented by $C(E_0) = U_0$ with $n_{\zero,\textsc{p}}$ prep wires and $n_{\zero,\textsc{f}}$ flag wires, and assume that $E_1$ is implemented by $C(E_1) = U_1$ with $n_{\one,\textsc{p}}$ prep wires and $n_{\one,\textsc{f}}$ flag wires.
	Assume that $U_0$ adheres to the requirements of Lemma~\ref{lem:invalidflag}, and let $B'$ be a basis for $\Hilb'$.
	The following qubit circuit $U$ then implements $E_1 E_0$ with $(n_{\zero,\textsc{p}} + n_{\one,\textsc{p}})$ prep wires and $(n_{\zero,\textsc{f}} + n_{\one,\textsc{f}})$ flag wires: 

	\[
	\begin{quantikz}
		\lstick{$\size(\Hilb_0)$} & \gate[2]{U_0} & \qwbundle{\size(\Hilb')} & \gate[2]{U_1} & \rstick{$\size(\Hilb_1)$} \qw \\
		\lstick{$n_{\zero,\textsc{p}}$} & & \gate[swap]{} & & \rstick{$n_{\one,\textsc{f}}$} \qw \\
		\lstick{$n_{\one,\textsc{p}}$} & \qw & & \qw & \rstick{$n_{\zero,\textsc{f}}$} \qw
	\end{quantikz}
	\]

	\begin{alignat*}{2}
		&&\;& \bra{\enc(v_1), 0, 0} U \ket{\enc(v_0), 0, 0} \\
		&=&& \bra{\enc(v_1), 0, 0} (U_1 \otimes \mathbb{I}) (\mathbb{I} \otimes \textsc{swap}) (U_0 \otimes \mathbb{I}) \ket{\enc(v_0), 0, 0} \\
		&=&& \left(\bra{\enc(v_1), 0} U_1 \otimes \bra{0}\right) (\mathbb{I} \otimes \textsc{swap}) \left(U_0\ket{\enc(v_0),0} \otimes \ket{0}\right) \\
		&=&& \left(\bra{\enc(v_1), 0} U_1\right) (\mathbb{I} \otimes \op{0}{0}) \left(U_0\ket{\enc(v_0),0}\right) \\
		&=&& \sum_{b \in \{\zero, \one\}^{\size(\Hilb')}} \bra{\enc(v_1), 0} U_1 \op{b,0}{b,0} U_0\ket{\enc(v_0), 0} \\
		&=&& \sum_{\ket{v'} \in B'} \bra{\enc(v_1), 0} C(E_1) \op{\enc(v'), 0}{\enc(v'), 0} C(E_0) \ket{\enc(v_0), 0} \\
		&=&& \sum_{\ket{v'} \in B'} \bra{v_1} E_1 \op{v'}{v'} E_0 \ket{v_0} \\
		&=&& \bra{v_1} E_1 \left( \sum_{v' \in B'} \op{v'}{v'} \right) E_0 \ket{v_0} \\
		&=&& \bra{v_1} E_1 \mathbb{I} E_0 \ket{v_0} \\
		&=&& \bra{v_1} E_1 E_0 \ket{v_0} \\
	\end{alignat*}

\end{proof}

The adjoint is another useful construction that we will need in Appendix~\ref{app:hl-qunity}, motivating the following lemma:
\begin{lemma}
	\label{lem:compile-adj}
	Suppose it is possible to implement the operator $E : \Hilb \to \Hilb'$.
	Then it is possible to implement the operator $E^\dagger : \Hilb' \to \Hilb$.
\end{lemma}
\begin{proof}
	Assume $E$ is implemented by $U$ with $n\subcap{p}$ prep wires and $n\subcap{f}$ flag wires.
	Given the circuit for $U$, one can construct a circuit for $U^\dagger$ in the standard way: by taking the adjoint of each gate in the circuit and reversing the order.
	This circuit $U^\dagger$ then implements $E^\dagger$ with $n\subcap{f}$ prep wires and $n\subcap{p}$ flag wires.
	\[
	\begin{quantikz}
		\lstick{$\size(\Hilb')$} & \gate[2]{U^\dagger} & \rstick{$\size(\Hilb)$} \qw \\
		\lstick{$n\subcap{f}$} & & \rstick{$n\subcap{p}$} \qw \\
	\end{quantikz}
	\]
	\begin{alignat*}{2}
		&&\;& \bra{\enc(v'), 0} U^\dagger \ket{\enc(v), 0} \\
		&=&& \left(\bra{\enc(v), 0} U \ket{\enc(v'), 0}\right)^{*} \\
		&=&& \left(\bra{v} E \ket{v'}\right)^{*} \\
		&=&& \bra{v'} E^\dagger \ket{v}
	\end{alignat*}
\end{proof}

\begin{lemma}
	Suppose it is possible to implement the operators $E_0 : \Hilb_0 \to \Hilb_0'$ and $E_1 : \Hilb_1 \to \Hilb_1'$.
	Then it is possible to implement the operator $E_0 \oplus E_1 : \Hilb_0 \oplus \Hilb_1 \to \Hilb_0' \oplus \Hilb_1'$.
\end{lemma}
\begin{proof}
	Assume $E_0$ is implemented with $\preps_0$ prep wires and $\flags_0$ flag wires, and assume that $E_1$ is implemented with $\preps_1$ prep wires and $\flags_1$ flag wires.
	Let $n = \max\{\size(\Hilb_0) + \preps_0, \size(\Hilb_1) + \preps_1\}$.
	The following qubit circuit then implements $E_0 \oplus E_1$ with $1 + n$ total wires; that is, $n - \max\{\size(\Hilb_0), \size(\Hilb_1)\}$ prep wires and $n - \max\{\size(\Hilb_0'), \size(\Hilb_1')\}$ flag wires: 

		\begin{DIFnomarkup}
	\[
	\begin{quantikz}
		\lstick{1} & \octrl{1} & \ctrl{1} & \rstick{1} \qw \\
		\lstick{$n$} & \gate{E_0 \otimes \mathbb{I}} & \gate{E_1 \otimes \mathbb{I}} & \rstick{$n$} \qw
	\end{quantikz}
\]
		\end{DIFnomarkup}

	\begin{alignat*}{2}
		\bra{0, \enc(v_0'), 0} C(E_0 \oplus E_1) \ket{0, \enc(v_0), 0} &=&\;& \bra{\enc(v_0'), 0} C(E_0) \ket{\enc(v_0), 0} \\
		&=&& \bra{v_0'} E_0 \ket{v_0} \\
		&=&& (\bra{v_0'} \oplus 0) (E_0 \oplus E_1) (\ket{v_0} \oplus 0) \\
		\bra{1, \enc(v_1'), 0} C(E_0 \oplus E_1) \ket{0, \enc(v_0), 0} &=&& 0 \\
		&=&& (0 \oplus \bra{v_1'}) (E_0 \oplus E_1) (\ket{v_0} \oplus 0) \\
		\bra{1, \enc(v_1'), 0} C(E_0 \oplus E_1) \ket{1, \enc(v_1), 0} &=&& \bra{\enc(v_1'), 0} C(E_1) \ket{\enc(v_1), 0} \\
		&=&& \bra{v_1'} E_1 \ket{v_1} \\
		&=&& (0 \oplus \bra{v_1'}) (E_0 \oplus E_1) (0 \oplus \ket{v_1}) \\
		\bra{0, \enc(v_0'), 0} C(E_0 \oplus E_1) \ket{1, \enc(v_1), 0} &=&& 0 \\
		&=&& (\bra{v_0'} \oplus 0) (E_0 \oplus E_1) (0 \oplus \ket{v_1})
	\end{alignat*}
\end{proof}

It is standard to use the tensor product monoidally, implicitly using vector space isomorphisms $(\Hilb_1 \otimes \Hilb_2) \otimes \Hilb_3 \approx \Hilb_1 \otimes (\Hilb_2 \otimes \Hilb_3)$ and $\complex \otimes \Hilb \approx \Hilb \approx \Hilb \otimes \complex$.
Under the most general definition of the tensor product, these are isomorphisms rather than strict equality \cite[Chapter 14]{advanced-linear-algebra}, but these isomorphisms can typically be used implicitly.
Our encoding ensures that these implicit isomorphisms do not require any low-level gates because $\enc(\pair{\pair{v_1}{v_2}}{v_3}) = \enc(\pair{v_1}{\pair{v_2}{v_3}})$ and $\enc(\pair{\unit}{v}) = \enc(v) = \enc(\pair{v}{\unit})$.

Similarly, it will be convenient for our analysis in Appendix~\ref{app:hl-qunity} to treat the direct sum monoidally, implicitly using isomorphisms $(\Hilb_1 \oplus \Hilb_2) \oplus \Hilb_3 \approx \Hilb_1 \oplus (\Hilb_2 \oplus \Hilb_3)$ and $\{0\} \oplus \Hilb \approx \Hilb \approx \Hilb \oplus \{0\}$.
However, these isomorphisms \emph{do} correspond to different encodings, so we must show how to implement them.
Note that the additive unit isomorphisms were already implemented in Lemma~\ref{lem:compile-inj}.
They are $\lef{T}{\Void}$ and $\rit{\Void}{T}$, and their inverses are their adjoints as compiled in Lemma~\ref{lem:compile-adj}.
The following lemma demonstrates how to compile the associativity isomorphism:

\begin{lemma}
	Let $T_1$, $T_2$, and $T_3$ be arbitrary types.
	Then it is possible to implement the direct sum's associativity isomorphism $\textsc{assoc} : (\Hilb(T_1) \oplus \Hilb(T_2)) \oplus \Hilb(T_3) \to \Hilb(T_1) \oplus (\Hilb(T_2) \oplus \Hilb(T_3))$.
\end{lemma}
\begin{proof}
	Define the following integers:
	\begin{align*}
		n\subcap{max} &= \max\{\size(T_1), \size(T_2), \size(T_3)\} \\
		n\subcap{p} &= n\subcap{max} - \max\{\max\{\size(T_1), \size(T_2)\}, \size(T_3) - 1\} \\
		n\subcap{f} &= n\subcap{max} - \max\{\size(T_1) - 1, \max\{\size(T_2), \size(T_3)\}
	\end{align*}
	We can implement ``shifting'' operations \textsc{rsh} and \textsc{lsh} via a series of swap gates that enacts a rotation permutation, so that $\textsc{rsh} \ket{\psi_1, \psi_2, \ldots, \psi_{n-1}, \psi_n} = \ket{\psi_n, \psi_1, \psi_2, \ldots, \psi_{n-1}}$ and $\textsc{lsh} = \textsc{rsh}^{-1}$.
	The following qubit circuit $U$ then implements the associator with the defined $n\subcap{p}$ and $n\subcap{f}$.
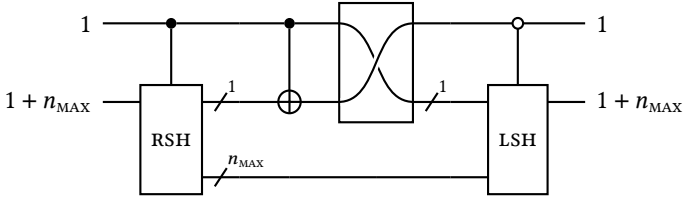
	\[
	\begin{quantikz}
		\lstick{1} & \ctrl{1} & \qw & \ctrl{1} & \gate[swap]{} & \qw & \octrl{1} & \rstick{1} \qw \\
		\lstick{$1 + n\subcap{max}$} & \gate[2,nwires=2]{\textsc{rsh}} & \qwbundle{1} & \targ{} & & \qwbundle{1} & \gate[2,nwires=2]{\textsc{lsh}} & \rstick{$1 + n\subcap{max}$} \qw \\
																 & & \qwbundle{n\subcap{max}} & \qw & \qw & \qw & \qw & &
	\end{quantikz}
	\]
	We can show that this circuit has the desired behavior:
	\begin{alignat*}{2}
		&&\;& \ket{\enc\left(\lef{(T_1 \oplus T_2)}{T_3} (\lef{T_1}{T_2} v) \right), 0} \\
		&=&& \ket{0, 0, \enc(v), 0} \\
		&\mapsto^{*}&& \ket{0, \enc(v), 0, 0} \\
		&=&& \ket{\enc(\lef{T_1}{(T_2 \oplus T_3)} v), 0} \\
		&&\;& \ket{\enc\left(\lef{(T_1 \oplus T_2)}{T_3} (\rit{T_1}{T_2} v) \right), 0} \\
		&=&& \ket{0, 1, \enc(v), 0} \\
		&\mapsto^{*}&& \ket{1, 0, \enc(v), 0} \\
		&=&& \ket{\enc\left(\rit{T_1}{(T_2 \oplus T_3)} (\lef{T_2}{T_3} v)\right), 0} \\
		&&\;& \ket{\enc\left(\rit{(T_1 \oplus T_2)}{T_3} v\right), 0} \\
		&=&& \ket{1, \enc(v), 0} \\
		&\mapsto&& \ket{1, 0, \enc(v), 0} \\
		&\mapsto&& \ket{1, 1, \enc(v), 0} \\
		&=&& \ket{\enc\left(\rit{T_1}{(T_2 \oplus T_3)} (\rit{T_2}{T_3} v)\right), 0}
	\end{alignat*}
\end{proof}

We can now implement the monoidal isomorphisms for the tensor product and direct sum.
Both of these are \emph{symmetric} monoidal categories, and the swap maps are straightforward to implement, using swap gates to implement $T_0 \otimes T_1 \cong T_1 \otimes T_0$, and using a single Pauli-X gate on the indicator qubit to implement $T_0 \oplus T_1 \cong T_1 \oplus T_0$.
We will also need \emph{distributivity} of the tensor product over the direct sum, part of the definition of a \emph{bimonoidal} category \cite{symmetric-bimonoidal} (sometimes known as a ``rig category'' \cite{rig-programming}).
\begin{lemma}
	Let $T$, $T_0$, and $T_1$ be arbitrary types.
	It is possible to implement the distributivity isomorphism $\textsc{distr} : \Hilb(T) \otimes (\Hilb(T_0) \oplus \Hilb(T_1)) \cong (\Hilb(T) \otimes \Hilb(T_0)) \oplus (\Hilb(T) \otimes \Hilb(T_1))$, which acts like $\ket{v} \otimes \left( \ket{v_0} \oplus \ket{v_1} \right) \mapsto \ket{v, v_0} \oplus \ket{v, v_1}$.
\end{lemma}
\begin{proof}
	This can be done with no prep or flag wires:
	
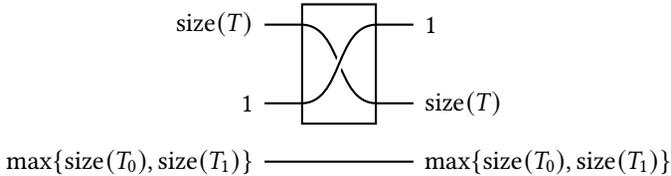
\[
	\begin{quantikz}
		\lstick{$\size(T)$} & \gate[swap]{} & \rstick{1} \qw \\
		\lstick{1} & & \rstick{$\size(T)$} \qw \\
		\lstick{$\max\{\size(T_0), \size(T_1)\}$} & \qw & \rstick{$\max\{\size(T_0), \size(T_1)\}$} \qw
	\end{quantikz}
	\]
	It should be clear from the value encoding that this circuit has the correct behavior.
\end{proof}
In the high-level circuits, we will often use some transformation of the \textsc{distr} construction, for example:
\begin{itemize}
	\item its adjoint $(\Hilb(T) \otimes \Hilb(T_0)) \oplus (\Hilb(T) \otimes \Hilb(T_1)) \cong \Hilb(T) \otimes (\Hilb(T_0) \oplus \Hilb(T_1))$,
	\item its composition with swaps $(\Hilb(T_0) \oplus \Hilb(T_1)) \otimes \Hilb(T) \cong (\Hilb(T_0) \otimes \Hilb(T)) \oplus (\Hilb(T_1) \otimes \Hilb(T))$,
	\item compositions of distributions $(\Hilb(T_{00}) \oplus \Hilb(T_{01})) \otimes (\Hilb(T_{10}) \oplus \Hilb(T_{11})) \cong (\Hilb(T_{00}) \otimes \Hilb(T_{10})) \oplus (\Hilb(T_{00}) \otimes \Hilb(T_{11})) \oplus (\Hilb(T_{01}) \otimes \Hilb(T_{10})) \oplus (\Hilb(T_{01}) \otimes \Hilb(T_{11}))$.
\end{itemize}
We will denote all of these as ``\textsc{distr},'' but the transformation being applied should always be clear from context.

\begin{lemma}[pure error handling]
	\label{lem:pure-error}
	Suppose it is possible to implement a Kraus operator $E : \Hilb \to \Hilb'$.
	Then it is possible to implement a norm-preserving operator $E\subcap{f} : \Hilb \to \Hilb' \oplus \Hilb\subcap{f}$ for some ``flag space'' $\Hilb\subcap{f}$ such that $\msem{\lef{T}{T'}}^\dagger E\subcap{f} = E$.
\end{lemma}
\begin{proof}
	Assume that $E$ is implemented by the unitary $U$ with $n\subcap{p}$ prep wires and $n\subcap{f}$ flag wires.
	Then using Lemma~\ref{lem:invalidflag}, we can assume that
	\[
		U\ket{\enc{v}, \zero} = \left(\sum_{v' \in \V(T')} \bra{v'}E\ket{v} \ket{\enc{v'}, \zero}\right) + \ket{\psi_{v}}
	\]
	for some $\ket{\psi_v} \in \Hilb\subcap{f}$ with $(\mathbb{I}\otimes \op{\zero}{\zero})\ket{\psi_v} = 0$.
	We use $E\subcap{f}$ defined such that 
	\[
		E\subcap{f}\ket{v} = \left(\sum_{v' \in \V(T')} \bra{v'}E\ket{v} \ket{v'}\right) \oplus \ket{\psi_{v}} = E\ket{v} \oplus \ket{\psi_{v}}.
	\]

	This definition ensures that $\msem{\lef{T}{T'}}^\dagger E\subcap{f}\ket{v} = \msem{\lef{T}{T'}}^\dagger \left(E\ket{v} \oplus \ket{\psi_{v}}\right) = E\ket{v}$.
	See that $\norm{E\ket{v}} = \norm{U\ket{\enc{v}, \zero}} = 1$, so $E\subcap{f}$ is norm-preserving.

	The following circuit implements $E\subcap{f}$ with $n\subcap{p}+1$ prep wires and no flag wires.
	\[
		\begin{DIFnomarkup}
	\begin{quantikz}
		\lstick{$\size(\Hilb)$} & \gate[swap]{} & \qw & \hphantomgate{} & \targ{} & \rstick{1} \qw \\
		\lstick{1} & & \gate[2]{U} & \qwbundle{\size(\Hilb')} & \qw &  \rstick{$\size(\Hilb')$} \qw \\
		\lstick{$n\subcap{p}$} & \qw & & \qwbundle{n\subcap{f}} & \ctrl{-2} & \rstick{$n\subcap{f}$} \qw
	\end{quantikz}
		\end{DIFnomarkup}
	\]
	\begin{align*}
		\ket{\enc(v), 0, 0} &\mapsto \ket{0, \enc(v), 0} \\
												&\mapsto \ket{0} \otimes U\ket{\enc(v), 0} \\
												&= \left(\sum_{v' \in \V(T')} \bra{v'}E\ket{v} \ket{\zero, \enc{v'}, \zero}\right) + \ket{\zero, \psi_{v}} \\
												&\mapsto \left(\sum_{v' \in \V(T')} \bra{v'}E\ket{v} \ket{\zero, \enc{v'}, \zero}\right) + \ket{\one, \psi_{v}} \\
		\bra{0, v', 0} C(E\subcap{f}) \ket{\enc(v), 0, 0} &= \bra{0, v', 0} \left(\sum_{v' \in \V(T')} \bra{v'}E\ket{v} \ket{\zero, \enc{v'}, \zero}\right) \\
																											&= (\bra{v'} \oplus 0) E\subcap{f}\ket{v} \\
		\bra{1, v'} C(E\subcap{f}) \ket{\enc(v), 0, 0} &= \ip{\one, v'}{\one, \psi_{v}} \\
																									 &= (0 \oplus \bra{v'}) E\subcap{f}\ket{v}
	\end{align*}
\end{proof}

\begin{lemma}[mixed error handling]
	\label{lem:mixed-error}
Suppose it is possible to implement the completely positive trace-non-increasing linear superoperator $\mathcal E \in \linear(\linear(\Hilb), \linear(\Hilb'))$.
Then, it is possible to implement a completely positive trace-preserving linear superoperator $\cptp(\mathcal{E}) \in \linear(\linear(\Hilb), \linear(\Hilb' \oplus \complex))$ such that for all $\rho \in \linear(\Hilb)$,
\[
	\cptp(\mathcal{E})(\rho) = \mathcal{E}(\rho) \oplus (\tr(\rho) - \tr(\mathcal{E}(\rho))).
\]
\end{lemma}
\begin{proof}
	Assuming $\mathcal{E}$ is implemented by the unitary $U$ with $n\subcap{p}$ prep wires, $n\subcap{f}$ flag wires, and $n\subcap{g}$ garbage wires, the following circuit achieves the desired result with $n\subcap{p}+\size(\Hilb')+1$ prep wires, no flag wires, and $n\subcap{f} + \size(\Hilb') + n\subcap{g}$ garbage wires: 
	\[
	\begin{quantikz}
		\lstick{$\size(\Hilb)$} & \gate[swap]{} & \qw & \hphantomgate{} & \targ{} & \qw & \rstick{1} \qw \\
		\lstick{1} & & \gate[3,nwires=3]{U} & \qwbundle{\size(\Hilb')} & \qw &  \swap{1} & \rstick{$\size(\Hilb')$} \qw \\
		\lstick{$n\subcap{p}$} & \qw & & \qwbundle{n\subcap{f}} & \ctrl{-2} & \control{} & \rstick{$n\subcap{f}$} \qw \\
													 &&& \qw & \qw & \qw & \rstick{$n\subcap{g}$} \qw \\
		\lstick{$\size(\Hilb')$} & \qw & \qw & \qw & \qw & \swap{-2} & \rstick{$\size(\Hilb')$} \qw 
	\end{quantikz}
	\]
	This circuit works by turning flag wires from $\mathcal{E}$ into garbage wires for $\cptp(\mathcal{E})$.
	The first qubit is used as an indicator of failure, and in the event of failure, the output of $U$ is treated as garbage and replaced with a fresh set of $\ket\zero$ qubits from the prep register, as required by the bitstring encoding of sum types.
	Here, the ``control'' on the flag wires should be understood to apply the gate conditioned on \emph{any} of the qubits on the control register being in the $\ket\one$ state; this may be easiest to implement in a language like Open\textsc{Qasm} using a control construct conditioned on \emph{all} of the qubits on the control register being in the $\ket\zero$ state, properly sandwiched by Pauli-X gates.

	To prove that this circuit superoperator $\Ccptp$ correctly implements $\cptp(\E)$, we must show that
	\[
		\bra{v_1'}\Ccptp(\op{v_1}{v_2})\ket{v_2'} = \sum_b \bra{\enc(v_1'), b} \mathcal{E}(\op{\enc(v_1), \zero}{\enc(v_2), \zero}) \ket{\enc(v_2'), b}
	\]
	for all $v_1, v_2 \in \V(T), v_1', v_2' \in \V(T' \oplus \Unit)$.

	We will consider three cases for $v_1'$ and $v_2'$.
	\begin{itemize}
		\item First, suppose both are in the ``success'' (non-error) subspace, and see that $(\bra{v_1'} \oplus 0)\cptp(\E)(\op{v_1}{v_2})(\ket{v_2'} \oplus 0) = \bra{v_1'}\E(\op{v_1}{v_2})\ket{v_2'}$.
			To see that this equals $\sum_b \bra{\zero, \enc(v_1'), b} \Ccptp(\op{\enc(v_1), \zero}{\enc(v_2), \zero}) \ket{\zero, \enc(v_2'), b}$, see from the circuit that a $\zero$ output on the first wire also implies a $\zero$ output on the $n\subcap{f}$ segment of the output wires, as well as the final $\size(\Hilb')$ section, so we really care about
			\[
				\sum_{b \in \{\zero,\one\}^n\subcap{g}} \bra{\zero, \enc(v_1'), \zero, b, \zero} \Ccptp(\op{\enc(v_1), \zero}{\enc(v_2), \zero}) \ket{\zero, \enc(v_2'), \zero, b, \zero}.
			\]
			The requirements placed on $U$ ensure that this equality holds.
		\item Consider the case where exactly one of the two values is in the error subspace, for example $(\bra{v_1'} \oplus 0)\cptp(\E)(\op{v_1}{v_2})(0 \oplus \ket{v_2'}) = 0$.
			Our circuit works correctly in this case because $\bra{\zero, \enc(v_1'), b} \Ccptp(\op{\enc(v_1), \zero}{\enc(v_2), \zero}) \ket{\one, \enc(v_2'), b}$ is always zero, regardless of $b$.
			To see this, see that the $n\subcap{f}$ garbage bits are all zero if and only if the first indicator output bit is zero, so any setting of $b$ would cause one of the two sides of the expression to vanish. In other words, there is no superposition between the error and non-error subspaces because the discarded garbage collapses the state into one of the two.
		\item In the final case, both $v_1'$ and $v_2'$ are in the error subspace, so we must show that
			$\sum_b \bra{\one, \zero, b} \Ccptp(\op{\enc(v_1), \zero}{\enc(v_2), \zero}) \ket{\one, \zero, b} = \tr(\op{v_1}{v_2}) - \tr(\mathcal{E}(\op{v_1}{v_2}))$, as we are encoding an error value as $\texttt{"1"} \doubleplus \texttt{"0"}^{\size{\Hilb'}}$.
			This case is constrained by the first two, as it is the only possible value that would ensure that $\cptp(\E)$ is trace-preserving.
			To verify that $\cptp(\E)$ is trace-preserving, see that there are no flag wires and $U$ cannot output invalid encodings without setting flag qubits (due to Lemma~\ref{lem:invalidflagmix}), so there is no way for the trace to decrease.
	\end{itemize}
\end{proof}

\begin{lemma}[purification]
	\label{lem:purify}
	Suppose it is possible to implement the completely positive trace-non-increasing linear superoperator $\mathcal E \in \linear(\linear(\Hilb), \linear(\Hilb'))$.
	Then, it is possible to implement a Kraus operator $E \in \linear(\Hilb, \Hilb' \otimes \Hilb_{\textsc{g}})$ for some ``garbage'' Hilbert space $\Hilb_{\textsc{g}}$ with the following property: for any $\rho \in \linear(\Hilb), \ket{\psi} \in \Hilb'$, there is some $\ket{g_{\rho,\psi}} \in \Hilb_{\textsc{g}}$ such that
	\[
		\bra{\psi} \mathcal{E}(\rho) \ket{\psi} = \bra{g_{\rho,\psi}, \psi} E\rho E^\dagger \ket{g_{\rho,\psi}, \psi}
	\]
\end{lemma}
\begin{proof}
	Assuming $\mathcal{E}$ is implemented by the unitary $U$ with $n\subcap{p}$ prep wires, $n\subcap{f}$ flag wires, and $n\subcap{g}$ garbage wires, the following circuit achieves the desired result with $n\subcap{p}$ prep wires and $n\subcap{f}$ flag wires by setting $\Hilb\subcap{g} = \complex^{2^{n\subcap{g}}}$:
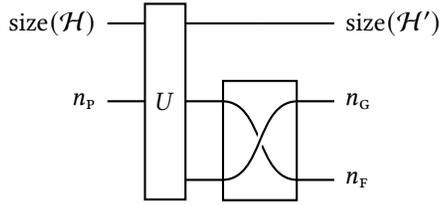
	\[
	\begin{quantikz}
		\lstick{$\size(\Hilb)$} & \gate[3,nwires=3]{U} & \qw & \rstick{$\size(\Hilb')$} \qw \\
		\lstick{$n\subcap{p}$} & & \gate[swap]{} & \rstick{$n\subcap{g}$} \qw \\
													 &&& \rstick{$n\subcap{f}$} \qw
	\end{quantikz}
	\]
	This circuit simply feeds the existing garbage wires into an additional output.
\end{proof}

\begin{definition}[indexed injection]
	For notational convenience, for any Hilbert spaces $\{\Hilb_1, \ldots, \Hilb_n\}$ we define $\inj_j : \Hilb_j \to \bigoplus_{k=1}^n \Hilb_k$ to be the indexed injection into the $j$\textsuperscript{th} subspace.
	Precisely,
	\[
		\inj_j \ket{v} = 0^{\oplus(j-1)} \oplus \ket{v} \oplus 0^{\oplus(n-j)}
		\]
\end{definition}

Next, we consider some results about how to implement trace-non-increasing superoperators with low-level qubit-based unitaries.

\begin{lemma}
	For any type $T$, it is possible to implement a superoperator $\mathcal{E} : \linear(\Hilb(T)) \to \linear(\Hilb(\Unit))$ that computes the trace of its input, effectively discarding it.
\end{lemma}
\begin{proof}
	This gate is implemented with an empty (identity) circuit by setting $n\subcap{p} = n\subcap{f} = 0$, $n\subcap{g} = \size(T)$.
	\[
	\begin{quantikz}
		\lstick{$\size(T)$} & \rstick{$n\subcap{g}$} \qw
	\end{quantikz}
	\]
	\[
	\sum_{g \in \{\zero, \one\}^{\size(T)}} \bra{\encode(\unit), g} \mathbb{I} \rho \mathbb{I}^\dagger \ket{\encode(\unit), g} = \tr(\rho)
	\]
\end{proof}

	In Appendix~\ref{app:hl-qunity}, we will also be constructing circuits from these trace-non-increasing superoperators.
	Again, we must justify this by demonstrating that tensor products and function composition are possible.

\begin{lemma}
	Suppose it is possible to implement the superoperators $\mathcal{E}_0 : \linear(\Hilb_0) \to \linear(\Hilb_0')$ and $\mathcal{E}_1 : \linear(\Hilb_1) \to \linear(\Hilb_1')$.
	Then it is possible to implement the operator $\mathcal{E}_0 \otimes \mathcal{E}_1 : \linear(\Hilb_0 \otimes \Hilb_1) \to \linear(\Hilb_0' \otimes \Hilb_1')$.
\end{lemma}
\begin{proof}
	Assume $E_0$ is implemented by $U_0$ with $n_{\zero,\textsc{p}}$ prep wires, $n_{\zero,\textsc{f}}$ flag wires, and $n\subcap{g}$ garbage wires.
	Assume $E_1$ is implemented by $U_1$ with $n_{\one,\textsc{p}}$ prep wires, $n_{\one,\textsc{f}}$ flag wires, and $n\subcap{g}$ garbage wires.
	The following qubit circuit $U$ then implements $E_0 \otimes E_1$ with $n\subcap{p} = (n_{\zero,\textsc{p}} + n_{\one,\textsc{p}})$ prep wires, $n\subcap{f} = (n_{\zero,\textsc{f}} + n_{\one,\textsc{f}})$ flag wires, and $n\subcap{g} = (n\subcap{\zero,g} + n\subcap{\one,g})$ garbage wires:

	\[
	\begin{quantikz}
		\lstick{$\size(\Hilb_0)$} & \qw & \gate[3,nwires=2]{U_0} & \qw & \qw & \rstick{$\size(\Hilb_0')$} \qw \\
															&&& \qw & \gate[swap]{} & \rstick{$\size(\Hilb_1')$} \qw \\
		\lstick{$\size(\Hilb_1)$} & \gate[swap]{} & & \gate[swap]{} & & \rstick{$n\subcap{\zero,f}$} \qw \\
		\lstick{$n_{\zero,\textsc{p}}$} & & \gate[3,nwires=2]{U_1} & & \gate[swap]{} & \rstick{$n\subcap{\one,f}$} \qw \\
																		&&& \qw & & \rstick{$n\subcap{\zero,g}$} \qw \\
		\lstick{$n_{\one,\textsc{p}}$} & \qw & & \qw & \qw & \rstick{$n\subcap{\one,g}$} \qw
	\end{quantikz}
	\]

	\begin{alignat*}{2}
		&&\;& \sum_{g_0 \in \{\zero, \one\}^{n\subcap{\zero,g}}} \sum_{g_1 \in \{\zero, \one\}^{n\subcap{\one,g}}} \bra{\encode(v_{0,1}'), \encode(v_{1,1}'), 0, g_0, g_1} U \left(\rho_0 \otimes \rho_1 \otimes \op{\zero}{\zero}\right) \\ &&&\qquad \cdot U^\dagger \ket{\encode(v_{0,2}'), \encode(v_{1,2}'), \zero, g_0, g_1} \\
		&=&& \sum_{g_0 \in \{\zero, \one\}^{n\subcap{\zero,g}}} \sum_{g_1 \in \{\zero, \one\}^{n\subcap{\one,g}}} \bra{\encode(v_{0,1}'), \zero, g_0, \encode(v_{1,1}'), \zero, g_1} (U_0 \otimes U_1) \left(\rho_0 \otimes \op{\zero}{\zero} \otimes \rho_1 \otimes \op{\zero}{\zero}\right) \\ &&&\qquad \cdot (U_0 \otimes U_1)^\dagger \ket{\encode(v_{0,2}'), \zero, g_0, \encode(v_{1,2}'), \zero, g_1} \\
		&=&& \sum_{g_0 \in \{\zero, \one\}^{n\subcap{\zero,g}}} \bra{\encode(v_{0,1}'), \zero, g_0} U_0 \left(\rho_0 \otimes \op{\zero}{\zero}\right) U_0^\dagger \ket{\encode(v_{0,2}'), \zero, g_0} \\ &&&\qquad \cdot \sum_{g_1 \in \{\zero, \one\}^{n\subcap{\one,g}}} \bra{\encode(v_{1,1}'), \zero, g_1} U_1 (\rho_1 \otimes \op{\zero}{\zero}) U_1^\dagger \ket{\encode(v_{1,2}'), \zero, g_1} \\
		&=&& \bra{v_{0,1}'} \mathcal{E}_0(\rho_0) \ket{v_{0,2}'} \cdot \bra{v_{1,1}'} \mathcal{E}_1(\rho_1) \ket{v_{1,2}'} \\
		&=&& \bra{v_{0,1}', v_{1,1}'} (\mathcal{E}_0 \otimes \mathcal{E}_1)(\rho_0 \otimes \rho_1) \ket{v_{0,2}', v_{1,2}'} \\
	\end{alignat*}
\end{proof}

\begin{lemma}
	Suppose it is possible to implement the superoperators $\mathcal{E}_0 : \linear(\Hilb_0) \to \linear(\Hilb')$ and $\mathcal{E}_1 : \linear(\Hilb') \to \linear(\Hilb_1)$.
	Then it is possible to implement the operator $\mathcal{E}_1 \circ \mathcal{E}_0 : \linear(\Hilb_0) \to \linear(\Hilb_1)$.
\end{lemma}
\begin{proof}
	Assume $E_0$ is implemented by $U_0$ with $n_{\zero,\textsc{p}}$ prep wires, $n_{\zero,\textsc{f}}$ flag wires, and $n\subcap{g}$ garbage wires.
	Assume $E_1$ is implemented by $U_1$ with $n_{\one,\textsc{p}}$ prep wires, $n_{\one,\textsc{f}}$ flag wires, and $n\subcap{g}$ garbage wires.
	The following qubit circuit $U$ then implements $E_0 \otimes E_1$ with $n\subcap{p} = (n_{\zero,\textsc{p}} + n_{\one,\textsc{p}})$ prep wires, $n\subcap{f} = (n_{\zero,\textsc{f}} + n_{\one,\textsc{f}})$ flag wires, and $n\subcap{g} = (n\subcap{\zero,g} + n\subcap{\one,g})$ garbage wires:

	
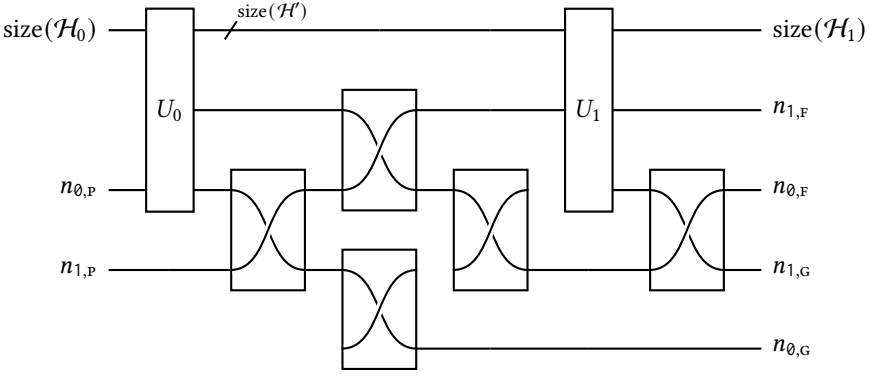
\[
		\begin{quantikz}
			\lstick{$\size(\Hilb_0)$} & \gate[3,nwires=2]{U_0} & \qwbundle{\size(\Hilb')} & \qw & \qw & \gate[3,nwires=3]{U_1} & \qw & \rstick{$\size(\Hilb_1)$} \qw \\
																& & \qw & \gate[swap]{} & \qw & & \qw & \rstick{$n\subcap{\one,f}$} \qw \\
			\lstick{$n\subcap{\zero,p}$} & & \gate[swap]{} & & \gate[swap,nwires=2]{} & & \gate[swap]{} & \rstick{$n\subcap{\zero,f}$} \qw \\
			\lstick{$n\subcap{\one,p}$} & \qw & & \gate[swap,nwires=2]{} & & \qw & & \rstick{$n\subcap{\one,g}$} \qw \\
																	& & & & \qw & \qw & \qw & \rstick{$n\subcap{\zero,g}$} \qw
		\end{quantikz}
	\]
	As with the other compositions, we're grouping the flags from the two together and grouping the garbage from the two together.
	This is essentially the same as Lemma~\ref{lem:compose}, but relying on Lemma~\ref{lem:invalidflagmix} rather than Lemma~\ref{lem:invalidflag}.
	This diagram uses a couple of \textsc{swap} gates that do not correspond to physical gates but are just there to ensure wires in the diagram don't collide.
	This is necessary because of the way that we are using single wires to represent different numbers of qubits, for example the unitaries $U_0$ and $U_1$ would have the same number of input qubits as output qubits.
\end{proof}

\subsection{Higher-level circuits from Qunity}
\label{app:hl-qunity}

Here we give the compiled circuits for all of the typing judgment cases and demonstrate algebraically that they are correct.

\paragraph{\textsc{T-Unit}}
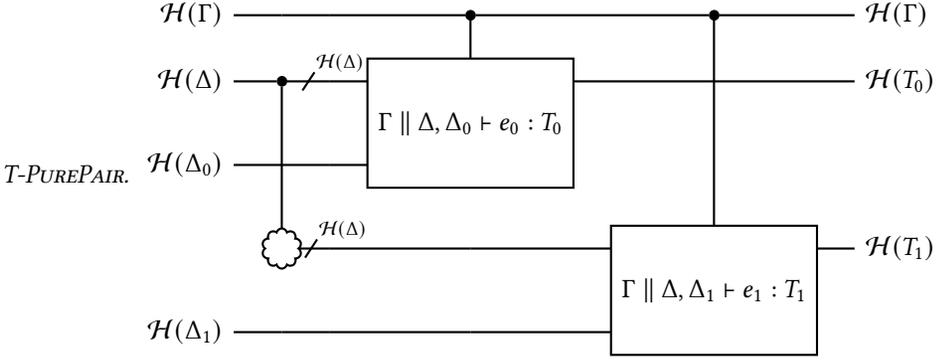
\begin{quantikz}
	\lstick{$\Hilb(\Gamma)$} & \rstick{$\Hilb(\Gamma)$} \qw \\
	\lstick{$\Hilb(\varnothing)$} & \rstick{$\Hilb(\Unit)$} \\
\end{quantikz}
\begin{alignat*}{2}
	&&\;& \ket{\sigma, \varnothing} \in \Hilb(\Gamma) \otimes \Hilb(\varnothing) \\
	&=&& \ket{\sigma, \unit} \in \Hilb(\Gamma) \otimes \Hilb(\Unit)
\end{alignat*}

\paragraph{\textsc{T-Cvar}}
\begin{quantikz}
	\lstick{$\Hilb(\Gamma)$} & \qw & \rstick{$\Hilb(\Gamma)$} \qw \\
	\lstick{$\Hilb(x : T)$} & \ctrl{2} & \rstick{$\Hilb(x : T)$} \qw \\
	\lstick{$\Hilb(\Gamma')$} & \qw & \rstick{$\Hilb(\Gamma')$} \qw \\
	\lstick{$\Hilb(\varnothing)$} & \gate[style={cloud},nwires=1]{} & \rstick{$\Hilb(T)$} \qw
\end{quantikz}
\begin{alignat*}{2}
	&&\;& \ket{\sigma, x \mapsto v, \sigma'} \in \Hilb(\Gamma, x : T, \Gamma') \\
	&\mapsto&& \ket{\sigma, x \mapsto v, \sigma', v} \in \Hilb(\Gamma, x : T, \Gamma') \otimes \Hilb(T)
\end{alignat*}

\paragraph{\textsc{T-Qvar}}
\begin{quantikz}
	\lstick{$\Hilb(\Gamma)$} & \rstick{$\Hilb(\Gamma)$} \qw \\
	\lstick{$\Hilb(x : T)$} & \rstick{$\Hilb(T)$} \qw
\end{quantikz}
\begin{alignat*}{2}
	&&\;& \ket{\sigma, x \mapsto v} \in \Hilb(\Gamma, x : T) \\
	&=&& \ket{\sigma, v} \in \Hilb(\Gamma) \otimes \Hilb(T)
\end{alignat*}

\paragraph{\textsc{T-PurePair}}
\begin{quantikz}
	\lstick{$\Hilb(\Gamma)$} & \qw & \qw & \ctrl{1} & \ctrl{3} & \qw \rstick{$\Hilb(\Gamma)$} \\
	\lstick{$\Hilb(\Delta)$} & \ctrl{2} & \qwbundle{\Hilb(\Delta)} & \gate[2]{{\Gamma \partition \Delta, \Delta_0 \vdash e_0 : T_0}} & \qw & \qw \rstick{$\Hilb(T_0)$} \\
	\lstick{$\Hilb(\Delta_0)$} & \qw & \qw & & \\
														 & \gate[style={cloud},nwires=1]{} & \qwbundle{\Hilb(\Delta)} & \qw & \gate[2]{{\Gamma \partition \Delta, \Delta_1 \vdash e_1 : T_1}} & \qw \rstick{$\Hilb(T_1)$} \\
	\lstick{$\Hilb(\Delta_1)$} & \qw & \qw & \qw & & 
\end{quantikz}
\begin{alignat*}{3}
	&&& \ket{\sigma, \tau, \tau_0, \tau_1} &&\in \Hilb(\Gamma, \Delta, \Delta_0, \Delta_1) \\
	&\mapsto&\;& \ket{\sigma, \tau, \tau_0, \tau, \tau_1} &&\in \Hilb(\Gamma, \Delta, \Delta_0, \Delta, \Delta_1) \\
	&\mapsto&\;& \ket{\sigma} \otimes \msem{\sigma : \Gamma \partition \Delta, \Delta_0 \vdash e_0 : T_0} \ket{\tau, \tau_0} \otimes \ket{\tau, \tau_1} &&\in \Hilb(\Gamma) \otimes \Hilb(T_0) \otimes \Hilb(\Delta, \Delta_1) \\
	&\mapsto&\;& \ket{\sigma} \otimes \msem{\sigma : \Gamma \partition \Delta, \Delta_0 \vdash e_0 : T_0} \ket{\tau, \tau_0} \otimes \msem{\sigma : \Gamma \partition \Delta, \Delta_1 \vdash e_1 : T_1} \ket{\tau, \tau_1} &&\in \Hilb(\Gamma) \otimes \Hilb(T_0 \otimes T_1)
\end{alignat*}

\paragraph{\textsc{T-Ctrl}}
The compiled circuit for \textsc{T-Ctrl} uses modified versions of its subcircuits.
We use Lemma~\ref{lem:purify} to get a purified version of the circuit for $e$ with semantics $\msem{e} : \Hilb(\Gamma, \Delta) \to \Hilb(T) \otimes \Hilb_{\textsc{g}}$.
Here, $\Hilb_{\textsc{g}}$ is some ``garbage'' Hilbert space containing vectors $\{\ket{g_{\sigma, \tau, v}} : \sigma \in \V(\Gamma), \tau \in \V(\Delta), v \in \V(T)\} \subset \Hilb_{\textsc{g}}$ such that $\msem{e}\ket{\sigma, \tau} = \sum_{v \in \V(T)} \bra{g_{\sigma,\tau,v}, v} \msem{e}\ket{\sigma, \tau} \cdot \ket{g_{\sigma, \tau,v}, v}$ for all $\sigma, \tau, v$.

The circuit below is too large to fit on a single page, so the dots denote where the two pieces must fit together.
The ``$\perp$'' circuit is derived from the orthogonality judgment and defined on page~\pageref{sec:compile-ortho}.

\[
\begin{quantikz}
	\lstick{$\Hilb(\Gamma)$} & \ctrl{2} & \qwbundle{\Hilb(\Gamma)} & \qw & \qw & \qw & \qw & \qw & \qw & \qw & \qw & \qw & \qwbundle{\Hilb(\Gamma)} \rstick{$\cdots$} \\
	\lstick{$\Hilb(\Gamma')$} & \qw & \hphantomgate{} & \qw & \qw & \qw & \qw & \qw & \qw & \qw & \qw & \qw & \qwbundle{\Hilb(\Gamma')} \rstick{$\cdots$} \\
														& \gate[style={cloud},nwires=1]{} & \qwbundle{\Hilb(\Gamma)} & \gate[2]{e} & \qwbundle{\Hilb_{\textsc{g}}} & \qw & \qw & \qw & \qw & \qw & \qw & \qw & \qwbundle{\Hilb_{\textsc{g}}} \rstick{$\cdots$} \\
														& \gate[style={cloud},nwires=1]{} & \qwbundle{\Hilb(\Delta)} & & \qwbundle{\Hilb(T)} & \gate{\perp} & \qwbundle{\substack{\Hilb(T)^{\oplus n} \\ \hfill}} & \gate{\bigoplus_{j} \msem{e_j^\dagger}} & \qwbundle{\substack{\bigoplus_j \Hilb(\Gamma_j) \\ \hfill}} & \gate[3]{\textsc{distr}} \\
	\lstick{$\Hilb(\Delta)$} & \ctrl{-1} & \qwbundle{\Hilb(\Delta)} & \qw & \qw & \qw & \qw & \qw & \qw & & \qwbundle{\substack{\bigoplus_j \Hilb(\Gamma_j, \Delta, \Delta') \\ \hfill}} & \gate{\bigoplus_j \msem{e_j'}} & \qwbundle{\bigoplus_j (\Hilb(\Gamma_j) \otimes \Hilb(T'))} \rstick{$\cdots$} \\
	\lstick{$\Hilb(\Delta')$} & \qw & \qw & \qw & \qw & \qw & \qw & \qw & \qw & \\
\end{quantikz}
\]
\[
\begin{quantikz}
	\lstick{$\cdots$} & \qwbundle{\Hilb(\Gamma)} & \qw & \qw & \ctrl{1} & \qw & \qw & \qw & \qw & \qw & \ctrl{2} & \qw \rstick{$\Hilb(\Gamma)$} \\
	\lstick{$\cdots$} & \qwbundle{\Hilb(\Gamma')} & \qw & \qw & \ctrl{2} & \qw & \qw & \qw & \qw & \qw & \qw & \qw \rstick{$\Hilb(\Gamma')$} \\
	\lstick{$\cdots$} & \qwbundle{\Hilb_{\textsc{g}}} & \qw & \qw & \qw & \qw & \qw & \qw & \gate[2]{e^\dagger} & \qwbundle{\Hilb(\Gamma)} & \gate[style=cloud]{} \\
										& \hphantom{ext} & \gate[2,nwires=1]{\textsc{distr}} & \qwbundle{\substack{\bigoplus_j \Hilb(\Gamma_j) \\ \\ \hfill}} & \gate{\bigoplus_j \msem{e_j}} & \qwbundle{\substack{\Hilb(T)^{\oplus n} \\ \hfill}} & \gate{\perp^\dagger} & \qwbundle{\Hilb(T)} & & \qwbundle{\Hilb(\Delta)} & \gate[2]{\textsc{erase}} \\
	\lstick{$\cdots$} & \qwbundle{\bigoplus_j \cdots} & & \qwbundle{\Hilb(T')} & \qw & \qw & \qw & \qw & \qw & \qw & & \qw \rstick{$\Hilb(T')$} \\
\end{quantikz}
\]

\begin{alignat*}{2}
	&&& \ket{\sigma, \sigma', \tau, \tau'} \\ &&&\in \Hilb(\Gamma, \Gamma', \Delta, \Delta') \\
	&\mapsto&\;& \ket{\sigma, \sigma', \sigma, \tau, \tau, \tau'} \\ &&&\in \Hilb(\Gamma, \Gamma', \Gamma, \Delta, \Delta, \Delta') \\
	&\mapsto&\;& \ket{\sigma, \sigma'} \otimes \msem{e} \ket{\sigma, \tau} \otimes \ket{\tau, \tau'} \\ %
	&=&& \ket{\sigma, \sigma'} \otimes \sum_{v \in \V(T)} \bra{g_{\sigma,\tau,v}, v} \msem{e} \ket{\sigma, \tau} \cdot \ket{g_{\sigma,\tau,v}, v} \otimes \ket{\tau, \tau'} \\ &&&\in \Hilb(\Gamma, \Gamma') \otimes \Hilb_{\textsc{g}} \otimes \Hilb(T) \otimes \Hilb(\Delta, \Delta') \\
	&\mapsto&& \ket{\sigma, \sigma'} \otimes \sum_{v \in \V(T)} \bra{g_{\sigma,\tau,v}, v} \msem{e} \ket{\sigma, \tau} \cdot \ket{g_{\sigma,\tau,v}} \otimes \left(\bigoplus_{j=1}^n \sum_{\sigma_j \in \V(\Gamma_j)} \bra{\sigma_j} \msem{e_j^\dagger} \ket{v} \cdot \ket{v} \right) \otimes \ket{\tau, \tau'} \\ &&&\in \Hilb(\Gamma, \Gamma') \otimes \Hilb_{\textsc{g}} \otimes \Hilb(T)^{\oplus n} \otimes \Hilb(\Delta, \Delta') \\
	&\mapsto&& \ket{\sigma, \sigma'} \otimes \sum_{v \in \V(T)} \bra{g_{\sigma,\tau,v}, v} \msem{e} \ket{\sigma, \tau} \cdot \ket{g_{\sigma,\tau,v}} \otimes \left(\bigoplus_{j=1}^n \sum_{\sigma_j \in \V(\Gamma_j)} \bra{\sigma_j} \msem{e_j^\dagger} \ket{v} \cdot \msem{e_j^\dagger} \ket{v} \right) \otimes \ket{\tau, \tau'} \\
	&=&& \ket{\sigma, \sigma'} \otimes \sum_{v \in \V(T)} \bra{g_{\sigma,\tau,v}, v} \msem{e} \ket{\sigma, \tau} \cdot \ket{g_{\sigma,\tau,v}} \otimes \left(\bigoplus_{j=1}^n \sum_{\sigma_j \in \V(\Gamma_j)} \bra{\sigma_j} \msem{e_j^\dagger} \ket{v} \cdot \bra{\sigma_j} \msem{e_j^\dagger} \ket{v} \cdot \ket{\sigma_j} \right) \otimes \ket{\tau, \tau'} \\
	&=&& \ket{\sigma, \sigma'} \otimes \sum_{v \in \V(T)} \bra{g_{\sigma,\tau,v}, v} \msem{e} \ket{\sigma, \tau} \cdot \ket{g_{\sigma,\tau,v}} \otimes \left(\bigoplus_{j=1}^n \sum_{\sigma_j \in \V(\Gamma_j)} \bra{\sigma_j} \msem{e_j^\dagger} \ket{v} \cdot \ket{\sigma_j} \right) \otimes \ket{\tau, \tau'} \\ &&&\in \Hilb(\Gamma, \Gamma') \otimes \Hilb_{\textsc{g}} \otimes \bigoplus_{j=1}^n \Hilb(\Gamma_j) \otimes \Hilb(\Delta, \Delta') \\
	&\mapsto&& \ket{\sigma, \sigma'} \otimes \sum_{v \in \V(T)} \bra{g_{\sigma,\tau,v}, v} \msem{e} \ket{\sigma, \tau} \cdot \ket{g_{\sigma,\tau,v}} \otimes \left(\bigoplus_{j=1}^n \sum_{\sigma_j \in \V(\Gamma_j)} \bra{\sigma_j} \msem{e_j^\dagger} \ket{v} \cdot \ket{\sigma_j, \tau, \tau'} \right) \\ &&&\in \Hilb(\Gamma, \Gamma') \otimes \Hilb_{\textsc{g}} \otimes \bigoplus_{j=1}^n \left(\Hilb(\Gamma_j, \Delta, \Delta')\right) \\
	&\mapsto&& \ket{\sigma, \sigma'} \otimes \sum_{v \in \V(T)} \bra{g_{\sigma,\tau,v}, v} \msem{e} \ket{\sigma, \tau} \cdot \ket{g_{\sigma,\tau,v}} \otimes \left(\bigoplus_{j=1}^n \sum_{\sigma_j \in \V(\Gamma_j)} \bra{\sigma_j} \msem{e_j^\dagger} \ket{v} \cdot \msem{e_j'} \ket{\sigma_j, \tau, \tau'} \right)
\end{alignat*}
\begin{alignat*}{3}
	&&&\ket{\sigma, \sigma'} \otimes \sum_{v \in \V(T)} \bra{g_{\sigma,\tau,v}, v} \msem{e} \ket{\sigma, \tau} \cdot \ket{g_{\sigma,\tau,v}} \\ &&&\qquad\otimes \left(\bigoplus_{j=1}^n \sum_{\sigma_j \in \V(\Gamma_j)} \bra{\sigma_j} \msem{e_j^\dagger} \ket{v} \cdot \sum_{v' \in \V(T')} \bra{\sigma_j, v'} \msem{e_j'} \ket{\sigma_j, \tau, \tau'} \cdot \ket{\sigma_j, v'} \right) \\ %
	&=&&\ket{\sigma, \sigma'} \otimes \sum_{v \in \V(T)} \bra{g_{\sigma,\tau,v}, v} \msem{e} \ket{\sigma, \tau} \cdot \ket{g_{\sigma,\tau,v}} \\ &&&\qquad\otimes \left(\sum_{j=1}^n \sum_{\sigma_j \in \V(\Gamma_j)} \bra{\sigma_j} \msem{e_j^\dagger} \ket{v} \cdot \sum_{v' \in \V(T')} \bra{\sigma_j, v'} \msem{e_j'} \ket{\sigma_j, \tau, \tau'} \cdot \inj_j \ket{\sigma_j, v'} \right) \\ &&&\in \Hilb(\Gamma, \Gamma') \otimes \Hilb_{\textsc{g}} \otimes \bigoplus_{j=1}^n \left(\Hilb(\Gamma_j) \otimes \Hilb(T')\right) \\
	&\mapsto&&\ket{\sigma, \sigma'} \otimes \sum_{v \in \V(T)} \bra{g_{\sigma,\tau,v}, v} \msem{e} \ket{\sigma, \tau} \cdot \ket{g_{\sigma,\tau,v}} \\ &&&\qquad\otimes \left(\sum_{j=1}^n \sum_{\sigma_j \in \V(\Gamma_j)} \bra{\sigma_j} \msem{e_j^\dagger} \ket{v} \cdot \inj_j \ket{\sigma_j} \otimes \sum_{v' \in \V(T')} \bra{\sigma_j, v'} \msem{e_j'} \ket{\sigma_j, \tau, \tau'} \cdot \ket{v'} \right) \\ &&&\in \Hilb(\Gamma, \Gamma') \otimes \Hilb_{\textsc{g}} \otimes \bigoplus_{j=1}^n \Hilb(\Gamma_j) \otimes \Hilb(T') \\
	&\mapsto&&\ket{\sigma, \sigma'} \otimes \sum_{v \in \V(T)} \bra{g_{\sigma,\tau,v}, v} \msem{e} \ket{\sigma, \tau} \cdot \ket{g_{\sigma,\tau,v}} \\ &&&\qquad\otimes \left(\sum_{j=1}^n \sum_{\sigma_j \in \V(\Gamma_j)} \bra{\sigma_j} \msem{e_j^\dagger} \ket{v} \cdot \inj_j \msem{e_j} \ket{\sigma_j} \otimes \sum_{v' \in \V(T')} \bra{\sigma_j, v'} \msem{e_j'} \ket{\sigma_j, \tau, \tau'} \cdot \ket{v'} \right) \\ %
	&=&&\ket{\sigma, \sigma'} \otimes \sum_{v \in \V(T)} \bra{g_{\sigma,\tau,v}, v} \msem{e} \ket{\sigma, \tau} \cdot \ket{g_{\sigma,\tau,v}} \\ &&&\qquad\otimes \left(\sum_{j=1}^n \sum_{\sigma_j \in \V(\Gamma_j)} \bra{\sigma_j} \msem{e_j^\dagger} \ket{v} \cdot \inj_j \ket{v} \otimes \sum_{v' \in \V(T')} \bra{\sigma_j, v'} \msem{e_j'} \ket{\sigma_j, \tau, \tau'} \cdot \ket{v'} \right) \\ &&&\in \Hilb(\Gamma, \Gamma') \otimes \Hilb_{\textsc{g}} \otimes \bigoplus_{j=1}^n \Hilb(T) \otimes \Hilb(T') \\
	&\mapsto&&\ket{\sigma, \sigma'} \otimes \sum_{v \in \V(T)} \bra{g_{\sigma,\tau,v}, v} \msem{e} \ket{\sigma, \tau} \cdot \ket{g_{\sigma,\tau,v}} \\ &&&\qquad\otimes \left(\sum_{j=1}^n \sum_{\sigma_j \in \V(\Gamma_j)} \bra{\sigma_j} \msem{e_j^\dagger} \ket{v} \cdot \ket{v} \otimes \sum_{v' \in \V(T')} \bra{\sigma_j, v'} \msem{e_j'} \ket{\sigma_j, \tau, \tau'} \cdot \ket{v'} \right) \\ &&&\in \Hilb(\Gamma, \Gamma') \otimes \Hilb_{\textsc{g}} \otimes \Hilb(T) \otimes \Hilb(T') \\
	&\mapsto&&\ket{\sigma, \sigma'} \otimes \sum_{v \in \V(T)} \bra{g_{\sigma,\tau,v}, v} \msem{e} \ket{\sigma, \tau} \\ &&&\qquad\otimes \left(\sum_{j=1}^n \sum_{\sigma_j \in \V(\Gamma_j)} \bra{\sigma_j} \msem{e_j^\dagger} \ket{v} \cdot \msem{e^\dagger} \ket{g_{\sigma,\tau,v},v} \otimes \sum_{v' \in \V(T')} \bra{\sigma_j, v'} \msem{e_j'} \ket{\sigma_j, \tau, \tau'} \cdot \ket{v'} \right) \\ &&&\in \Hilb(\Gamma, \Gamma', \Gamma, \Delta) \otimes \Hilb(T') \\
\end{alignat*}
\begin{alignat*}{3}
	&&&\ket{\sigma, \sigma'} \otimes \sum_{v \in \V(T)} \bra{g_{\sigma,\tau,v}, v} \msem{e} \ket{\sigma, \tau} \\ &&&\qquad\otimes \left(\sum_{j=1}^n \sum_{\sigma_j \in \V(\Gamma_j)} \bra{\sigma_j} \msem{e_j^\dagger} \ket{v} \cdot \sum_{\sigma_\star, \tau_\star} \bra{\sigma_\star,\tau_\star}\msem{e^\dagger} \ket{g_{\sigma,\tau,v},v} \cdot \ket{\sigma_\star,\tau_\star} \otimes \sum_{v' \in \V(T')} \bra{\sigma_j, v'} \msem{e_j'} \ket{\sigma_j, \tau, \tau'} \cdot \ket{v'} \right) \\ &&&\in \Hilb(\Gamma, \Gamma', \Gamma, \Delta) \otimes \Hilb(T') \\
	&\mapsto&&\ket{\sigma, \sigma'} \otimes \sum_{v \in \V(T)} \bra{g_{\sigma,\tau,v}, v} \msem{e} \ket{\sigma, \tau} \\ &&&\qquad\otimes \left(\sum_{j=1}^n \sum_{\sigma_j \in \V(\Gamma_j)} \bra{\sigma_j} \msem{e_j^\dagger} \ket{v} \cdot \sum_{\tau_\star} \bra{\sigma,\tau_\star}\msem{e^\dagger} \ket{g_{\sigma,\tau,v},v} \cdot \ket{\tau_\star} \otimes \sum_{v' \in \V(T')} \bra{\sigma_j, v'} \msem{e_j'} \ket{\sigma_j, \tau, \tau'} \cdot \ket{v'} \right) \\ &&&\in \Hilb(\Gamma, \Gamma', \Delta) \otimes \Hilb(T') \\
	&=&&\ket{\sigma, \sigma'} \otimes \sum_{v \in \V(T)} \bra{g_{\sigma,\tau,v}, v} \msem{e} \ket{\sigma, \tau} \bra{\sigma,\tau}\msem{e^\dagger} \ket{g_{\sigma,\tau,v},v}\\ &&&\qquad\otimes \left(\sum_{j=1}^n \sum_{\sigma_j \in \V(\Gamma_j)} \bra{\sigma_j} \msem{e_j^\dagger} \ket{v} \cdot \ket{\tau} \otimes \sum_{v' \in \V(T')} \bra{\sigma_j, v'} \msem{e_j'} \ket{\sigma_j, \tau, \tau'} \cdot \ket{v'} \right) \\ &&&\in \Hilb(\Gamma, \Gamma', \Delta) \otimes \Hilb(T') \\
	&\mapsto&&\ket{\sigma, \sigma'} \otimes \sum_{v \in \V(T)} \bra{g_{\sigma,\tau,v}, v} \msem{e} \ket{\sigma, \tau} \bra{\sigma,\tau}\msem{e^\dagger} \ket{g_{\sigma,\tau,v},v}\\ &&&\qquad\otimes \left(\sum_{j=1}^n \sum_{\sigma_j \in \V(\Gamma_j)} \bra{\sigma_j} \msem{e_j^\dagger} \ket{v} \otimes \sum_{v' \in \V(T')} \bra{\sigma_j, v'} \msem{e_j'} \ket{\sigma_j, \tau, \tau'} \cdot \ket{v'} \right) \\ &&&\in \Hilb(\Gamma, \Gamma') \otimes \Hilb(T') \\
\end{alignat*}

\paragraph{\textsc{T-PureApp}}
\begin{quantikz}
	\lstick{$\Hilb(\Gamma)$} & \ctrl{1} & \qw & \qw & \rstick{$\Hilb(\Gamma)$} \qw \\
	\lstick{$\Hilb(\Delta)$} & \gate{e} & \qwbundle{\Hilb(T)} & \gate{f} & \rstick{$\Hilb(T')$} \qw
\end{quantikz}
\begin{alignat*}{2}
	&&\;& \ket{\sigma, \tau} \in \Hilb(\Gamma, \Delta) \\
	&\mapsto&& \ket{\sigma} \otimes \msem{e} \ket{\tau} \in \Hilb(\Gamma) \otimes \Hilb(T) \\
	&\mapsto&& \ket{\sigma} \otimes \msem{f} \msem{e} \ket{\tau} \in \Hilb(\Gamma) \otimes \Hilb(T') \\
\end{alignat*}

\paragraph{\textsc{T-PurePerm}}
\[
\begin{quantikz}
	\lstick{$\Hilb(\pi\subcap{g}(\Gamma))$} & \gate{\pi\subcap{g}^{-1}} & \qwbundle{\Hilb(\Gamma)} & \ctrl{1} & \gate{\pi\subcap{g}} & \rstick{$\Hilb(\pi\subcap{g}(\Gamma))$} \qw \\
	\lstick{$\Hilb(\pi\subcap{d}(\Delta))$} & \gate{\pi\subcap{d}^{-1}} & \qwbundle{\Hilb(\Delta)} & \gate{e} & \qw & \rstick{$\Hilb(T)$} \qw
\end{quantikz}
\]
\begin{alignat*}{2}
	&&\;& \ket{\pi\subcap{g}(\sigma), \pi\subcap{d}(\tau)} \in \Hilb(\pi\subcap{g}(\Gamma), \pi\subcap{d}(\Delta)) \\
	&\mapsto&& \ket{\sigma, \tau} \in \Hilb(\Gamma, \Delta) \\
	&\mapsto&& \ket{\sigma} \otimes \msem{e} \ket{\tau} \in \Hilb(\Gamma) \otimes \Hilb(T) \\
	&\mapsto&& \ket{\pi\subcap{g}(\sigma)} \otimes \msem{e} \ket{\tau} \in \Hilb(\pi\subcap{g}(\Gamma)) \otimes \Hilb(T)
\end{alignat*}

\paragraph{\textsc{T-Gate}} (We assume that our low-level circuits contain these gates as primitives.)
\[
\begin{quantikz}
	\lstick{$\Hilb(\Bit)$} & \gate{\begin{matrix} \cos(r_\theta / 2) & -e^{i r_\lambda} \sin(r_\theta / 2) \\ e^{i r_\phi} \sin(r_\theta / 2) & e^{i(r_\phi + r_\lambda)} \cos(r_\theta / 2) \end{matrix}} & \rstick{$\Hilb(\Bit)$} \qw
\end{quantikz}
\]
\begin{alignat*}{3}
	&&\;& \ket{\zero} &&\in \Hilb(\Bit) \\
	&\mapsto&& \cos(r_\theta / 2) \ket{\zero} + e^{i r_\phi} \sin(r_\theta / 2) \ket{\one} &&\in \Hilb(\Bit) \\
	&&\;& \ket{\one} &&\in \Hilb(\Bit) \\
	&\mapsto&& -e^{i r_\lambda} \sin(r_\theta / 2) \ket{\zero} + e^{i(r_\phi + r_\lambda)} \cos(r_\theta / 2) \ket{\one} &&\in \Hilb(\Bit)
\end{alignat*}

\paragraph{\textsc{T-Left}} (This was already implemented in the previous section.)
\[
\begin{quantikz}
	\lstick{$\Hilb(T_0)$} & \gate{\lef{T_0}{T_1}} & \rstick{$\Hilb(T_0 \oplus T_1)$} \qw
\end{quantikz}
\]
\begin{alignat*}{3}
	&&\;& \ket{v} &&\in \Hilb(T_0) \\
	&\mapsto&& \ket{v} \oplus 0 \\
\end{alignat*}

\paragraph{\textsc{T-Right}} (This was already implemented in the previous section.)
\[
\begin{quantikz}
	\lstick{$\Hilb(T_1)$} & \gate{\rit{T_0}{T_1}} & \rstick{$\Hilb(T_0 \oplus T_1)$} \qw
\end{quantikz}
\]
\begin{alignat*}{3}
	&&\;& \ket{v} &&\in \Hilb(T_1) \\
	&\mapsto&& 0 \oplus \ket{v} \\
\end{alignat*}

\paragraph{\textsc{T-PureAbs}}
\begin{quantikz}
	\lstick{$\Hilb(T)$} & \gate{e^\dagger} & \qwbundle{\Hilb(\Delta)} & \gate{e'} &\rstick{$\Hilb(T')$} \qw
\end{quantikz}
\begin{alignat*}{3}
	&&\;& \ket{v} &&\in \Hilb(T) \\
	&\mapsto&& \msem{e}^\dagger \ket{v} &&\in \Hilb(\Delta) \\
	&\mapsto&& \msem{e'} \msem{e}^\dagger \ket{v} &&\in \Hilb(T') \\
\end{alignat*}

\paragraph{\textsc{T-Rphase}}
Let $E\subcap{f} : \Hilb(T) \to \Hilb(\Delta) \oplus \Hilb\subcap{f}$ be the norm-preserving operator constructed from $\msem{\varnothing \partition \Delta \vdash e : T}^\dagger$ using Lemma~\ref{lem:pure-error}.
The compiled \textsc{T-Rphase} circuit then looks like this:
\[
	\begin{DIFnomarkup}
\begin{quantikz}
	\lstick{$\Hilb(T)$} & \gate{E\subcap{f}} & \qwbundle{\Hilb(\Delta) \oplus \Hilb\subcap{f}} & \hphantomgate{ex} & \gate{e^{i r} \mathbb{I}_\Delta \oplus e^{i r'} \mathbb{I}\subcap{f}} & \qwbundle{\Hilb(\Delta) \oplus \Hilb\subcap{f}} & \hphantomgate{ex} & \gate{E\subcap{f}^\dagger} &\rstick{$\Hilb(T)$} \qw
\end{quantikz}
\end{DIFnomarkup}
\]
Using the fact that $\msem{\varnothing \partition \Delta \vdash e : T}^\dagger = \msem{\lef{\Delta}{\textsc{G}}}^\dagger E\subcap{f}$,
\begin{align*}
	& E\subcap{f}^\dagger \left(e^{i r} \mathbb{I}_\Delta \oplus e^{i r'} \mathbb{I}\subcap{f}\right) E\subcap{f} \\
	=&\; E\subcap{f}^\dagger \left(e^{i r} \msem{\lef{\Delta}{\textsc{G}}} \msem{\lef{\Delta}{\textsc{G}}}^\dagger + e^{i r'} \msem{\rit{\Delta}{\textsc{G}}}\msem{\rit{\Delta}{\textsc{G}}}^\dagger\right) E\subcap{f} \\
	=&\; E\subcap{f}^\dagger \left(e^{i r} \msem{\lef{\Delta}{\textsc{G}}} \msem{\lef{\Delta}{\textsc{G}}}^\dagger + e^{i r'} (\mathbb{I} - \msem{\lef{\Delta}{\textsc{G}}} \msem{\lef{\Delta}{\textsc{G}}}^\dagger)\right) E\subcap{f} \\
	=&\; e^{i r} E\subcap{f}^\dagger \msem{\lef{\Delta}{\textsc{G}}} \msem{\lef{\Delta}{\textsc{G}}}^\dagger E\subcap{f} + e^{i r'} \left(\mathbb{I} - E\subcap{f}^\dagger \msem{\lef{\Delta}{\textsc{G}}} \msem{\lef{\Delta}{\textsc{G}}}^\dagger E\subcap{f}\right) \\
	=&\; e^{i r} \msem{\varnothing \partition \Delta \vdash e : T} \msem{\varnothing \partition \Delta \vdash e : T}^\dagger + e^{i r'} \left(\mathbb{I} - \msem{\varnothing \partition \Delta \vdash e : T} \msem{\varnothing \partition \Delta \vdash e : T}^\dagger\right) \\
\end{align*}

\paragraph{\textsc{T-Mix}}
\begin{quantikz}
	\lstick{$\Hilb(\Delta)$} & \gate{e} & \rstick{$\Hilb(T)$} \qw
\end{quantikz}
\begin{alignat*}{2}
	&&\;& \op{\tau}{\tau'} \in \linear(\Hilb(\Delta)) \\
	&\mapsto&& \msem{e} \op{\tau}{\tau'} \msem{e}^\dagger \in \linear(\Hilb(T)) \\
\end{alignat*}

\paragraph{\textsc{T-MixedPair}}
\begin{quantikz}
	\lstick{$\Hilb(\Delta)$} & \ctrl{2} & \qwbundle{\Hilb(\Delta)} & \gate[2]{{\Gamma \partition \Delta, \Delta_0 \vdash e_0 : T_0}} & \qw \rstick{$\Hilb(T_0)$} \\
	\lstick{$\Hilb(\Delta_0)$} & \qw & \qw & & \\
														 & \gate[style={cloud},nwires=1]{} & \qwbundle{\Hilb(\Delta)} & \gate[2]{{\Gamma \partition \Delta, \Delta_1 \vdash e_1 : T_1}} & \qw \rstick{$\Hilb(T_1)$} \\
	\lstick{$\Hilb(\Delta_1)$} & \qw & \qw & & 
\end{quantikz}
\begin{alignat*}{3}
	&&\;& \op{\tau, \tau_0, \tau_1}{\tau', \tau_0', \tau_1'} &&\in \linear(\Hilb(\Delta, \Delta_0, \Delta_1)) \\
	&\mapsto&& \op{\tau, \tau_0, \tau, \tau_1}{\tau', \tau_0', \tau', \tau_1'} \\
	&=&& \op{\tau,\tau_0}{\tau',\tau_0'} \otimes \op{\tau,\tau_1}{\tau',\tau_1'} &&\in \linear(\Hilb(\Delta, \Delta_0)) \otimes \linear(\Hilb(\Delta, \Delta_1)) \\
	&\mapsto&& \msem{e_0} \left( \op{\tau,\tau_0}{\tau',\tau_0'} \right) \otimes \msem{e_1} \left( \op{\tau,\tau_1}{\tau',\tau_1'} \right) &&\in \linear(\Hilb(T_0 \otimes T_1)) \\
\end{alignat*}

\paragraph{\textsc{T-Try}}
This circuit relies on circuits transformed by Lemma~\ref{lem:mixed-error}.
	\[
\begin{quantikz}
	\lstick{$\Hilb(\Delta_0)$} & \gate{{\cptp(\Delta_0 \Vdash e_0 : T)}} & \qwbundle{\Hilb(T) \oplus \complex} &[1cm] \gate[2]{\textsc{distr}} & \qwbundle{\substack{\Hilb(T) \otimes (\Hilb(T) \oplus \complex) \oplus \Hilb(T) \oplus \complex \\ \\ \hfill}} & \gate{\texttt{left}^\dagger} & \qwbundle{\substack{\Hilb(T) \otimes (\Hilb(T) \oplus \complex) \oplus \Hilb(T) \\ \hfill}} & \gate[2,nwires=2]{\textsc{distr}} & \rstick{$\Hilb(T)$} \qw \\
	\lstick{$\Hilb(\Delta_1)$} & \gate{{\cptp(\Delta_1 \Vdash e_1 : T)}} & \qwbundle{\Hilb(T) \oplus \complex} & &&&&& \trash{\Hilb(T) \oplus \complex \oplus \complex}
\end{quantikz}
\]
For brevity, define:
\begin{align*}
	\rho_0' &\defeq \msem{\Delta_0 \Vdash e_0 : T}(\rho_0) \\
	\rho_1' &\defeq \msem{\Delta_1 \Vdash e_1 : T}(\rho_1)
\end{align*}
Then, the circuit acts as follows:
\begin{alignat*}{3}
	&&\;& \rho_0 \otimes \rho_1 \\
	&\mapsto&& \left(\rho_0' \oplus (\tr\left(\rho_0\right) - \tr(\rho_0'))\right) \otimes \left(\rho_1' \oplus (\tr\left(\rho_1\right) - \tr(\rho_1')) \right) \\
	&\mapsto&& \rho_0' \otimes \left( \rho_1' \oplus (\tr\left(\rho_1\right) - \tr(\rho_1'))\right) \oplus (\tr\left(\rho_0\right) - \tr(\rho_0')) \left( \rho_1' \oplus (\tr\left(\rho_1\right) - \tr(\rho_1')) \right) \\
	&\mapsto&& \rho_0' \otimes \left( \rho_1' \oplus (\tr\left(\rho_1\right) - \tr(\rho_1'))\right) \oplus (\tr\left(\rho_0\right) - \tr(\rho_0')) \rho_1' \\
	&\mapsto&& \rho_0' \otimes \left( \rho_1' \oplus (\tr\left(\rho_1\right) - \tr(\rho_1')) \oplus 0 \right) + \rho_1' \left(0 \oplus 0 \oplus (\tr\left(\rho_0\right) - \tr(\rho_0'))\right) \\
	&\mapsto&& \tr\left(\rho_1\right) \rho_0' + (\tr\left(\rho_0\right) - \tr(\rho_0'))\rho_1'
\end{alignat*}

\paragraph{\textsc{T-MixedApp}}
\begin{quantikz}
	\lstick{$\Hilb(\Delta)$} & \gate{e} & \qwbundle{\Hilb(T)} & \gate{f} & \rstick{$\Hilb(T')$} \qw
\end{quantikz}
\begin{alignat*}{2}
	&&\;& \op{\tau}{\tau'} \in \linear(\Hilb(\Delta)) \\
	&\mapsto&& \msem{e} \left( \op{\tau}{\tau'} \right) \in \linear(\Hilb(T)) \\
	&\mapsto&& \msem{f} \left( \msem{e} ( \op{\tau}{\tau'} ) \right) \in \linear(\Hilb(T'))
\end{alignat*}

\paragraph{\textsc{T-MixedPerm}}
\[
\begin{quantikz}
	\lstick{$\Hilb(\pi(\Delta))$} & \gate{\pi^{-1}} & \qwbundle{\Hilb(\Delta)} & \gate{e} & \rstick{$\Hilb(T)$} \qw
\end{quantikz}
\]
\begin{alignat*}{2}
	&&\;& \op{\pi(\tau)}{\pi(\tau')} \in \linear(\Hilb(\pi(\Delta))) \\
	&\mapsto&& \op{\tau}{\tau'} \in \linear(\Hilb(\Delta)) \\
	&\mapsto&& \msem{e} \left( \op{\tau}{\tau'} \right) \in \linear(\Hilb(T))
\end{alignat*}

\paragraph{\textsc{T-Channel}}
\begin{quantikz}
	\lstick{$\Hilb(T)$} & \gate{f} & \rstick{$\Hilb(T')$} \qw
\end{quantikz}
\begin{alignat*}{3}
	&&\;& \op{v}{v'} &&\in \linear(\Hilb(T)) \\
	&\mapsto&& \msem{f} \op{v}{v'} \msem{f}^\dagger &&\in \linear(\Hilb(T')) \\
\end{alignat*}

\paragraph{\textsc{T-MixedAbs}}
\begin{quantikz}
	\lstick{$\Hilb(T)$} & \gate[2,nwires=2]{e^\dagger} & \qwbundle{\Hilb(\Delta)} & \gate{e'} & \rstick{$\Hilb(T')$} \qw \\
											&& \trash{\Hilb(\Delta_0)}
\end{quantikz}
\begin{alignat*}{3}
	&&\;& \op{v}{v'} &&\in \linear(\Hilb(T)) \\
	&\mapsto&& \msem{e}^\dagger \op{v}{v'} \msem{e} &&\in \linear(\Hilb(\Delta, \Delta_0)) \\
	&\mapsto&& \tr_{\Delta_0} \left( \msem{e}^\dagger \op{v}{v'} \msem{e} \right) &&\in \linear(\Hilb(\Delta)) \\
	&\mapsto&& \msem{e'} \left(\tr_{\Delta_0} ( \msem{e}^\dagger \op{v}{v'} \msem{e} ) \right) &&\in \linear(\Hilb(T'))
\end{alignat*}

\begin{lemma}[orthogonality compilation]
	Suppose $(\varnothing \partition \Gamma_j \vdash e_j : T)$ for all $j \in \{1, \ldots, n\}$ and $\ortho{T}{e_1, \ldots, e_n}$.
	Then, one can implement a Kraus operator $\msem{\ortho{T}{e_1, \ldots, e_n}} : \Hilb(T) \to \Hilb(T)^{\oplus n}$ such that for all $v \in \V(T)$:
	\begin{align*}
		\msem{\ortho{T}{e_1, \ldots, e_n}} \ket{v} &= \bigoplus_{j=1}^n \sum_{\sigma_j \in \V(\Gamma_j)} \bra{\sigma_j} \msem{\varnothing : \varnothing \partition \Gamma_j \vdash e_j : T}^\dagger \ket{v} \cdot \ket{v} \\
		\msem{\ortho{T}{e_1, \ldots, e_n}}^\dagger \left(\inj_j \ket v\right) &= \sum_{\sigma_j \in \V(\Gamma_j)} \bra{\sigma_j} \msem{\varnothing : \varnothing \partition \Gamma_j \vdash e_j : T}^\dagger \ket{v} \cdot \ket{v}
	\end{align*}
\end{lemma}
\begin{proof}
	We will first prove this for the case where $\spanning{T}{e_1, \ldots, e_n}$ holds, and then generalize.
	We construct the required circuit recursively based on specific rule used to prove the spanning judgment.
	
	\paragraph{\textsc{S-Unit}}
	In this case, $n=1$, $e_1 = \unit$, $T = \Unit$, and $\Gamma_1 = \varnothing$.
	The output Hilbert space $\Hilb(\Unit)^{\oplus 1}$ is isomorphic to the input Hilbert space $\Hilb(\Unit)$, and the required behavior reduces to the identity function because $\bra{\varnothing} \msem{\varnothing : \varnothing \partition \varnothing \vdash \unit : \Unit}^\dagger \ket{\unit} = 1$, so an empty circuit suffices:
	
	\[
	\begin{quantikz}
		\lstick{$\Hilb(\Unit)$} & \qw \rstick{$\Hilb(\Unit)^{\oplus 1}$}
	\end{quantikz}
	\]

	\begin{align*}
		\msem{\spanning{\Unit}{\unit}} \ket{\unit}
		&= \ket{\unit} \\
		&= \bra{\varnothing} \msem{\varnothing : \varnothing \partition \varnothing \vdash \unit : \Unit}^\dagger \ket{\unit} \cdot \ket{\unit} \\
		&= \sum_{\sigma_1 \in \V(\varnothing)} \bra{\sigma_1} \msem{\varnothing : \varnothing \partition \varnothing \vdash \unit : \Unit}^\dagger \ket{\unit} \cdot \ket{\unit} \\
		&= \bigoplus_{j=1}^1 \sum_{\sigma_j \in \V(\Gamma_j)} \bra{\sigma_j} \msem{\varnothing : \varnothing \partition \Gamma_j \vdash e_j : \Unit}^\dagger \ket{\unit} \cdot \ket{\unit} \\
		\msem{\spanning{\Unit}{\unit}}^\dagger \left(\inj_1 \ket{\unit}\right)
		&= \ket{\unit} \\
		&= \sum_{\sigma_1 \in \V(\varnothing)} \bra{\sigma_1} \msem{\varnothing : \varnothing \partition \varnothing \vdash \unit : \Unit}^\dagger \ket{\unit} \cdot \ket{\unit}
	\end{align*}

	\paragraph{\textsc{S-Var}}
	In this case, $n=1$, $e_1=x$, and $\Gamma_1 = (x : T)$.
	Again, an identity circuit suffices:

	\[
		\begin{quantikz}
			\lstick{$\Hilb(T)$} & \qw \rstick{$\Hilb(T)^{\oplus 1}$}
		\end{quantikz}
	\]
	\begin{align*}
		\msem{\spanning{T}{x}} \ket{v}
		&= \ket{v} \\
		&= \bra{x \mapsto v} \msem{\varnothing : \varnothing \partition x : T \vdash x : T}^\dagger \ket{v} \cdot \ket{v} \\
		&= \sum_{\sigma_1 \in \V(x : T)} \bra{\sigma_1} \msem{\varnothing : \varnothing \partition x : T \vdash x : T}^\dagger \ket{v} \cdot \ket{v} \\
		&= \bigoplus_{j=1}^1 \sum_{\sigma_j \in \V(\Gamma_j)} \bra{\sigma_j} \msem{\varnothing : \varnothing \partition \Gamma_j \vdash e_j : T}^\dagger \ket{v} \cdot \ket{v} \\
		\msem{\spanning{T}{x}}^\dagger \left(\inj_1 \ket{v}\right)
		&= \ket{v} \\
		&= \sum_{\sigma_1 \in \V(x : T)} \bra{\sigma_1} \msem{\varnothing : \varnothing \partition x : T \vdash x : T}^\dagger \ket{v} \cdot \ket{v}
	\end{align*}

	\paragraph{\textsc{S-Sum}}
	In this case, our list of expressions is:
	\[
		[\lef{T}{T'} e_1, \ldots, \lef{T}{T'} e_n, \rit{T}{T'} e_1', \ldots, \rit{T}{T'} e_{n'}']
	\]
	\[
		\begin{quantikz}
			\lstick{$\Hilb(T \oplus T')$} & \gate{\texttt{spanning} \oplus \texttt{spanning}'} & \qwbundle{\substack{\Hilb(T)^{\oplus n} \oplus \Hilb(T')^{\oplus n'} \\ \\ \hfill}} & \gate{\msem{\lef{T}{T'}} \oplus \msem{\rit{T}{T'}}} & \qw \rstick{$\Hilb(T \oplus T')^{\oplus (n + n')}$}
		\end{quantikz}
	\]
	\begin{alignat*}{2}
		&&& \msem{\spanning{T\oplus T'}{\cdots}} (\ket{v} \oplus 0) \\
																				 &=&& \left(\bigoplus_{j=1}^n \sum_{\sigma_j \in \V(\Gamma_j)} \bra{\sigma_j} \msem{\varnothing : \varnothing \partition \Gamma_j \vdash e_j : T}^\dagger \ket{v} \cdot \left(\ket{v} \oplus 0\right)\right) \oplus 0 \\
																				 &=&& \left(\bigoplus_{j=1}^n \sum_{\sigma_j \in \V(\Gamma_j)} \bra{\sigma_j} \msem{\varnothing : \varnothing \partition \Gamma_j \vdash \lef{T}{T'} e_j : T \oplus T'}^\dagger (\ket{v} \oplus 0) \cdot \left(\ket{v} \oplus 0\right)\right) \oplus 0 \\
																				 &&&\msem{\spanning{T\oplus T'}{\cdots}} (0 \oplus \ket{v}) \\
																				 &=&& 0 \oplus \left(\bigoplus_{j=1}^{n'} \sum_{\sigma_j \in \V(\Gamma_j')} \bra{\sigma_j} \msem{\varnothing : \varnothing \partition \Gamma_j' \vdash e_j' : T'}^\dagger \ket{v} \cdot \left(0 \oplus \ket{v}\right)\right) \\
																				 &=&& 0 \oplus \left(\bigoplus_{j=1}^{n'} \sum_{\sigma_j \in \V(\Gamma_j')} \bra{\sigma_j} \msem{\varnothing : \varnothing \partition \Gamma_j' \vdash \rit{T}{T'} e_j' : T \oplus T'}^\dagger (0 \oplus \ket{v}) \cdot \left(0 \oplus \ket{v}\right)\right) \\
																				 &&&\msem{\spanning{T\oplus T'}{\cdots}}^\dagger \left(\inj_j (\ket{v} \oplus 0) \oplus 0\right) \\
																								 &=&& \sum_{\sigma_j \in \V(\Gamma_j)} \bra{\sigma_j} \msem{\varnothing : \varnothing \partition \Gamma_j \vdash e_j : T}^\dagger \ket{v} \cdot \ket{v} \oplus 0 \\
																								 &=&& \sum_{\sigma_j \in \V(\Gamma_j)} \bra{\sigma_j} \msem{\varnothing : \varnothing \partition \Gamma_j \vdash \lef{T}{T'} e_j : T \oplus T'}^\dagger (\ket{v} \oplus 0) \cdot (\ket{v} \oplus 0) \\
																				 &&&\msem{\spanning{T\oplus T'}{\cdots}}^\dagger \left(\inj_j (0 \oplus \ket{v}) \oplus 0\right) \\
																								 &=&\;& 0 \\
																								 &=&& \sum_{\sigma_j \in \V(\Gamma_j)} \bra{\sigma_j} \msem{\varnothing : \varnothing \partition \Gamma_j \vdash \lef{T}{T'} e_j : T \oplus T'}^\dagger (0 \oplus \ket{v}) \cdot (0 \oplus \ket{v}) \\
																				 &&&\msem{\spanning{T\oplus T'}{\cdots}}^\dagger \left(0 \oplus \inj_j (\ket{v} \oplus 0)\right) \\
																								 &=&\;& 0 \\
																								 &=&& \sum_{\sigma_j \in \V(\Gamma_j')} \bra{\sigma_j} \msem{\varnothing : \varnothing \partition \Gamma_j' \vdash \rit{T}{T'} e_j' : T \oplus T'}^\dagger (\ket{v} \oplus 0) \cdot (\ket{v} \oplus 0) \\
																				 &&&\msem{\spanning{T\oplus T'}{\cdots}}^\dagger \left(0 \oplus \inj_j (0 \oplus \ket{v})\right) \\
																								 &=&& \sum_{\sigma_j \in \V(\Gamma_j')} \bra{\sigma_j} \msem{\varnothing : \varnothing \partition \Gamma_j' \vdash e_j' : T}^\dagger \ket{v} \cdot (0 \oplus \ket{v}) \\
																								 &=&& \sum_{\sigma_j \in \V(\Gamma_j')} \bra{\sigma_j} \msem{\varnothing : \varnothing \partition \Gamma_j' \vdash \rit{T}{T'} e_j' : T \oplus T'}^\dagger (0 \oplus \ket{v}) \cdot (0 \oplus \ket{v})
	\end{alignat*}

	\paragraph{\textsc{S-Pair}}
	We can compile this judgment for $T \otimes T'$ in terms of the compiled judgments for $\spanning{T}{e_1, \ldots, e_m}$ and $\spanning{T'}{e_{j,1}', \ldots, e_{j,n_j}'}$ for $j = 1, \ldots, m$.
	The output from the $T$ compilation is used to control which of the $T'$ judgments is used.
	With abbreviation $\textsc{span}_j \defeq \spanning{T'}{e_{j,1}', \ldots, e_{j,n_j}'}$:

	\[
	\begin{quantikz}
		\lstick{$\Hilb(T)$} & \gate{\spanning{T}{\cdots}} & \qwbundle{\substack{\Hilb(T)^{\oplus m} \\ \hfill}} & \gate[2]{\textsc{distr}} & \qwbundle{\Hilb(T \otimes T')^{\oplus m}} &[1cm] \gate{\bigoplus_{j=1}^m \textsc{span}_j} & \qw \rstick{$\Hilb(T \otimes T')^{\oplus \sum_{j=1}^m n_j}$} \\
		\lstick{$\Hilb(T')$} & \qw & \qw &
	\end{quantikz}
	\]

	Assume that each pattern is well-typed as $\varnothing \partition \Gamma_{j,k} \vdash \pair{e_j}{e_{j,k}'} : T \otimes T'$.
	The spanning judgment ensures that the free variables in $e_j$ are distinct from those of $e_{j,k}$, so the context can be partitioned as $\Gamma_{j,k} = \Gamma_j, \Gamma_{j,k}'$, with $\varnothing \partition \Gamma_j \vdash e_j : T$ and $\varnothing \partition \Gamma_{j,k}' \vdash e_{j,k}' : T'$.
	\begin{alignat*}{2}
		&&\;&\msem{\spanning{T\otimes T'}{\cdots}} \ket{v, v'} \\
		&=&& \bigoplus_{j=1}^m \bigoplus_{k=1}^{n_j} \sum_{\sigma_j \in \V(\Gamma_j)} \sum_{\sigma_{j,k}' \in \V(\Gamma_{j,k}')} \bra{\sigma_j} \msem{\varnothing : \varnothing \partition \Gamma_j \vdash e_j : T}^\dagger \ket{v} \cdot \ket{v} \otimes \bra{\sigma_{j,k}'} \msem{\varnothing : \varnothing \partition \Gamma_{j,k}' \vdash e_{j,k}' : T'}^\dagger \ket{v'} \cdot \ket{v'} \\
		&=&& \bigoplus_{j=1}^n \bigoplus_{k=1}^{n_j} \sum_{\sigma_j \in \V(\Gamma_j)} \sum_{\sigma_{j,k}' \in \V(\Gamma_{j,k}')} \bra{\sigma_j, \sigma_{j,k}'} \msem{\varnothing : \varnothing \partition \Gamma_j, \Gamma_{j,k}' \vdash \pair{e_j}{e_{j,k}'} : T \otimes T'}^\dagger \ket{v, v'} \cdot \ket{v, v'} \\
		&&\;&\msem{\spanning{T\otimes T'}{\cdots}}^\dagger \inj_j \inj_{k} \ket{v, v'} \\
		&=&& \sum_{\sigma_j \in \V(\Gamma_j)} \sum_{\sigma_{j,k}' \in \V(\Gamma_{j,k}')} \bra{\sigma_j} \msem{\varnothing : \varnothing \partition \Gamma_j \vdash e_j : T}^\dagger \ket{v} \cdot \ket{v} \otimes \bra{\sigma_{j,k}'} \msem{\varnothing : \varnothing \partition \Gamma_{j,k}' \vdash e_{j,k}' : T'}^\dagger \ket{v'} \cdot \ket{v'} \\
		&=&& \sum_{\sigma_j \in \V(\Gamma_j)} \sum_{\sigma_{j,k}' \in \V(\Gamma_{j,k}')} \bra{\sigma_j, \sigma_{j,k}'} \msem{\varnothing : \varnothing \partition \Gamma_j, \Gamma_{j,k}' \vdash \pair{e_j}{e_{j,k}'} : T \otimes T'}^\dagger \ket{v, v'}
	\end{alignat*}

	\paragraph{\textsc{S-Flip}}

	\begin{quantikz}
		\lstick{$\Hilb(T)$} & \gate{\spanning{T}{\cdots}} & \qwbundle{\substack{\Hilb(T)^{\oplus (n_1 + n_2 + n_3 + n_4)} \\ \hfill}} & \gate{\mathbb{I} \oplus \textsc{flip} \oplus \mathbb{I}} & \rstick{$\Hilb(T)^{\oplus (n_1 + n_3 + n_2 + n_4)}$} \qw
	\end{quantikz}
	\begin{alignat*}{2}
		&&\;& \msem{\spanning{T}{\cdots}} \\
		&:&& \ket{v} \\
		&\mapsto&& \bigoplus_{j=1}^{n_1} \sum_{\sigma \in \V(\Gamma_{1,j})} \bra{\sigma} \msem{\varnothing : \varnothing \partition \Gamma_{1,j} \vdash e_{1,j} : T}^\dagger \ket{v} \cdot \ket{v} \\
		&&& \oplus \bigoplus_{j=1}^{n_2} \sum_{\sigma \in \V(\Gamma_{2,j})} \bra{\sigma} \msem{\varnothing : \varnothing \partition \Gamma_{2,j} \vdash e_{2,j} : T}^\dagger \ket{v} \cdot \ket{v} \\
		&&& \oplus \bigoplus_{j=1}^{n_3} \sum_{\sigma \in \V(\Gamma_{3,j})} \bra{\sigma} \msem{\varnothing : \varnothing \partition \Gamma_{3,j} \vdash e_{3,j} : T}^\dagger \ket{v} \cdot \ket{v} \\
		&&& \oplus \bigoplus_{j=1}^{n_4} \sum_{\sigma \in \V(\Gamma_{4,j})} \bra{\sigma} \msem{\varnothing : \varnothing \partition \Gamma_{4,j} \vdash e_{4,j} : T}^\dagger \ket{v} \cdot \ket{v} \\
		&\mapsto&& \bigoplus_{j=1}^{n_1} \sum_{\sigma \in \V(\Gamma_{1,j})} \bra{\sigma} \msem{\varnothing : \varnothing \partition \Gamma_{1,j} \vdash e_{1,j} : T}^\dagger \ket{v} \cdot \ket{v} \\
		&&& \oplus \bigoplus_{j=1}^{n_3} \sum_{\sigma \in \V(\Gamma_{3,j})} \bra{\sigma} \msem{\varnothing : \varnothing \partition \Gamma_{3,j} \vdash e_{3,j} : T}^\dagger \ket{v} \cdot \ket{v} \\
		&&& \oplus \bigoplus_{j=1}^{n_2} \sum_{\sigma \in \V(\Gamma_{2,j})} \bra{\sigma} \msem{\varnothing : \varnothing \partition \Gamma_{2,j} \vdash e_{2,j} : T}^\dagger \ket{v} \cdot \ket{v} \\
		&&& \oplus \bigoplus_{j=1}^{n_4} \sum_{\sigma \in \V(\Gamma_{4,j})} \bra{\sigma} \msem{\varnothing : \varnothing \partition \Gamma_{4,j} \vdash e_{4,j} : T}^\dagger \ket{v} \cdot \ket{v} \\
	\end{alignat*}

	\paragraph{Ortho}
	\label{sec:compile-ortho}
	We have now shown that the needed Kraus operator can be implemented whenever $\spanning{T}{e_1, \ldots, e_n}$ is true.
	In the general case, $\ortho{T}{e_1, \ldots, e_n}$ is true instead, where $[e_1, \ldots, e_n]$ is a subsequence of some larger sequence $[e_1', \ldots, e_{n'}']$ such that $\spanning{T}{e_1', \ldots, e_{n'}'}$ is true.
	The circuit for the orthogonallity judgment can be constructed from the circuit for the spanning judgment.
	Given an expression $e_j'$ in the spanning list, define $V_j$ to be $\mathbb{I} \in \linear(\Hilb(T))$ whenever $e_j'$ is in the ortho list, and $0 \in \linear(\Hilb(T), \Hilb(\Void))$ whenever it is not.
	The ortho circuit is then constructed as follows:

	\begin{quantikz}
		\lstick{$\Hilb(T)$} & \gate{\spanning{T}{e_1', \ldots, e_{n'}'}} & \qwbundle{\Hilb(T)^{\oplus n'}} & \qw & \gate{\bigoplus_{j=1}^{n'} V_j} & \rstick{$\Hilb(T)^{\oplus n}$} \qw
	\end{quantikz}
	\begin{alignat*}{2}
		&&\;& \msem{\ortho{T}{e_1, \ldots, e_n}} \\
		&:&& \ket{v} \\
		&\mapsto&& \bigoplus_{j=1}^{n'} \sum_{\sigma \in \V(\Gamma_j)} \bra{\sigma} \msem{\varnothing : \varnothing \partition \Gamma_j \vdash e_j' : T}^\dagger \ket{v} \cdot \ket{v} \in \Hilb(T)^{\oplus n'} \\
		&\mapsto&& \bigoplus_{j=1}^{n'} \sum_{\sigma \in \V(\Gamma_j)} \bra{\sigma} \msem{\varnothing : \varnothing \partition \Gamma_j \vdash e_j' : T}^\dagger \ket{v} \cdot V_j \ket{v} \\
		&=&& \bigoplus_{j=1}^{n} \sum_{\sigma \in \V(\Gamma_j)} \bra{\sigma} \msem{\varnothing : \varnothing \partition \Gamma_j \vdash e_j : T}^\dagger \ket{v} \cdot \ket{v} \in \Hilb(T)^{\oplus n} \\
		&&\;& \msem{\ortho{T}{e_1, \ldots, e_n}}^\dagger \\
		&:&& \inj_j \ket{v} \\
		&\mapsto&& \sum_{\sigma \in \V(\Gamma_j)} \bra{\sigma} \msem{\varnothing : \varnothing \partition \Gamma_j \vdash e_{j'}' : T}^\dagger \ket{v} \cdot \ket{v} \\
		&\mapsto&& \sum_{\sigma \in \V(\Gamma_j)} \bra{\sigma} \msem{\varnothing : \varnothing \partition \Gamma_j \vdash e_j : T}^\dagger \ket{v} \cdot \ket{v} \\
	\end{alignat*}
\end{proof}

From the operators constructed here, it is straightforward to prove Lemmas~\ref{lem:spanning-sem} and \ref{lem:ortho-sem}, by noting that the operators in those lemmas are $\msem{\ortho{T}{e_1, \ldots, e_n}}^\dagger\msem{\ortho{T}{e_1, \ldots, e_n}}$ and \\ $\msem{\spanning{T}{e_1, \ldots, e_n}}^\dagger \msem{\spanning{T}{e_1, \ldots, e_n}}$.
In the spanning case, this construction is norm-preserving, and in the ortho case, it is norm-non-increasing.

\begin{lemma}[erasure compilation]
	\label{lem:erasure}
	Suppose $(\Gamma, \Gamma_j \partition \Delta, \Delta' \vdash e_j' : T')$ for all $j \in \{1, \ldots, n\}$ and $\erases{T'}(x; e_1', \ldots, e_n')$ is true for all $x \in \dom(\Delta)$.
	Then, one can implement a Kraus operator $\msem{\erases{T'}(\Delta; e_1', \ldots, e_n')} : \Hilb(\Delta) \otimes \Hilb(T') \to \Hilb(T')$ with the following behavior for all $\sigma \in \V(\Gamma), \sigma_j \in \V(\Gamma_j), \tau \in \V(\Delta), \tau' \in \V(\Delta')$:
	\[
		\ket{\tau} \otimes \msem{\sigma, \sigma_j : \Gamma, \Gamma_j \partition \Delta, \Delta' \vdash e_j' : T'} \ket{\tau, \tau'}
		\mapsto \msem{\sigma, \sigma_j : \Gamma, \Gamma_j \partition \Delta, \Delta' \vdash e_j' : T'} \ket{\tau, \tau'}
	\]
\end{lemma}
\begin{proof}
	We construct the circuit by recursing on $\Delta$.
	In the base case, $\Delta = \varnothing$ and an identity operator (empty circuit) suffices.
	Thus, we focus on the inductive case where our context is $(x : T_x, \Delta)$, assuming the inductive hypothesis that $\msem{\erases{T'}(\Delta; e_1', \ldots, e_n')}$ is implementable with the behavior:
	\[
		\ket{\tau} \otimes \msem{\sigma, \sigma_j : \Gamma, \Gamma_j \partition x : T_x, \Delta, \Delta' \vdash e_j' : T'} \ket{x \mapsto v, \tau, \tau'}
		\mapsto \msem{\sigma, \sigma_j : \Gamma, \Gamma_j \partition x : T_x, \Delta, \Delta' \vdash e_j' : T'} \ket{x \mapsto v, \tau, \tau'}
	\]
	(Note that $x$ appears here even though this is the inductive hypothesis, which would normally be free of $x$. This is still a valid induction principle; we are effectively inducting on the number of variables that must be erased.)
	The problem is then reduced to implementing a Kraus operator $\msem{\erases{T'}(x; e_1', \ldots, e_n')} : \Hilb(T_x) \otimes \Hilb(T') \to \Hilb(T')$ with the following behavior:
	\[
		\ket{v} \otimes \msem{\sigma, \sigma_j : \Gamma, \Gamma_j \partition x : T_x, \Delta, \Delta' \vdash e_j' : T'} \ket{x \mapsto v, \tau, \tau'}
		\mapsto \msem{\sigma, \sigma_j : \Gamma, \Gamma_j \partition x : T_x, \Delta, \Delta' \vdash e_j' : T'} \ket{x \mapsto v, \tau, \tau'}
	\]
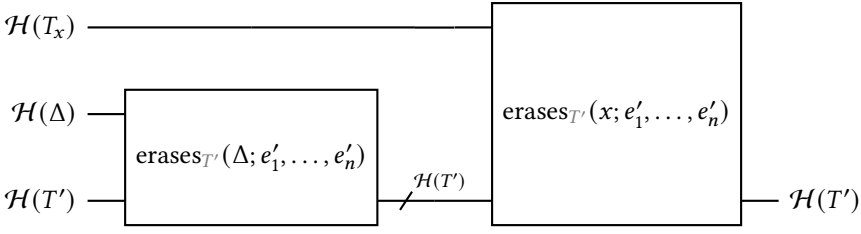
	\begin{quantikz}
		\lstick{$\Hilb(T_x)$} & \qw & \hphantomgate{} & \gate[3,nwires=2]{\erases{T'}(x; e_1', \ldots, e_n')} \\
		\lstick{$\Hilb(\Delta)$} & \gate[2]{\erases{T'}(\Delta; e_1', \ldots, e_n')} & & \\
		\lstick{$\Hilb(T')$} & & \qwbundle{\Hilb(T')} & & \rstick{$\Hilb(T')$} \qw
	\end{quantikz}
	\begin{alignat*}{2}
		&&\;& \msem{\erases{T'}(x : T_x, \Delta; e_1', \ldots, e_n')} \\
		&:&& \ket{x \mapsto v, \tau} \otimes \msem{\sigma, \sigma_j : \Gamma, \Gamma_j \partition x : T_x, \Delta, \Delta' \vdash e_j' : T'} \ket{x \mapsto v, \tau, \tau'} \\
		&\mapsto&& \ket{v} \otimes \msem{\sigma, \sigma_j : \Gamma, \Gamma_j \partition x : T_x, \Delta, \Delta' \vdash e_j' : T'} \ket{x \mapsto v, \tau, \tau'} \\
		&\mapsto&& \msem{\sigma, \sigma_j : \Gamma, \Gamma_j \partition x : T_x, \Delta, \Delta' \vdash e_j' : T'} \ket{x \mapsto v, \tau, \tau'} \\
	\end{alignat*}

	The rest of this proof constructs this gate $\msem{\erases{T'}(x; e_1', \ldots, e_n')}$ by induction on the rule used to prove the erasure judgment.

	\paragraph{\textsc{E-Var}}
	In this case, $e_1' = \cdots = e_n' = x$ and $T' = T_x$.
	We know that $\Delta = \Delta' = \varnothing$ because these contexts must be relevant.

	
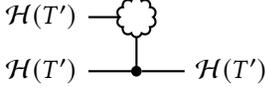
\begin{quantikz}
		\lstick{$\Hilb(T')$} & \gate[style={cloud}]{} \\
		\lstick{$\Hilb(T')$} & \ctrl{-1} & \rstick{$\Hilb(T')$} \qw
	\end{quantikz}
	\begin{alignat*}{2}
		&&\;& \msem{\erases{T'}(x; e_1', \ldots, e_n')} \\
		&:&& \ket{v} \otimes \msem{\sigma, \sigma_j : \Gamma, \Gamma_j \partition x : T' \vdash x : T'} \ket{x \mapsto v} \\
		&=&& \ket{v} \otimes \ket{v} \\
		&\mapsto&& \ket{v} \\
		&=&& \msem{\sigma, \sigma_j : \Gamma, \Gamma_j \partition x : T' \vdash x : T'} \ket{x \mapsto v}
	\end{alignat*}

	\paragraph{\textsc{E-Gphase}}
	In this case, the circuit produced by the inductive hypothesis already has the needed behavior.
	\begin{alignat*}{2}
		&&\;& \msem{\erases{T'}(x; e_1', \ldots, e_{j-1}', e_j' \triangleright \gphase{T}{r}, e_{j+1}', \ldots, e_n')} \\
		&:&& \ket{v} \otimes \msem{\sigma, \sigma_j : \Gamma, \Gamma_j \partition x : T_x, \Delta, \Delta' \vdash e_j' \triangleright \gphase{T}{r} : T'} \ket{x \mapsto v, \tau, \tau'} \\
		&=&& \ket{v} \otimes e^{ir} \msem{\sigma, \sigma_j : \Gamma, \Gamma_j \partition x : T_x, \Delta, \Delta' \vdash e_j' : T'} \ket{x \mapsto v, \tau, \tau'} \\
		&\mapsto&& e^{ir} \msem{\sigma, \sigma_j : \Gamma, \Gamma_j \partition x : T_x, \Delta, \Delta' \vdash e_j' : T'} \ket{x \mapsto v, \tau, \tau'} \\
		&=&& \msem{\sigma, \sigma_j : \Gamma, \Gamma_j \partition x : T_x, \Delta, \Delta' \vdash e_j' \triangleright \gphase{T}{r} : T'} \ket{x \mapsto v, \tau, \tau'}
	\end{alignat*}

	\paragraph{\textsc{E-Ctrl}}
	This rule allows us effectively to ``inline'' the right side of the \texttt{ctrl} expressions for the purpose of the erases judgment.
	Assume one of the expressions is of the following form:
	\[
\cntrl{e}{T}{e_{j,1} &\mapsto e_{j,1}' \\ &\cdots \\ e_{j,m} &\mapsto e_{j,m}'}{T'}
	\]

	Then, we can use the fact that the semantics of \texttt{ctrl} is a linear combination of the semantics of its subexpressions:
	\begin{alignat*}{2}
		&&\;& \msem{\erases{T'}(x; e_1', \ldots, e_{j-1}', \cntrl{e}{T}{\cdots}{T'}, e_{j+1}', \ldots, e_n')} \\
		&:&& \ket{v} \otimes \msem{\sigma, \sigma_j : \Gamma, \Gamma_j \partition x : T_x, \Delta, \Delta' \vdash \cntrl{e}{T}{\cdots}{T'} : T'} \ket{x \mapsto v, \tau, \tau'} \\
		&=&& \sum \cdots \ket{v} \otimes \msem{\sigma, \sigma_{j,k} : \Gamma, \Gamma_{j,k} \partition x : T_x, \Delta, \Delta' \vdash e_{j,k}' : T'} \ket{x \mapsto v, \tau, \tau'} \\
		&\mapsto&& \sum \cdots \msem{\sigma, \sigma_{j,k} : \Gamma, \Gamma_{j,k} \partition x : T_x, \Delta, \Delta' \vdash e_{j,k}' : T'} \ket{x \mapsto v, \tau, \tau'} \\
		&=&& \msem{\sigma, \sigma_j : \Gamma, \Gamma_j \partition x : T_x, \Delta, \Delta' \vdash \cntrl{e}{T}{\cdots}{T'} : T'} \ket{x \mapsto v, \tau, \tau'}
	\end{alignat*}

	\paragraph{\textsc{E-Pair0}}
	\[
	\begin{quantikz}
		\lstick{$\Hilb(T_x)$} & \gate[2]{\erases{T_0}(x; e_{0,1}, \ldots, e_{0,n})} \\
		\lstick{$\Hilb(T_0)$} & & \rstick{$\Hilb(T_0)$} \qw \\
		\lstick{$\Hilb(T_1)$} & \qw & \rstick{$\Hilb(T_1)$} \qw
	\end{quantikz}
\]
	\begin{alignat*}{2}
		&&\;& \msem{\erases{T_0 \otimes T_1}(x; \pair{e_{0,1}}{e_{1,1}}, \ldots, \pair{e_{0,n}}{e_{1,n}})} \\
		&:&& \ket{v} \otimes \msem{\sigma, \sigma_j : \Gamma, \Gamma_j \partition x : T_x, \Delta, \Delta' \vdash \pair{e_{0,j}}{e_{1,j}} : T_0 \otimes T_1} \ket{x \mapsto v, \tau, \tau'} \\
		&=&& \ket{v} \otimes \msem{\sigma, \sigma_j : \Gamma, \Gamma_j \partition x : T_x, \Delta_{*}, \Delta_0, \Delta_1, \Delta_{*}', \Delta_0', \Delta_1' \vdash \pair{e_{0,j}}{e_{1,j}} : T_0 \otimes T_1} \ket{x \mapsto v, \tau_{*}, \tau_0, \tau_1, \tau_{*}', \tau_0', \tau_1'} \\
		&=&& \ket{v} \otimes \msem{\sigma, \sigma_j : \Gamma, \Gamma_j \partition x : T_x, \Delta_{*}, \Delta_0, \Delta_{*}', \Delta_0' \vdash e_{0,j} : T_0} \ket{x \mapsto v, \tau_{*}, \tau_0, \tau_{*}', \tau_0'} \\ &&&\otimes \msem{\sigma, \sigma_j : \Gamma, \Gamma_j \partition x : T_x, \Delta_{*}, \Delta_1, \Delta_{*}', \Delta_1' \vdash e_{1,j} : T_1} \ket{x \mapsto v, \tau_{*}, \tau_1, \tau_{*}', \tau_1'} \\
		&\mapsto&& \msem{\sigma, \sigma_j : \Gamma, \Gamma_j \partition x : T_x, \Delta_{*}, \Delta_0, \Delta_{*}', \Delta_0' \vdash e_{0,j} : T_0} \ket{x \mapsto v, \tau_{*}, \tau_0, \tau_{*}', \tau_0'} \\ &&&\otimes \msem{\sigma, \sigma_j : \Gamma, \Gamma_j \partition x : T_x, \Delta_{*}, \Delta_1, \Delta_{*}', \Delta_1' \vdash e_{1,j} : T_1} \ket{x \mapsto v, \tau_{*}, \tau_1, \tau_{*}', \tau_1'} \\
		&=&& \msem{\sigma, \sigma_j : \Gamma, \Gamma_j \partition x : T_x, \Delta, \Delta' \vdash \pair{e_{0,j}}{e_{1,j}} : T_0 \otimes T_1} \ket{x \mapsto v, \tau, \tau'}
	\end{alignat*}

	\paragraph{\textsc{E-Pair1}}
	\begin{quantikz}
		\lstick{$\Hilb(T_x)$} & \qw & \gate[2]{\erases{T_1}(x; e_{1,1}, \ldots, e_{1,n})} \\
		\lstick{$\Hilb(T_0)$} & \gate[swap]{} & & \gate[swap]{} & \rstick{$\Hilb(T_0)$} \qw \\
		\lstick{$\Hilb(T_1)$} & & \qw & & \rstick{$\Hilb(T_1)$} \qw
	\end{quantikz}
	\begin{alignat*}{2}
		&&\;& \msem{\erases{T_0 \otimes T_1}(x; \pair{e_{0,1}}{e_{1,1}}, \ldots, \pair{e_{0,n}}{e_{1,n}})} \\
		&:&& \ket{v} \otimes \msem{\sigma, \sigma_j : \Gamma, \Gamma_j \partition x : T_x, \Delta, \Delta' \vdash \pair{e_{0,j}}{e_{1,j}} : T_0 \otimes T_1} \ket{x \mapsto v, \tau, \tau'} \\
		&=&& \ket{v} \otimes \msem{\sigma, \sigma_j : \Gamma, \Gamma_j \partition x : T_x, \Delta_{*}, \Delta_0, \Delta_1, \Delta_{*}', \Delta_0', \Delta_1' \vdash \pair{e_{0,j}}{e_{1,j}} : T_0 \otimes T_1} \ket{x \mapsto v, \tau_{*}, \tau_0, \tau_1, \tau_{*}', \tau_0', \tau_1'} \\
		&=&& \ket{v} \otimes \msem{\sigma, \sigma_j : \Gamma, \Gamma_j \partition x : T_x, \Delta_{*}, \Delta_0, \Delta_{*}', \Delta_0' \vdash e_{0,j} : T_0} \ket{x \mapsto v, \tau_{*}, \tau_0, \tau_{*}', \tau_0'} \\ &&&\otimes \msem{\sigma, \sigma_j : \Gamma, \Gamma_j \partition x : T_x, \Delta_{*}, \Delta_1, \Delta_{*}', \Delta_1' \vdash e_{1,j} : T_1} \ket{x \mapsto v, \tau_{*}, \tau_1, \tau_{*}', \tau_1'} \\
		&\mapsto&& \msem{\sigma, \sigma_j : \Gamma, \Gamma_j \partition x : T_x, \Delta_{*}, \Delta_0, \Delta_{*}', \Delta_0' \vdash e_{0,j} : T_0} \ket{x \mapsto v, \tau_{*}, \tau_0, \tau_{*}', \tau_0'} \\ &&&\otimes \msem{\sigma, \sigma_j : \Gamma, \Gamma_j \partition x : T_x, \Delta_{*}, \Delta_1, \Delta_{*}', \Delta_1' \vdash e_{1,j} : T_1} \ket{x \mapsto v, \tau_{*}, \tau_1, \tau_{*}', \tau_1'} \\
		&=&& \msem{\sigma, \sigma_j : \Gamma, \Gamma_j \partition x : T_x, \Delta, \Delta' \vdash \pair{e_{0,j}}{e_{1,j}} : T_0 \otimes T_1} \ket{x \mapsto v, \tau, \tau'}
	\end{alignat*}

	We have thus demonstrated that a circuit with this semantics can always be constructed.
\end{proof}

\fi

\end{document}